\documentclass[11pt]{article}

\usepackage[pagebackref,colorlinks]{hyperref}
\usepackage{amsmath} 
\usepackage{amsthm} 
\usepackage{thmtools}
\usepackage{amssymb}	
\usepackage{graphicx} 
\usepackage{multicol} 
\usepackage{multirow}
\usepackage{color}
\usepackage[dvips,letterpaper,margin=1in,bottom=1in]{geometry}
\usepackage[capitalize,noabbrev]{cleveref}

\usepackage[utf8]{inputenc}
\usepackage[english]{babel}

\usepackage{mathtools}

\newtheorem{theorem}{Theorem}[section]

\newtheorem{lemma}[theorem]{Lemma}

\newtheorem{fact}[theorem]{Fact}

\newtheorem{definition}{Definition}[section]

\newtheorem{remark}{Remark}[section]

\newcommand{\braket}[2]{\left< #1 \vphantom{#2} \middle| #2 \vphantom{#1} \right>} 

\DeclarePairedDelimiter\rbra{\lparen}{\rparen}
\DeclarePairedDelimiter\sbra{\lbrack}{\rbrack}
\DeclarePairedDelimiter\cbra{\{}{\}}
\DeclarePairedDelimiter\abs{\lvert}{\rvert}
\DeclarePairedDelimiter\Abs{\lVert}{\rVert}
\DeclarePairedDelimiter\ceil{\lceil}{\rceil}
\DeclarePairedDelimiter\floor{\lfloor}{\rfloor}
\DeclarePairedDelimiter\ket{\lvert}{\rangle}
\DeclarePairedDelimiter\bra{\langle}{\rvert}
\DeclarePairedDelimiter\ave{\langle}{\rangle}

\newcommand{\set}[2] {\left\{\, #1 \colon #2 \,\right\}}

\newcommand{\tr} {\operatorname{tr}}
\newcommand{\poly} {\operatorname{poly}}
\newcommand{\diag} {\operatorname{diag}}

\newcommand{\rank} {\operatorname{rank}}

\newcommand{\Real} {\operatorname{Re}}

\usepackage{algorithm}
\usepackage{algpseudocode}
\renewcommand{\algorithmicrequire}{\textbf{Input:}} 

\usepackage{tabularx}
\usepackage{booktabs}
\usepackage{threeparttable}
\usepackage{adjustbox}

\newcommand{\footremember}[2]{%
    \footnote{#2}
    \newcounter{#1}
    \setcounter{#1}{\value{footnote}}%
}

\usepackage{tikz}
\usetikzlibrary{quantikz2}

\newcommand\tikzmark[1]{\tikz[overlay,remember picture,baseline] \coordinate (#1);}

\title{Time-Efficient Quantum Entropy Estimator via Samplizer\footnote{An extended abstract of this paper \cite{WZ24b} was presented at the 32nd Annual European Symposium on Algorithms (ESA 2024).}}
\author{
    Qisheng Wang \footremember{1}{Qisheng Wang is with the School of Informatics, University of Edinburgh, EH8 9AB Edinburgh, United Kingdom (e-mail: \href{mailto:QishengWang1994@gmail.com}{\nolinkurl{QishengWang1994@gmail.com}}). Part of the work was done when the author was with the Graduate School of Mathematics, Nagoya University, Nagoya 464-8602, Japan.}
    \and Zhicheng Zhang \footremember{2}{Zhicheng Zhang is with the Centre for Quantum Software and Information, University of Technology Sydney, Ultimo, NSW 2007, Australia (e-mail: \href{mailto:iszczhang@gmail.com}{\nolinkurl{iszczhang@gmail.com}}).}
}
\date{}

\begin{document}

\maketitle

\begin{abstract}
    Entropy is a measure of the randomness of a system. Estimating the entropy of a quantum state is a basic problem in quantum information. 
    In this paper, we introduce a time-efficient quantum approach to estimating the von Neumann entropy $S(\rho)$ and R\'enyi entropy $S_\alpha(\rho)$ of an $N$-dimensional quantum state $\rho$, given access to independent samples of $\rho$. 
    Specifically, we provide the following quantum estimators. 
    \begin{itemize}
        \item A quantum estimator for $S(\rho)$ with time complexity $\widetilde{O}(N^2)$,\footnote{$\widetilde O\rbra{\cdot}$ suppresses polylogarithmic factors.} improving the prior best time complexity $\widetilde{O}(N^6)$ by \hyperlink{cite.AISW20}{Acharya, Issa, Shende, and Wagner (2020)} and \hyperlink{cite.BMW16}{Bavarian, Mehraba, and Wright (2016)}.
        \item A quantum estimator for $S_\alpha(\rho)$ with time complexity $\widetilde{O}(N^{4/\alpha-2})$ for $0<\alpha<1$ and $\widetilde{O}(N^{4-2/\alpha})$ for $\alpha>1$, improving the prior best time complexity $\widetilde{O}(N^{6/\alpha})$ for $0<\alpha<1$ and $\widetilde{O}(N^6)$ for $\alpha>1$ by \hyperlink{cite.AISW20}{Acharya, Issa, Shende, and Wagner (2020)}, though at a cost of a slightly larger sample complexity. 
    \end{itemize}
    Moreover, these estimators are naturally extensible to the low-rank case.
    We also provide a sample lower bound $\Omega(\max\{N/\varepsilon, N^{1/\alpha-1}/\varepsilon^{1/\alpha}\})$ for estimating $S_{\alpha}\rbra{\rho}$.
    
    Technically, our method is quite different from the previous ones that are based on weak Schur sampling and Young diagrams.
    At the heart of our construction, is a novel tool called \textit{samplizer}, which can ``samplize'' a quantum query algorithm to a quantum algorithm with similar behavior using only samples of quantum states; this suggests a new framework for estimating quantum entropies. 
    Specifically, when a quantum oracle $U$ block-encodes a mixed quantum state $\rho$, any quantum query algorithm using $Q$ queries to $U$ can be samplized to a $\delta$-close (in the diamond norm) quantum algorithm using $\widetilde{\Theta}(Q^2/\delta)$ samples of $\rho$.
    Moreover, this samplization is proven to be \textit{optimal}, up to a polylogarithmic factor.
\end{abstract}


\newpage

\tableofcontents
\newpage

\section{Introduction}
Entropy is a basic measure of the randomness of a quantum system in quantum information theory (cf.\ \cite{NC10,Wil13,Hay16,Wat18}), which can be understood as the quantum generalization of the entropy of a probability distribution.
Quantum entropy can be used to quantify important quantum properties, e.g., the compressibility of quantum data \cite{Sch95,JS94,Lo95} and the entanglement of quantum states \cite{HHHH09,Laf16}.
As an analog to the classical learning task of probability distributions, a natural question is: how can we learn the entropy of a quantum state from its independent samples?

Indeed, this is a real question raised in physics for measuring quantum entanglement, e.g., \cite{FIK08,HGKM10,IMP+15}.
Recently, Acharya, Issa, Shende, and Wagner \cite{AISW20} and Bavarian, Mehraba, and Wright \cite{BMW16} proposed sample-efficient quantum algorithms for estimating quantum entropy based on the Empirical Young Diagram (EYD) algorithm \cite{ARS88,KW01}. 
Their algorithms, however, have a large time complexity that is cubic in the sample complexity (i.e., the number of independent samples used in the algorithm), due to the use of weak Schur sampling.\footnote{The quantum algorithms proposed in both \cite{AISW20} and \cite{BMW16} are based on the weak Schur sampling \cite{CHW07} (cf.\ \cite{MdW16}), so (as noted by \cite{Wri22}) they have quantum time complexity $\widetilde O\rbra{n^3}$ on input $n$ independent samples of a quantum state, using the current best quantum Fourier transform over symmetric groups \cite{KS16}. \label{footnote:weak-schur-sampling}}
By stark contrast, classically estimating the entropy of a probability distribution only takes time linear in the sample complexity \cite{VV11,VV17,JVHW15,JVHW17,WY16,AOST17}. 
Regarding these, one may ask:
\[
\textit{Can we estimate quantum entropy with time complexity linear in the sample complexity?}
\]

This is not solely a theoretical question: 
a time-efficient approach to estimating quantum entropy will benefit many practical applications, e.g., preparing quantum Gibbs states \cite{WH19,CLW20,WLW21} and learning Hamiltonians \cite{AAKS21}.

\subsection{Main results}

We introduce a new quantum approach to estimating the entropy of a quantum state, which takes time \textit{linear} in the sample complexity. 
For a quantum algorithm\footnote{In this paper, we only consider \textit{uniform} quantum algorithms. That is, there is a polynomial-time classical Turing machine that, on input $1^n$, outputs the circuit description of the quantum algorithm for the problem of size $n$.} that only takes independent samples of a quantum state as input (this input model is called \textit{sample access}), the sample complexity is the number of samples used in the algorithm, and the time complexity is the sum of the number of one- and two-qubit quantum gates and the number of one-qubit measurements in its circuit description.

We will state the sample and time complexity of our von Neumann entropy estimator and R\'enyi entropy estimator in \cref{sec:von-thm-intro} and \cref{sec:renyi-thm-intro}, respectively. 
In comparison with the additive error $\varepsilon$, we are more interested in the dependence on the size of the input quantum state. 
For simplicity, we assume constant additive error $\varepsilon = \Theta\rbra{1}$ in this section, even though our quantum algorithms are polynomially scalable as $1/\varepsilon$ increases. 

In \cref{tab:entropy-estimator}, we summarize our entropy estimators and compare them with prior best approaches. 
There are also other approaches for estimating the entropy of a quantum state in the literature, which assume access to the quantum circuit that prepares the purification of $\rho$ (this input model is called \textit{purified quantum query access}), thus very different from our setting that only allows access to independent samples of $\rho$. This line of work will be reviewed in \cref{sec:related-work}.

\begin{table} [!htp]
\centering
\caption{Sample and time complexities for entropy estimation of quantum states.}
\label{tab:entropy-estimator}
\adjustbox{max width=\textwidth}{
\begin{threeparttable}
\centering
\begin{tabular}{lllll}
\toprule
& Reference & \multicolumn{1}{c}{$0 < \alpha < 1$} & \multicolumn{1}{c}{$\alpha = 1$ (von Neumann)} & \multicolumn{1}{c}{$\alpha > 1$} \\ \midrule
\multicolumn{1}{l}{\multirow{2}{*}{\begin{tabular}{l}
    \vspace{-3pt}\\
     Upper \\
     Bounds
\end{tabular}}} & \cite{AISW20} & \begin{tabular}{l}
     $O\rbra{N^{2/\alpha}}$ samples \\
     $\widetilde O\rbra{N^{6/\alpha}}$ time
\end{tabular} & \begin{tabular}{l}
     $O\rbra{N^{2}}$ samples \\
     $\widetilde O\rbra{N^{6}}$ time
\end{tabular} & \begin{tabular}{l}
     $O\rbra{N^{2}}$ samples \\
     $\widetilde O\rbra{N^{6}}$ time
\end{tabular} \\ \cmidrule{2-5} 
\multicolumn{1}{l}{}                              & This work   & \begin{tabular}{l}
     $\widetilde O\rbra{N^{4/\alpha-2}}$ samples \\
     $\widetilde O\rbra{N^{4/\alpha-2}}$ time \\
     \cref{thm:sim-renyi}
\end{tabular} & \begin{tabular}{l}
     $\widetilde O\rbra{N^{2}}$ samples \\
     $\widetilde O\rbra{N^{2}}$ time \\
     \cref{thm:simp-von}
\end{tabular} & \begin{tabular}{l}
     $\widetilde O\rbra{N^{4-2/\alpha}}$ samples \\
     $\widetilde O\rbra{N^{4-2/\alpha}}$ time \\
     \cref{thm:sim-renyi}
\end{tabular} \\ \midrule
\multicolumn{1}{l}{\multirow{2}{*}{\begin{tabular}{l}
\vspace{-3pt} \\
     Lower \\
     Bounds 
\end{tabular}}} & \cite{AISW20} & \begin{tabular}{l}
     $\Omega \rbra{N^{1+1/\alpha}}$ {samples}\\
     (EYD)
\end{tabular} & \begin{tabular}{l}
     $\Omega \rbra{N^{2}}$ {samples}\\ (EYD) 
\end{tabular} & \begin{tabular}{l}
     $\Omega \rbra{N^{2}}$ {samples}\\ (EYD) 
\end{tabular} \\ \cmidrule{2-5} 
\multicolumn{1}{l}{}                              & This work   & \begin{tabular}{l}
     $\Omega \rbra{N+N^{1/\alpha - 1}}$ {samples}\\
     \cref{thm:sample-lower-bound-intro}
\end{tabular} & \begin{tabular}{l}
     $\Omega \rbra{N}$ {samples} \\
     \cref{thm:sample-lower-bound-intro}
\end{tabular} & \begin{tabular}{l}
     $\Omega \rbra{N}$ {samples}  \\
     \cref{thm:sample-lower-bound-intro}
\end{tabular} \\ \bottomrule
\end{tabular}
\begin{tablenotes}
\footnotesize
\item (EYD) These lower bounds are for Empirical Young Diagram algorithms.
\end{tablenotes}
\end{threeparttable}
}
\end{table}

\subsubsection{Von Neumann entropy estimator} \label{sec:von-thm-intro}

Our first result is a time-efficient estimator for von Neumann entropy, defined by (cf.\ \cite{vN32})
\[
S\rbra*{\rho} = -\tr\rbra*{\rho \ln \rbra*{\rho}}.
\]

\begin{theorem} [\cref{thm:von-sample} simplified]
\label{thm:simp-von}
    There is a quantum estimator for the von Neumann entropy $S\rbra{\rho}$ of an $N$-dimensional quantum state $\rho$ with sample and time complexity $\widetilde O\rbra{N^2}$. 
\end{theorem}

The prior best quantum estimators for the von Neumann entropy \cite{AISW20,BMW16} have sample complexity $O\rbra{N^2}$ and time complexity $\widetilde O\rbra{N^6}$.\footnote{See \cref{footnote:weak-schur-sampling}.} 
Our estimator is cubicly faster than theirs in the time complexity, while with the same sample complexity (up to a logarithmic factor).
Technically, our method is quite different from the previous ones based on weak Schur sampling and Empirical Young Diagrams.
By comparison, our algorithm builds on our new tool --- \textit{samplizer} (which will be introduced in \cref{sec:intro-samplizer}), together with the block-encoding techniques (cf.\ \cite{GSLW19}). 

Our von Neumann entropy estimator has an advantage in that it can exploit prior knowledge of a relatively low rank $r$ of the quantum state $\rho$. 
In this case, our von Neumann entropy estimator has time complexity $\widetilde O\rbra{r^2}$, which is polynomial in $r$ while only polylogarithmic in $N$.
Note that the work of \cite{AISW20} does not consider the low-rank case.
Recently, a von Neumann entropy estimator was proposed in \cite{WZW22} with sample complexity $\widetilde O\rbra{\kappa^2}$, where $\kappa$ is the reciprocal of the minimum non-zero eigenvalue of $\rho$. 
The rank-dependent version of our algorithm immediately reproduces their result by noting that $\kappa$ is always an upper bound on the rank $r$ of $\rho$.

\subsubsection{R\'enyi entropy estimator} \label{sec:renyi-thm-intro}

We also provide time-efficient estimators for $\alpha$-R\'enyi entropy, defined by (cf.\ \cite{Ren61})
\[
S_\alpha\rbra*{\rho} = \frac{1}{1-\alpha} \ln \rbra*{ \tr\rbra*{\rho^{\alpha}} },
\]
with von Neumann entropy a limiting case: $S\rbra{\rho} = S_{1}\rbra{\rho}$.

\begin{theorem} [\cref{thm:estimate-renyi-gt1} and \cref{thm:estimate-renyi-lt1} simplified]
\label{thm:sim-renyi}
    There is a quantum estimator for the $\alpha$-R\'enyi entropy $S_\alpha\rbra{\rho}$ of an $N$-dimensional quantum state $\rho$ with sample and time complexity $\widetilde O\rbra{N^{4/\alpha-2}}$ for $0 < \alpha < 1$ and $\widetilde O\rbra{N^{4-2/\alpha}}$ for $\alpha > 1$. 
\end{theorem}

The prior best quantum estimators for the $\alpha$-R\'enyi entropy \cite{AISW20} have sample complexity $O\rbra{N^{2/\alpha}}$ and time complexity $\widetilde O\rbra{N^{6/\alpha}}$ for $\alpha < 1$, and sample complexity $O\rbra{N^{2}}$ and time complexity $\widetilde O\rbra{N^{6}}$ for $\alpha > 1$.\footnote{See \cref{footnote:weak-schur-sampling}.} 
By comparison, our estimators for the $\alpha$-R\'enyi entropy and von Neumann entropy are faster (in time) than the approaches of \cite{AISW20} for any constant $\alpha > 0$.\footnote{For integer $\alpha > 1$, the approach of \cite{AISW20} has sample complexity $O\rbra{N^{2-2/\alpha}}$ and time complexity $\widetilde O\rbra{N^{6-6/\alpha}}$. Our algorithm is faster with the only exception that $\alpha = 2$. To address this special case, we provide a simple algorithm for estimating the $2$-R\'{e}nyi entropy $S_2\rbra{\rho}$ with sample and time complexity $\widetilde O\rbra{N^2}$ in \cref{sec:2-renyi}.}
Like our von Neumann entropy estimator, our R\'enyi entropy estimator is also extensible to the low-rank case, resulting in a time complexity polynomial in the rank $r$ of quantum state $\rho$. 

It can be seen that there is a trade-off between the sample and time complexities: our algorithms are more time-efficient, while the approaches of \cite{AISW20} are more sample-efficient. 
The estimators in \cite{AISW20} work in two steps:
\begin{enumerate}
    \item Compute the \textit{empirical} distribution, $\rbra{\hat\lambda_1, \hat\lambda_2, \dots, \hat\lambda_N}$, of the spectrum of the quantum state $\rho$. 
    \item Output the entropy of the quantum state $\sum_{j=1}^N \hat\lambda_j \ket{j}\bra{j}$ as an estimate of the entropy of $\rho$.
\end{enumerate}
This type of quantum algorithm is called the Empirical Young Diagram (EYD) algorithm, which is a quantum analog of the classical empirical/plug-in estimator (also known as the Maximum Likelihood Estimator).
It was shown in \cite{AISW20} that their EYD estimators for von Neumann and R\'enyi entropies are almost sample-optimal over all EYD estimators (but not known to be optimal over all possible estimators; see \cref{sec:lb-intro} for further discussions).
Note that the current best implementation of EYD algorithms has time complexity cubic in its sample complexity due to the use of weak Schur sampling \cite{CHW07,MdW16}.\footnote{See \cref{footnote:weak-schur-sampling}.} 
In sharp contrast, our estimators in \cref{thm:simp-von,thm:sim-renyi} are not EYD.
Consequently, we can estimate these quantum entropies with better time complexity, though at a cost of larger sample complexity. 

\subsection{Techniques}

The design of our quantum algorithms is based on a novel tool --- \textit{samplizer}.
Roughly speaking, the samplizer allows us to design a quantum algorithm with access to samples of quantum states by instead designing a quantum query algorithm (namely, a quantum algorithm with access to a quantum unitary oracle).
We first introduce the samplizer in \cref{sec:intro-samplizer} and then show how to design our quantum entropy estimators using the samplizer in \cref{sec:von-intro} (for von Neumann entropy) and in \cref{sec:renyi-intro} (for R\'enyi entropy).

\subsubsection{Samplizer} \label{sec:intro-samplizer}

Throughout this paper, we use the following concepts and notations for quantum query algorithms and quantum sample algorithms. 
A \textit{quantum query algorithm} with query access to a quantum unitary oracle $U$ is described by a quantum circuit family $C = \cbra{C\sbra{U}}_U$. 
Here, for a fixed unitary operator $U$,
the instance $C[U] \in C$ is a quantum circuit using queries to (controlled-)$U$ and (controlled-)$U^\dag$ (see \cref{def:quantum-circuit-family}). 
A \textit{quantum sample algorithm} with sample access to a mixed quantum state $\rho$ is described by a quantum channel family $\mathcal{E} = \cbra{\mathcal{E}\sbra{\rho}}_\rho$.
Here, for a fixed quantum state $\rho$, the instance $\mathcal{E}\sbra{\rho} \in \mathcal{E}$ is a quantum channel implemented by a unitary operator with the input state of the form $\rho^{\otimes k} \otimes \ket{0}\bra{0}^{\otimes \ell}$ (see \cref{def:quantum-channel-family}). We will use $\mathcal{C}\sbra{U}$ to denote the quantum channel $\mathcal{C}\sbra{U}\rbra{\varrho} = 
C\sbra{U} (\varrho) C\sbra{U}^{\dag}$ induced by $C\sbra{U}$.
In context without ambiguity,
we will simply write $C=\cbra{C\sbra{U}}$ and $\mathcal{E}=\cbra{\mathcal{E}\sbra{\rho}}$ by omitting the subscripts.
To justify the concepts defined here, we note that any quantum entropy estimator using independent samples of quantum states is indeed a quantum sample algorithm. 

Now we are able to introduce the notion of samplizer. 

\begin{definition} [Samplizer]
\label{def:samplizer}
    A samplizer $\mathsf{Samplize}_{*}\ave{*}$ is a converter from a quantum circuit family to a quantum channel family with the following property: for any $\delta > 0$, quantum circuit family $C = \cbra{C\sbra{U}}$, and quantum state $\rho$, there exists a unitary operator $U_\rho$ that is a block-encoding of $\rho/2$ such that\footnote{The scaling factor $1/2$ is due to technical reasons.}
    \begin{equation*}
    \Abs*{\mathsf{Samplize}_{\delta}\ave{C}\sbra{\rho} - \mathcal{C}\sbra{U_\rho}}_{\diamond} \leq \delta,
    \end{equation*}
    where $\Abs{\cdot}_{\diamond}$ denotes the diamond norm between quantum channels.
    Here, $U$ is a block-encoding of $A$ if the matrix $A$ is in the upper left corner in the matrix representation of $U$ (see \cref{def:block-encoding}). 
\end{definition}

\begin{remark}
    Here, we further clarify the relationships among the terminologies introduced above. 

    \begin{itemize}
        \item 
         A quantum circuit family $C = \cbra{C\sbra{U}}$ is a family of quantum circuits with a known structure but using queries to an unknown quantum unitary oracle $U$, where $U$ is considered to be a parameter of the quantum circuit family $C$.
        \item 
        
    A quantum channel family $\mathcal{E} = \cbra{\mathcal{E}\sbra{\rho}}$ is a family of quantum channels with a known structure but using copies of an unknown mixed quantum state $\rho$, where $\rho$ is considered to be a parameter of the quantum channel family $\mathcal{E}$. 

    \item 
    
    For a quantum circuit family $C = \cbra{C\sbra{U}}$ and a precision parameter $\delta > 0$, the notation $\mathsf{Samplize}_\delta\ave{C}$ denotes the quantum channel family $\mathcal{E} = \cbra{\mathcal{E}\sbra{\rho}}$ that is converted from $C$ through the samplizer $\mathsf{Samplize}_*\ave{*}$, satisfying $\Abs{\mathcal{E}\sbra{\rho} - \mathcal{C}\sbra{U_\rho}}_\diamond \leq \delta$ for every quantum state $\rho$ and its associated unitary operator $U_\rho$ (that is a block-encoding of $\rho/2$). 
    \end{itemize}

\end{remark}

The definition of samplizer is inspired by existing quantum query algorithms wherein the output depends only on the matrix block-encoded in the oracle, e.g., the quantum algorithms for Hamiltonian simulation and quantum Gibbs sampling in \cite{GSLW19} and solving systems of linear equations in \cite{CAS+22}.
For any such quantum query algorithm $C$, we can use the samplizer to construct a quantum sample algorithm $\mathsf{Samplize}_{\delta}\ave{C}$ that simulates the behavior of $C$ when the density matrix of a (mixed) quantum state is block-encoded in the oracle. 
The existence of the samplizer will allow us to design quantum sample algorithms by just designing quantum query algorithms instead. 
In the following, we provide an efficient samplizer, demonstrating its existence. 

\begin{theorem} [Optimal samplizer, Theorems \ref{lemma:block-encoding-to-sample} and \ref{thm:optimality-samplizer} informal]
\label{thm:samplizer-intro}
    There is an optimal samplizer $\mathsf{Samplize}_{*}\ave{*}$ such that for any $\delta > 0$ and quantum query algorithm $C$ with query complexity $Q$, the quantum sample algorithm $\mathsf{Samplize}_{\delta}\ave{C}$\footnote{Our samplizer is uniform. That is, there is a polynomial-time deterministic Turing machine that, on input the description of quantum circuit family $C = \cbra{C\sbra{U}}$ and the unary encodings of $Q$ and $\delta$, outputs the quantum circuit description of the implementation of the quantum channel family $\mathsf{Samplize}_{\delta}\ave{C}$.} has sample complexity $\mathsf{S} = \widetilde \Theta\rbra{Q^2/\delta}$ and incurs an extra time complexity of $O\rbra{n\mathsf{S}}$ over $C$ if the quantum oracle of $C$ acts on $n$ qubits. 
\end{theorem}

We design our samplizer based on density matrix exponentiation (also known as the LMR protocol), initially proposed in \cite{LMR14} for quantum principal component analysis and subsequently refined in \cite{KLL+17}, achieving optimal sample and time complexities.\footnote{Recently in \cite{GKP+24}, they fixed an error in the proof for density matrix exponentiation in \cite{KLL+17}. Moreover, they presented a non-asymptotic analysis of the sample complexity of density matrix exponentiation.}
The idea is inspired by the recent quantum algorithms for estimating fidelity \cite{GP22} and trace distance \cite{WZ23}. These algorithms can be employed to construct a quantum circuit that (approximately) block-encodes a quantum state given its independent samples, using quantum singular value transformation \cite{GSLW19}. 
Based on this idea, a lifting theorem was discovered in \cite{WZ23b} that relates quantum sample complexity to quantum query complexity.
In this paper, we further extend this technique to general quantum query algorithms.
This is done by replacing each oracle query with a quantum channel that 
simulates the oracle that is implemented by (samples of) the quantum state block-encoded in the oracle (if applicable).
After the replacement, we obtain a quantum sample algorithm that simulates the original quantum query algorithm.

We prove the optimality of the samplizer by observing that any samplizer can samplize a quantum query algorithm for Hamiltonian simulation (e.g., \cite{GSLW19, LC19}) to a quantum sample algorithm for sample-based Hamiltonian simulation \cite{LMR14}. Then, we use the quantum sample lower bound for the latter problem \cite{KLL+17} to derive a matching lower bound for the samplizer.

\begin{remark}
    Our samplizer studies the sample complexity of simulating quantum query algorithms, which is a generalization of \cite[Theorem 1.1]{WZ23b} (for $Q$-dependence only) and \cite[Corollary 21]{GP22} (for $\delta$-dependence only). 
    Let us make a brief comparison as follows.
    In \cite{WZ23b}, they showed a tight $Q$-dependence but did not consider the dependence on the overall error $\delta$ in the sample/time complexity (they only consider the case when $\delta$ is a constant).
    The $\delta$-dependence is extremely important, 
    as the time complexity of the samplizer grows polynomially in $1/\delta$. 
    In \cite{GP22}, they did not consider the $Q$-dependence and did not show the optimality of the $\delta$-dependence. 
    In \cref{thm:samplizer-intro}, we show matching upper and lower bounds with respect to both $Q$ and $\delta$.

    For the case that $\rho$ is (the density operator of) a pure state, a possible implementation of the samplizer was implied in \cite[Lemma 42 in the full version]{ARU14} with sample complexity $O\rbra{Q^2/\delta^2}$.
    In the recent work \cite{WZ25} after the work described in this paper, the samplizer for pure states was improved to have sample complexity $\Theta\rbra{Q^2/\delta}$, which was shown to be optimal only up to a constant factor through a technique for lower bounds different from ours.
\end{remark}

\subsubsection{How to use samplizer?}
Now we explain how to apply the samplizer from a high-level perspective. 
As an illustrative example, we consider how to estimate the von Neumann entropy, $S\rbra{\rho} = -\tr\rbra{\rho\ln\rbra{\rho}}$, of a mixed quantum state $\rho$ with the help of the samplizer. 
Roughly speaking, consider the quantum circuit shown in \cref{fig:q-von}, where the unitary operator $V$ is a (scaled) block-encoding of $-\ln\rbra{\rho/2}$.\footnote{As $-\ln\rbra{x} \to \infty$ when $x \to 0$, some truncation has to be done in order to ensure the validity of the block-encoding $V$. A rigorous description will be given in \cref{sec:von-intro}.} 
This quantum circuit performs the Hadamard test \cite{AJL09} on the unitary operator $V$ and the quantum state $\rho \otimes \ket{0}\bra{0}^{\otimes a}$, which estimates the value of $\Real\sbra{\tr\rbra{V \rbra{\rho \otimes \ket{0}\bra{0}^{\otimes a}}}}$ that is approximately proportional to $-\tr\rbra{\rho\ln\rbra{\rho/2}} = S\rbra{\rho} - \ln\rbra{2}$. 
It is evident that the von Neumann entropy $S\rbra{\rho}$ can then be estimated by repeated experiments.

\begin{figure} [!htp]
\centering
\begin{quantikz}
    \lstick{$\ket{0}$} & \gate{H} & \ctrl{1} & \gate{H} & \meter{} & \setwiretype{c} \rstick{$x$} \\
    \lstick{$\rho$}  & \setwiretype{b} & \gate[2]{ V \approx \begin{bmatrix}
        -\ln\rbra{\rho/2} & * \\
        * & *
    \end{bmatrix}}   & & \\
    \lstick{$\ket{0}^{\otimes a}$} & \setwiretype{b}  & & & 
\end{quantikz}
\caption{Quantum circuit for von Neumann entropy estimation.}
\label{fig:q-von}
\end{figure}

Suppose that a unitary operator $U_\rho$  is a block-encoding of $\rho/2$ and an approximation polynomial of $-\ln\rbra{x}$ is given. 
It is known that $V$, a (scaled) block-encoding of $-\ln\rbra{\rho/2}$, can be approximately implemented using queries to $U_\rho$ by quantum singular value transformation \cite{GSLW19}. 
To make it clear, let quantum circuit family $C = \cbra{C\sbra{U}}$ denote this implementation such that $V = C\sbra{U_\rho}$. 
Finally, 
by replacing the unitary operator $V$ with the quantum channel $\mathsf{Samplize}_{\delta}\ave{C}\sbra{\rho}$ with a small enough precision parameter $\delta$, we can therefore approximately implement the quantum circuit in \cref{fig:q-von}.
This gives an approach to von Neumann entropy estimation that uses only independent samples of $\rho$. 

For readability, we will give a comprehensive overview of our estimators for von Neumann entropy and R\'enyi entropy in \cref{sec:overview}. 

\subsection{Lower bounds} \label{sec:lb-intro}

For completeness, we prove lower bounds on the sample complexity for estimating the von Neumann entropy and R\'enyi entropy. 

\begin{theorem} [\cref{thm:entropy-estimation-sample-lower-bound} restated] \label{thm:sample-lower-bound-intro}
    For every constant $\alpha > 0$, any quantum estimator for the $\alpha$-R\'enyi entropy of an $N$-dimensional quantum state within additive error $\varepsilon$ requires sample complexity $\Omega\rbra{\max\cbra{N/\varepsilon, N^{1/\alpha-1}/\varepsilon^{1/\alpha}}}$. 
    In particular, estimating the von Neumann entropy ($\alpha = 1$) requires sample complexity $\Omega\rbra{N/\varepsilon}$. 
\end{theorem}

To the best of our knowledge, we are not aware of any general sample lower bounds for estimating the von Neumann entropy or R\'enyi entropy that are explicitly stated in the literature, except for the matching lower bound $\Omega\rbra{\max\cbra{N^{2-2/\alpha}/\varepsilon^{2/\alpha}, N^{1-1/\alpha}/\varepsilon^2}}$ for every constant integer $\alpha \geq 2$ in \cite{AISW20}.
Nevertheless, we note that the sample lower bound for the mixedness testing problem of quantum states given in \cite[Theorem 1.10]{OW21} actually implies an $\Omega\rbra{N}$ sample lower bound for entropy estimation. 
In \cref{thm:sample-lower-bound-intro}, our contribution is that we give a better sample lower bound for $0 < \alpha < 1/2$, and that we further consider the $\varepsilon$-dependence in the lower bounds. 
This is achieved by reducing the task of estimating the $\alpha$-R\'enyi entropy of quantum states to the mixed testing problem of quantum states in \cite{OW21} and to the distinguishing problem of a special probability distribution used in \cite{AOST17,AISW20}.

We note that in \cite{AISW20}, they provided sample lower bounds $\Omega\rbra{\max\cbra{N^2/\varepsilon, N^{1+1/\alpha}/\varepsilon^{1/\alpha}}}$ for any empirical Young diagram algorithms that estimate the $\alpha$-R\'enyi entropy for $\alpha > 0$ (including $\alpha = 1$ for the von Neumann entropy). 
Compared to the lower bounds given in \cref{thm:sample-lower-bound-intro}, their lower bounds do not apply to general algorithms that are not based on empirical Young diagrams (which is noted by \cite{Wri22}).
This is because their lower bounds highly rely on the EYD structure of the estimators, where an empirical distribution should be estimated in the first step (which can require a large sample complexity for a good estimation).
As general estimators do not necessarily estimate the empirical distribution, one may hope that they can perform better than existing EYD estimators in some estimation tasks. 
For example, the sample complexity of von Neumann entropy estimation is still open (see Question \ref{item:ques2} in \cref{sec:discussion}). 

In addition, we discuss the limiting cases $\alpha = 0$ and $\alpha = \infty$ of \cref{thm:sample-lower-bound-intro} as follows.
\begin{itemize}
    \item For the case of $\alpha = 0$, $S_0\rbra{\rho} = \ln\rbra{\rank\rbra{\rho}}$ is the Max (Hartley) entropy. 
    We further show that there is no estimator for the Max entropy (within constant additive error). To see this, consider the problem of distinguishing the two quantum states $\rho_0 = \ket{0}\bra{0}$ and $\rho_{\delta} = \rbra{1-\delta}\ket{0}\bra{0} + \delta \cdot \frac{I}{N}$, where $\delta > 0$ can be arbitrarily close to $0$. 
    Note that $\rank\rbra{\rho_0} = 1$ and $\rank\rbra{\rho_\delta} = N$, and thus $S_0\rbra{\rho_0} = 0$ and $S_0\rbra{\rho_\delta} = \ln\rbra{N}$. 
    On the other hand, according to the upper bound on the success probability of quantum state discrimination (see \cref{thm:HH-measurement}), distinguishing between $\rho_0$ and $\rho_\delta$ requires $\Omega\rbra{1/\delta}$ samples, which can be arbitrarily large and is independent of $N$. 
    \item For the case of $\alpha = \infty$, $S_{\infty}\rbra{\rho} = -\ln\rbra{\Abs{\rho}}$ is the Min entropy. 
    An estimator with sample complexity $O\rbra{N^2/\varepsilon^2}$ is implied by \cite[Theorem 1.18]{OW17}.\footnote{In \cite{OW17}, they proposed a quantum algorithm that finds the top-$k$ eigenvalues of an $N$-dimensional quantum state $\rho$ to $\delta$-accuracy in $\ell_2^2$ distance with sample complexity $O\rbra{k/\delta}$.
    Note that $\Abs{\rho}$ is the largest (i.e., top-$1$) eigenvalue of $\rho$ and $1/N \leq \Abs{\rho} \leq 1$.
    To obtain an estimate of $S_{\infty}\rbra{\rho}$ within additive error $\varepsilon$, an estimate of $\Abs{\rho}$ with multiplicative error $\varepsilon$ suffices. 
    This can be done by taking $k = 1$ and $\delta = \varepsilon^2/N^2$, resulting in a sample complexity of $O\rbra{N^2/\varepsilon^2}$.}
    On the other hand, the proof for $\alpha > 1$ (\cref{thm:sample-lower-bound-renyi-gt1}) also applies to $\alpha = \infty$, thus an $\Omega\rbra{N/\varepsilon}$ sample lower bound also holds for estimating $S_{\infty}\rbra{\rho}$.
\end{itemize}

\subsection{Related work}
\label{sec:related-work}

There are quantum query algorithms for estimating the entropy of a quantum state $\rho$, given purified access to $\rho$.
It was shown in \cite{GL20} that the von Neumann entropy $S\rbra{\rho}$ can be estimated with quantum query complexity $\widetilde O\rbra{N}$. 
The estimation of $S\rbra{\rho}$ was shown to be useful as a subroutine in variational quantum algorithms \cite{CLW20}, where they showed that $S\rbra{\rho}$ can be estimated with query complexity $\widetilde O\rbra{\kappa^2}$, and $\kappa$ is the reciprocal of the minimum non-zero eigenvalue of $\rho$.
A quantum query algorithm for estimating $S\rbra{\rho}$ with multiplicative error was proposed in \cite{GHS21}.
It was shown in \cite{SH21} that the $\alpha$-R\'enyi entropy $S_\alpha\rbra{\rho}$ can be estimated with quantum query complexity $\widetilde O\rbra{\kappa N^{\max\cbra{\alpha, 1}}}$, which was later shown in \cite{LWZ22} to be $\widetilde O\rbra{N^{1/2+1/2\alpha}}$ for $0 < \alpha < 1$ and $\widetilde O\rbra{N^{3/2-1/2\alpha}}$ for $\alpha > 1$.
When $\rho$ is low-rank, it was shown in \cite{WGL+22} that the quantum query complexity of estimating $S\rbra{\rho}$ and $S_\alpha\rbra{\rho}$ is $\poly\rbra{r}$. 
Other than upper bounds, it was shown in \cite{GH20} that estimating the entropy of shallow circuit outputs is hard. 
In addition to quantum approaches, a classical approach for estimating the von Neumann entropy was proposed in \cite{KDS+20}.
For probability distributions, quantum algorithms for estimating their entropy were investigated in \cite{LW19}.

\subsection{Discussion} \label{sec:discussion}

In this paper, we provide time-efficient quantum estimators for the von Neumann entropy and R\'enyi entropy of quantum states using their independent samples. 
They are designed under the unified framework of our novel tool --- samplizer. 
Very different from the prior approaches \cite{AISW20,BMW16} that are based on weak Schur sampling and Young diagrams, our quantum entropy estimators build on the samplizer and quantum singular value transformation, demonstrating that block-encoding techniques \cite{GSLW19} are also useful to obtain efficient quantum estimators that take only independent samples of quantum states as input.

We conclude by mentioning several open questions related to our work. 
\begin{enumerate}
    \item Can we improve the logarithmic factors in the sample complexity of the samplizer given in \cref{thm:samplizer-intro}? 
    The current upper and lower bounds on the sample complexity of the samplizer are only tight up to polylogarithmic factors. \label{item:ques1}
    \item All of the existing estimators for the von Neumann entropy, including the estimators based on the EYD (empirical Young diagram) by \cite{AISW20,BMW16} and ours (\cref{thm:simp-von}), have sample complexity $\widetilde O\rbra{N^2}$. 
    It was also shown in \cite{AISW20} that any quantum EYD estimator for the von Neumann entropy has sample complexity $\Omega\rbra{N^2}$.
    We conjecture that the same sample lower bound also holds for any von Neumann entropy estimator (that is not necessarily based on the EYD), though we can only prove a lower bound $\Omega\rbra{N}$ in \cref{thm:sample-lower-bound-von-neumann}. \label{item:ques2}
    \item Although our R\'enyi entropy estimator (\cref{thm:sim-renyi}) is more time-efficient than the estimator proposed in \cite{AISW20}, its sample complexity is worse. 
    Can we improve the sample-time trade-off or prove any sample-time lower bound for R\'enyi entropy estimators?
    \item The sample/time complexities for estimating the $\alpha$-R\'enyi entropy considered in this paper are only for constant $\alpha$. Can we find efficient estimators for the case of non-constant $\alpha$?
    \item We believe that the samplizer can be useful to design quantum algorithms for quantum property testing, especially for those concerning quantum states. 
    For example, we think that it could be used to simplify the fidelity estimator in \cite{GP22} and the trace distance estimator in \cite{WZ23}. Except for these direct applications, can we find new quantum sample algorithms for other computational tasks of interest through the samplizer? \label{item:ques4}
\end{enumerate}

\subsection{Recent developments}

After the work described in this paper, Hayashi \cite{Hay24} proposed a quantum estimator for von Neumann relative entropy $D\rbra{\rho \| \sigma} \coloneqq -\tr\rbra{\rho \log \sigma} - S\rbra{\rho}$ when $\sigma$ is known, using $O\rbra{N^2}$ samples of $\rho$ and with time complexity $\widetilde O\rbra{N^6}$ based on Schur transforms.\footnote{In \cite[Remarks 7 and 8]{Hay24}, the author discussed the time complexity for the Schur transform \cite{BCH06}, which was known to be $\poly\rbra{n, d, \log\rbra{1/\epsilon}}$ in \cite{BCH07}, where $n$ is the number of identical samples of the input quantum state, $d$ is the dimension of each sample, and $\epsilon$ is the implementation precision.
Later improvements also include \cite{KS18,Kro19}.
For the purpose of von Neumann (relative) entropy estimation, the best choice is the quantum Schur transform with time complexity $\widetilde O\rbra{nd^4}$ due to \cite{Ngu23,GBO23} and the weak Schur sampling with time complexity $\widetilde O\rbra{n^3}$ (see \cref{footnote:weak-schur-sampling}) with $n \coloneqq N^2$ and $d \coloneqq N$, both resulting in the same time complexity $\widetilde O\rbra{N^6}$ up to polylogarithmic factors.}

The samplizer defined in this paper has been applied in many other quantum estimation tasks, thereby addressing Question \ref{item:ques4} raised in \cref{sec:discussion}. 
Liu, Wang, Wilde, and Zhang \cite{LWWZ24} proposed an estimator for the fidelity of well-conditioned quantum states with sample complexity $\widetilde O\rbra{1/\varepsilon^3}$, where $\varepsilon$ is the additive error. 
Liu and Wang \cite{LW25} proposed an estimator for the quantum Tsallis entropy $S_{\alpha}^{T}\rbra{\rho} = \frac{1}{1-\alpha}\rbra{\tr\rbra{\rho^{\alpha}} - 1}$ with sample complexity $\widetilde O\rbra{1/\varepsilon^{3+\frac{2}{\alpha-1}}} = \poly\rbra{1/\varepsilon}$ for any constant $\alpha > 1$, exponentially improving the previous quantum Tsallis entropy estimators given or implied in \cite{AISW20,WGL+22,LWZ22} and this paper.\footnote{Any estimator for the quantum R\'enyi entropy $S_{\alpha}\rbra{\rho}$ is an estimator for the quantum Tsallis entropy $S_{\alpha}^T\rbra{\rho}$ to the same additive error for $\alpha > 1$, but not vice versa.} 
Liu and Wang \cite{LW25b} proposed an estimator for the quantum $\ell_\alpha$ distance $d_{\alpha}\rbra{\rho, \sigma} = \frac{1}{2}\Abs{\rho-\sigma}_{\alpha}$ with sample complexity $\widetilde O\rbra{1/\varepsilon^{3\alpha+2+\frac{2}{\alpha-1}}} = \poly\rbra{1/\varepsilon}$ for any constant $\alpha > 1$, exponentially improving the previous estimator in \cite{WGL+22}, where $\Abs{\cdot}_{\alpha}$ is the Schatten $\alpha$-norm.
Niwa, Rossi, Taranto and Murao \cite{NRTM25} proposed a singular value transformation scheme for (the block-encodings of the Hermitized Liouville representations of) quantum channels.

The notion of samplizer was specialized for pure states in \cite{WZ25} with optimal sample complexity $\Theta\rbra{Q^2/\delta}$, which removes the polylogarithmic factors compared to \cref{thm:samplizer-intro} and thus is a partial answer to Question \ref{item:ques1} raised in \cref{sec:discussion}. 
Using the samplizer for pure states, they showed an estimator for trace distance and square root fidelity of pure states with optimal sample complexity $\Theta\rbra{1/\varepsilon^2}$. 
For comparison, the query complexity of these tasks was recently shown to be $\Theta\rbra{1/\varepsilon}$ in \cite{Wan24} (and later generalized to estimating the fidelity of a mixed state to a pure state in \cite{FW25}). 

In addition, our sample complexity lower bounds for entropy estimation in \cref{thm:sample-lower-bound-intro} was used in \cite{CWZ24} to establish a sample complexity lower bound for estimating the entanglement entropy of bipartite pure quantum states.
Based on this, they further showed a quantum query complexity lower bound for the entanglement entropy problem (which was initiated in \cite{SY23}), improving the prior lower bounds due to \cite{SY23,WZ23b,Weg24}. 

\subsection{Organization}

The organization of the remaining of this paper is as follows.
We will give an overview of our estimators in \cref{sec:overview}, introducing the idea on how to use the samplizer in a convenient way.
We will provide an efficient implementation of the samplizer and show its optimality in \cref{sec:samplizer}.
Quantum estimators for the von Neumann entropy and the $\alpha$-R\'enyi entropy are presented in \cref{sec:von-neumann} and \cref{sec:renyi}, respectively.
Finally, lower bounds on the sample complexity for estimating the von Neumann entropy and R\'enyi entropy are given in \cref{sec:sample-lower-bound-entropy}.

\section{Overview} \label{sec:overview}
In this section, we present an overview of our estimators by providing formal algorithms with the help of the samplizer, which also shows the usefulness of the samplizer. 

\subsection{Preliminaries}

We first introduce the basic notations used in this paper and then the tools necessary to construct our estimators.

\subsubsection{Basic notations}
Throughout this paper, we use $[n]$ to denote the set $\{1,2,\ldots, n\}$ for $n\in \mathbb{N}$.
Let $\mathcal{H}$ be a finite-dimensional Hilbert space. 
A quantum state in $\mathcal{H}$ is represented by a complex-valued vector $\ket{\psi}$. 
The inner product of two quantum states $\ket\psi$ and $\ket\varphi$ is denoted by $\braket{\psi}{\varphi}$. 
The norm of a quantum state is defined by $\Abs{\ket{\psi}} = \sqrt{\braket{\psi}{\psi}}$. 
An $a$-qubit quantum state $\ket{\psi}$ is usually denoted as $\ket{\psi}_a$, with the subscript $a$ indicating the label (and also the number of qubits) of the quantum system; also, we write $\bra{\psi}_a$ for the conjugate of $\ket{\psi}_a$.
For example, $\ket{0}^{\otimes a}$ can be denoted as $\ket{0}_a$ and $\ket{0}^{\otimes a} \otimes \ket{0}^{\otimes b}$ can be denoted as $\ket{0}_a \ket{0}_b$.

A mixed quantum state in $\mathcal{H}$ can be represented by a density operator $\rho$ on $\mathcal{H}$ with $\tr\rbra{\rho} = 1$ and $\rho \sqsupseteq 0$, where $\sqsupseteq$ is the L\"owner order, i.e., $A \sqsupseteq B$ if and only if $A-B$ is positive semidefinite. 
Let $\mathcal{D}\rbra{\mathcal{H}}$ denote the set of all density operators on $\mathcal{H}$. 
The trace distance between two mixed quantum states $\rho$ and $\sigma$ is defined by
\[
\frac 1 2 \Abs{\rho - \sigma}_1,
\]
where $\Abs{A}_1 = \tr\rbra{\sqrt{A^\dag A}}$ and $A^\dag$ is the Hermitian conjugate of $A$. 
A quantum gate can be represented by a unitary operator $U$ satisfying $U^\dag U = UU^\dag = I$, where $I$ is the identity operator. 
The operator norm of an operator $A$ is defined by
\[
\Abs{A} = \sup_{\Abs{\ket{\psi}} = 1} \Abs{A \ket{\psi}}.
\]
A quantum measurement is a collection of operators $M = \cbra{M_m}$ with $\sum_m M_m M_m^\dag = I$. 
If we measure a quantum state $\rho$ using the quantum measurement $M$, then the outcome $m$ will be obtained with probability $p_m = \tr\rbra{M_m \rho M_m^\dag}$ and the quantum state will become $\rho' = M_m \rho M_m^\dag / p_m$ after the measurement. 
A positive operator-valued measure (POVM) is a collection of operators $E = \cbra{E_m}$ with $\sum_m E_m = I$. 
After measuring a quantum state $\rho$ by the POVM $E$, the outcome will be $m$ with probability $\tr\rbra{E_m \rho}$.

For two quantum channels $\mathcal{E}$ and $\mathcal{F}$ that act on $\mathcal{D}\rbra{\mathcal{H}}$, the trace norm distance between them is defined by
\[
\Abs{\mathcal{E} - \mathcal{F}}_{\tr} = \sup_{\varrho \in \mathcal{D}\rbra{\mathcal{H}}} \Abs{ \mathcal{E}\rbra{\varrho} - \mathcal{F}\rbra{\varrho} }_1,
\]
and the diamond norm distance between them is defined by
\[
\Abs{\mathcal{E} - \mathcal{F}}_{\diamond} = \sup_{\varrho \in \mathcal{D}\rbra{\mathcal{H} \otimes \mathcal{H}'}} \Abs{ \rbra*{\mathcal{E} \otimes \mathcal{I}_{\mathcal{H}'}}\rbra{\varrho} - \rbra*{\mathcal{F} \otimes \mathcal{I}_{\mathcal{H}'}}\rbra{\varrho} }_1,
\]
where the supremum is taken over all finite-dimensional Hilbert spaces $\mathcal{H}'$. 
The relationship between the trace norm distance and the diamond norm distance is given as follows. 
\begin{equation}
\label{eq:tr-vs-diamond}
    \frac{1}{\dim\rbra{\mathcal{H}}} \Abs{\mathcal{\mathcal{E} - \mathcal{F}}}_{\diamond} \leq \Abs{\mathcal{\mathcal{E} - \mathcal{F}}}_{\tr} \leq \Abs{\mathcal{\mathcal{E} - \mathcal{F}}}_{\diamond}.
\end{equation}

\subsubsection{Block-encoding}

Many quantum algorithms are described in the language of block-encoding \cite{GSLW19}, which is also adopted in this paper. 

\begin{definition} [Block-encoding]
\label{def:block-encoding}
    Suppose that $A$ is an $n$-qubit operator, $\alpha, \varepsilon \geq 0$ and $a \in \mathbb{N}$. 
    An $\rbra{n+a}$-qubit operator $B$ is an $\rbra{\alpha, a, \varepsilon}$-block-encoding of $A$, if 
    \[
    \Abs*{\alpha \bra{0}_a B \ket{0}_a - A} \leq \varepsilon.
    \]
    Especially when $\alpha = 1$ and $\varepsilon = 0$, we may simply call $B$ a block-encoding of $A$ (if the parameter $a$ is clear or unimportant in context).
\end{definition}

\subsubsection{Quantum Hadamard test}

The Hadamard test \cite{AJL09} is used to estimate the value of $\bra{\psi} U \ket{\psi}$ for given unitary operator $U$ and quantum state $\ket{\psi}$, which can also be generalized to estimate the value of $\tr\rbra{A\rho}$ for block-encoded operator $A$ and mixed quantum state $\rho$.
A version of Hadamard test in \cite{GP22} is presented as follows.

\begin{theorem} [Hadamard test, {\cite[Lemma 9]{GP22}}]
\label{thm:hadamard}
    Suppose $U$ is an $\rbra{n+a}$-qubit unitary operator that is a $\rbra{1, a, 0}$-block-encoding of $A$. 
    Then, we can implement a quantum circuit using $1$ query to $U$ and $O\rbra{1}$ one- and two-qubit quantum gates such that it outputs $1$ with probability $\frac{1 + \operatorname{Re}\rbra{\tr\rbra{A\rho}}}{2}$ (resp. $\frac{1 + \operatorname{Im}\rbra{\tr\rbra{A\rho}}}{2}$) on input $n$-qubit mixed quantum state $\rho$.
\end{theorem}

\subsubsection{Quantum eigenvalue transformation}

We will use the technique of polynomial eigenvalue transformation as a key tool in our algorithms.

\begin{theorem} [Polynomial eigenvalue transformation, {\cite[Theorem 31]{GSLW19}}]
\label{thm:qsvt}
    Suppose that unitary operator $U$ is an $\rbra{\alpha, a, \varepsilon}$-block-encoding of an Hermitian operator $A$. If $\delta > 0$ and $p\rbra{x} \in \mathbb{R}\sbra{x}$ is a polynomial of degree $d$ such that
    \[
    \forall x \in \sbra{-1, 1}, \quad \abs*{p\rbra{x}} \leq \frac 1 2,
    \]
    then there is a unitary operator $\widetilde U$ that is a $\rbra{1, a+2, 4d\sqrt{\varepsilon/\alpha} + \delta}$-block-encoding of $p\rbra{A/\alpha}$, using $O\rbra{d}$ queries to $U$, and $O\rbra{\rbra{a+1}d}$ one- and two-qubit quantum gates. 
    Moreover, the description of $\widetilde U$ can be computed classically in $\poly\rbra{d, \log\rbra{1/\delta}}$ time. 
\end{theorem}

\subsubsection{Polynomial approximations}

To apply polynomial eigenvalue transformation (\cref{thm:qsvt}), we introduce some useful polynomial approximations. 
The following is a polynomial approximation of rectangle functions. 

\begin{lemma} [Polynomial approximation of rectangle function, {\cite[Corollary 16]{GSLW19}}]
\label{lemma:poly-approx-rec}
    Suppose that $\delta, \varepsilon \in \rbra{0, 1/2}$ and $t \in \sbra{-1, 1}$. Then, there is an efficiently computable even polynomial $p\rbra{x} \in \mathbb{R}\sbra{x}$ of degree $O\rbra*{\frac{1}{\delta}\log\rbra*{\frac{1}{\varepsilon}}}$ such that
    \begin{align*}
        & \forall x \in \sbra*{-1, 1}, \quad \abs*{p\rbra*{x}} \leq 1, \\
        & \forall x \in \sbra*{-t+\delta, t-\delta}, \quad p\rbra*{x} \in \sbra*{1-\varepsilon, 1}, \\
        & \forall x \in \sbra*{-1,-t-\delta} \cup \sbra*{t+\delta, 1}, \quad p\rbra*{x} \in \sbra*{0, \varepsilon}.
    \end{align*}
\end{lemma}

The following is a polynomial approximation of negative power functions.

\begin{lemma} [Polynomial approximation of negative power functions, Corollary 67 of the full version of \cite{GSLW19}]
\label{lemma:poly-approx-negative-power}
    Suppose that $\delta, \varepsilon \in \rbra{0, 1/2}$ and $c > 0$. 
    Then, there is an efficiently computable even polynomial $p\rbra{x} \in \mathbb{R}\sbra{x}$ of degree $O\rbra*{\frac{c+1}{\delta}\log\rbra*{\frac{1}{\varepsilon}}}$ such that
    \begin{align*}
        & \forall x \in \sbra*{-1, 1}, \quad \abs*{p\rbra*{x}} \leq 1, \\
        & \forall x \in \sbra*{\delta, 1}, \quad \abs*{p\rbra*{x} - \frac 1 2 \rbra*{\frac{x}{\delta}}^{-c}} \leq \varepsilon.
    \end{align*}
\end{lemma}

Using the above lemmas, now we derive a polynomial approximation of positive power functions for our purpose. This kind of polynomial approximation was ever used in the property testing of quantum states, e.g., \cite[Theorem 6]{WZC+23}, \cite[Lemma 2.13]{WGL+22}, \cite[Corollary 18]{GP22}, and \cite[Lemma 6]{LWZ22}.

\begin{lemma} [Polynomial approximation of positive power functions]
\label{lemma:poly-approx-power}
Suppose that $c > 0$, $0 < \delta < \beta \leq 1/2$ and $0 < \epsilon < 1/2$. Let 
\[
f\rbra*{x} = \frac 1 4 \rbra*{\frac{x}{2\beta}}^{c}.
\]
Then, there is an efficiently computable polynomial $p\rbra*{x} \in \mathbb{R}\sbra*{x}$ of degree
$O\rbra*{\frac{c+1}{\delta} \log\rbra*{\frac{1}{\delta \epsilon}}}$ such that
\begin{align*}
    & \forall x \in \sbra*{0, \delta}, \quad \abs*{p\rbra*{x}} \leq 2f\rbra*{\delta}; \\
    & \forall x \in \sbra*{\delta, \beta}, \quad \abs*{p\rbra*{x} - f\rbra*{x}} \leq \epsilon; \\
    & \forall x \in \sbra*{-1, 1}, \quad \abs*{p\rbra*{x}} \leq \frac 1 2.
\end{align*}
\end{lemma}
\begin{proof}
    Let $\varepsilon = \epsilon \rbra{2\beta}^c \delta^{\ceil{c}-c} \in \rbra{0, 1/2}$.
    By \cref{lemma:poly-approx-negative-power}, there is an even polynomial $q\rbra{x}$ of degree $O\rbra*{\frac{1}{\delta} \log\rbra*{\frac{1}{\varepsilon}}}$ such that
    \begin{align*}
        & \forall x \in \sbra*{-1, 1}, \quad \abs*{q\rbra*{x}} \leq 1, \\
        & \forall x \in \sbra*{\delta, 1}, \quad \abs*{q\rbra*{x} - \frac 1 2 \rbra*{\frac{x}{\delta}}^{c-\ceil*{c}}} \leq \varepsilon.
    \end{align*}
    By \cref{lemma:poly-approx-rec}, there is an even polynomial $r\rbra{x}$ of degree $O\rbra*{\frac{1}{\beta} \log\rbra*{\frac{1}{\varepsilon}}}$ such that
    \begin{align*}
        & \forall x \in \sbra*{-1, 1}, \quad \abs*{r\rbra*{x}} \leq 1, \\
        & \forall x \in \sbra*{-\beta, \beta}, \quad r\rbra*{x} \in \sbra*{1-\varepsilon, 1}, \\
        & \forall x \in \sbra*{-1,-2\beta} \cup \sbra*{2\beta, 1}, \quad r\rbra*{x} \in \sbra*{0, \varepsilon}.
    \end{align*}
    Let
    \[
    p\rbra*{x} = \frac 1 2 \rbra*{2\beta}^{-c} \delta^{c-\ceil*{c}} q\rbra*{x} r\rbra*{x} x^{\ceil*{c}}.
    \]
    We consider the following cases. 
    \begin{itemize}
        \item $x = 0$. We have 
        \[
        p\rbra{0} = f\rbra*{0} = 0 \leq \frac 1 2.
        \]
        \item $x \in (0, \delta]$. Using $\abs*{q\rbra*{x}} \leq 1$ and $\abs*{r\rbra*{x}} \leq 1$, we have
        \[
        \abs*{p\rbra*{x}} \leq \abs*{\frac 1 2 \rbra*{2\beta}^{-c} \delta^{c-\ceil*{c}} x^{\ceil*{c}}} \leq \frac 1 2 \rbra*{2\beta}^{-c} \delta^c = 2f\rbra*{\delta} \leq \frac 1 2.
        \]
        \item $x \in (\delta, \beta]$. Note that
        \[
            \abs*{q\rbra*{x}r\rbra*{x} - \frac 1 2 \rbra*{\frac{x}{\delta}}^{c-\ceil*{c}}} 
            \leq \abs*{q\rbra*{x}r\rbra*{x} - q\rbra*{x}} + \abs*{q\rbra*{x} - \frac 1 2 \rbra*{\frac{x}{\delta}}^{c-\ceil*{c}}} 
            \leq \varepsilon + \varepsilon = 2\varepsilon.
        \]
        Then, we have
        \begin{align*}
            \abs*{ p\rbra*{x} - f\rbra*{x} }
            & = \frac 1 2 \rbra*{2\beta}^{-c} \delta^{c-\ceil*{c}} x^{\ceil*{c}} \abs*{q\rbra*{x}r\rbra*{x} - \frac 1 2 \rbra*{\frac{x}{\delta}}^{c-\ceil*{c}}} \leq \frac 1 2 \rbra*{2\beta}^{-c} \delta^{c-\ceil*{c}} \varepsilon = \frac 1 2 \epsilon \leq \frac 1 2.
        \end{align*}
        \item $x \in (\beta, 2\beta]$. Using $\abs{r\rbra{x}} \leq 1$ and $\abs{q\rbra{x}} \leq \frac 1 2 \rbra*{\frac{x}{\delta}}^{c-\ceil*{c}} + \epsilon$, we have
        \[
        \abs*{p\rbra*{x}} 
        \leq \frac 1 2 \rbra*{2\beta}^{-c} \delta^{c-\ceil*{c}} \rbra*{\frac 1 2 \rbra*{\frac{x}{\delta}}^{c-\ceil*{c}} + \varepsilon} x^{\ceil*{c}} 
        \leq \frac 1 4 + \frac 1 2 \rbra*{2\beta}^{-c} \delta^{c - \ceil*{c}} \varepsilon = \frac 1 4 + \frac 1 2 \epsilon \leq \frac 1 2.
        \]
        \item $x \in (2\beta, 1]$. Using $\abs{q\rbra{x}} \leq 1$ and $\abs{r\rbra{x}} \leq \varepsilon$, we have
        \[
        \abs*{p\rbra*{x}} 
        \leq \abs*{ \frac 1 2 \rbra*{2\beta}^{-c} \delta^{c-\ceil*{c}} r\rbra*{x} x^{\ceil*{c}} }
        \leq \frac 1 2 \rbra*{2\beta}^{-c} \delta^{c-\ceil*{c}} \varepsilon = \frac 1 2 \epsilon \leq \frac 1 2.
        \]
    \end{itemize}
    From the above, we have shown that 
    \begin{align*}
        & \forall x \in \sbra*{0, \delta}, \quad \abs*{p\rbra*{x}} \leq 2f\rbra*{\delta}; \\
        & \forall x \in \sbra*{\delta, \beta}, \quad \abs*{p\rbra*{x} - f\rbra*{x}} \leq \epsilon; \\
        & \forall x \in \sbra*{0, 1}, \quad \abs*{p\rbra*{x}} \leq \frac 1 2.
    \end{align*}
    To complete the proof, we only have to note that $p\rbra{x}$ is either even or odd, and thus $\abs{p\rbra{x}} = \abs{p\rbra{-x}} \leq 1/2$ for $x \in [-1, 0)$. The degree of $p\rbra{x}$ is 
    \[
    \deg \rbra*{p\rbra*{x}} = \ceil*{c} + \deg \rbra*{q\rbra*{x}} + \deg \rbra*{r\rbra*{x}} = O\rbra*{\frac{c+1}{\delta} \log\rbra*{\frac{1}{\delta\epsilon}}}.
    \]
\end{proof}

We also need the following polynomial approximation of logarithms.

\begin{lemma} [Polynomial approximation of logarithms, {\cite[Lemma 11]{GL20}}]
\label{lemma:poly-approx-ln}
    Suppose that $\delta \in (0, 1]$ and $\varepsilon \in (0, 1/2]$. Then, there is an efficiently computable even polynomial $p\rbra{x} \in \mathbb{R}\sbra{x}$ of degree $O\rbra*{\frac{1}{\delta}\log\rbra*{\frac{1}{\varepsilon}}}$ such that
    \begin{align*}
        & \forall x \in \sbra*{-1, 1}, \quad \abs*{p\rbra*{x}} \leq 1, \\
        & \forall x \in \sbra*{\delta, 1}, \quad \abs*{p\rbra{x} - \frac{\ln\rbra{1/x}}{2\ln\rbra{2/\delta}}} \leq \varepsilon.
    \end{align*}
\end{lemma}

\subsection{Von Neumann entropy estimator}
\label{sec:von-intro}

Now we present a quantum estimator for the von Neumann entropy in \cref{algo:von-main-intro} through the samplizer provided in \cref{thm:samplizer-intro}. 
As demonstrated, the samplizer is convenient and useful in designing quantum sample algorithms in a modular fashion.

\begin{algorithm}[!htp]
    \caption{$\texttt{estimate\_von\_Neumann\_main}\rbra{\varepsilon, \delta}$ --- \textit{quantum sample algorithm}}
    \label{algo:von-main-intro}
    \begin{algorithmic}[1]
    \renewcommand{\algorithmicrequire}{\textbf{Resources:}}
    \Require Access to independent samples of $N$-dimensional quantum state $\rho$ of rank $r$.
    \renewcommand{\algorithmicrequire}{\textbf{Input:}}
    \Require $\varepsilon \in \rbra{0, 1}$ and $\delta \in \rbra{0, 1}$. 
    \Ensure $\widetilde S$ such that $\abs{\widetilde S - S\rbra{\rho}} \leq \varepsilon$ with probability $\geq 1 - \delta$. 
    \vspace{1pt}
    \tikzmark{st}\Function {\texttt{von\_Neumann\_subroutine}}{$\delta_p, \varepsilon_p, \delta_Q$} --- \textit{quantum query algorithm}

    \renewcommand{\algorithmicrequire}{\qquad \textbf{Resources:}}
    \Require Unitary oracle $U_A$ that is a block-encoding of $A$.
    \renewcommand{\algorithmicrequire}{\textbf{Input:}}

    \State Let $p\rbra{x}$ be a polynomial of degree $d_p = O\rbra*{\frac{1}{\delta_p}\log\rbra*{\frac{1}{\varepsilon_p}}}$ such that $\abs{p\rbra{x}} \leq \frac 1 2$ for $x \in \sbra{-1, 1}$ and $\abs*{p\rbra{x} - \frac{\ln\rbra{1/x}}{4\ln\rbra{2/\delta_p}}} \leq \varepsilon_p$ for $x \in \sbra{\delta_p, 1}$ (by \cref{lemma:poly-approx-ln}). \label{step:def-polynomial}
    
    \State Construct unitary operator $U_{p\rbra{A}}$ that is a $\rbra{1, a, \delta_Q}$-block-encoding of $p\rbra{A}$, using $O\rbra{d_p}$ queries to $U_A$ (by \cref{thm:qsvt}).

    \State \Return $U_{p\rbra{A}}$.

    \EndFunction\tikzmark{ed}
    
    \vspace{2pt}
    \State \tikzmark{left}$\delta_p \gets \frac{\varepsilon}{128r\ln\rbra*{{32r}/{\varepsilon}}}$, $\varepsilon_p \gets \frac{\varepsilon}{32\ln\rbra{2/\delta_p}}$, $\delta_Q \gets \frac{\varepsilon}{32r\ln\rbra{2/\delta_p}}$, $\delta_a \gets \frac{\varepsilon}{64\ln\rbra{2/\delta_p}}$, $\varepsilon_H \gets \delta_a$, $k \gets \ceil*{\frac{1}{2\varepsilon_H^2} \ln\rbra*{\frac{2}{\delta}}}$.
    \begin{tikzpicture}[remember picture,overlay]
    \draw[draw=black,rounded corners=4pt,fill=black,opacity=0.08]
    ( $(left|-st)+(4pt-18pt,-4pt)$ ) rectangle ( $(left|-ed)+(16.5cm+2pt-18pt,-2pt)$ );
    \end{tikzpicture}
    \For{$i = 1 \dots k$}
    \State Perform the Hadamard test on $\mathsf{Samplize}_{\delta_a}\ave{\tikzmark{st}\texttt{von\_Neumann\_subroutine}\rbra{\delta_p, \varepsilon_p, \delta_Q}\tikzmark{ed}}\sbra{\rho}$ and $\rho$ (by \cref{thm:hadamard}). Let $X_i \in \cbra{0, 1}$ be the outcome. \label{step:hadamard}
    \begin{tikzpicture}[remember picture,overlay]
    \draw[draw=black,rounded corners=2pt,fill=black,opacity=0.08]
    ( $(st)+(0pt,10pt)$ ) rectangle ( $(ed)+(0pt,-4pt)$ );
    \end{tikzpicture}
    
    \EndFor

    \State $\widetilde S \gets 4 \rbra{ 2\sum_{i\in\sbra{k}}X_i/k - 1} \ln\rbra{2/\delta_p} - \ln\rbra{2}$.
    
    \State \Return $\widetilde S$.
    
    \end{algorithmic}
\end{algorithm}

The framework of our quantum estimator for the von Neumann entropy is inspired by the quantum query algorithm in \cite{WGL+22} for estimating the von Neumann entropy. 
In \cref{algo:von-main-intro}, we first design a quantum query algorithm 
\[
\texttt{von\_Neumann\_subroutine}\rbra{\delta_p, \varepsilon_p, \delta_Q}\sbra{U_A} = U_{p\rbra{A}},
\]
which implements a block-encoding $U_{p\rbra{A}}$ of $p\rbra{A}$, using queries to a block-encoding $U_A$ of $A$, where $p\rbra{\cdot}$ is a polynomial defined in Line \ref{step:def-polynomial} of \cref{algo:von-main-intro} that approximates the logarithm (up to some constant factor) in certain regime. 
If $\texttt{von\_Neumann\_subroutine}\rbra{\delta_p, \varepsilon_p, \delta_Q}\sbra{U_\rho}$ can be implemented as desired for every quantum state $\rho$, then we can estimate the von Neumann entropy through the Hadamard test. 
To see this, we provide the following lemma. 

\begin{lemma} [\cref{lemma:von-neumann-block-encoded} informal] \label{lemma:von-intro}
    Suppose that $U_\rho$ is a block-encoding of $\rho/2$ where $\rho$ is a quantum state of rank $r$. 
    Let random variable $X \in \cbra{0, 1}$ be the output of the Hadamard test (as in Line \ref{step:hadamard} of \cref{algo:von-main-intro}) on the unitary operator $\textup{\texttt{von\_Neumann\_subroutine}}\rbra{\delta_p, \varepsilon_p, \delta_Q}\sbra{U_\rho}$ and the quantum state $\rho$. 
    Then,
    \[
    \abs*{ \rbra*{4\rbra*{2\mathbb{E}\sbra{X}-1}\ln\rbra*{\frac{2}{\delta_p}} - \ln\rbra{2}} - S\rbra{\rho}} \leq 4 \rbra*{ 2r\delta_p + \varepsilon_p + r\delta_Q } \ln \rbra*{\frac{2}{\delta_p}}.
    \]
\end{lemma}

Using the samplizer provided in \cref{thm:samplizer-intro}, we are able to construct its ``samplized'' version 
\[
\mathsf{Samplize}_{\delta_a}\ave{\texttt{von\_Neumann\_subroutine}\rbra{\delta_p, \varepsilon_p, \delta_Q}}\sbra{\rho},
\]
which only uses independent samples of the input quantum state $\rho$.
Let $X' \in \cbra{0, 1}$ be the output of the Hadamard test on $\mathsf{Samplize}_{\delta_a}\ave{\texttt{von\_Neumann\_subroutine}\rbra{\delta_p, \varepsilon_p, \delta_Q}}\sbra{\rho}$ and the quantum state $\rho$, as analogous to \cref{lemma:von-intro}. 
It can be shown that $\abs{\mathbb{E}\sbra{X'} - \mathbb{E}\sbra{X}} \leq \delta_a$, which implies that
\[
\abs*{ \rbra*{4\rbra*{2\mathbb{E}\sbra{X'}-1}\ln\rbra*{\frac{2}{\delta_p}} - \ln\rbra{2}} - S\rbra{\rho}} \leq 4 \rbra*{ 2r\delta_p + \varepsilon_p + r\delta_Q + 2\delta_a } \ln \rbra*{\frac{2}{\delta_p}}.
\]
Therefore, once an estimate $p$ of $\mathbb{E}\sbra{X'}$ is obtained, we can use $4\rbra{2p-1}\ln\rbra{2/\delta_p} -\ln\rbra{2}$ as an estimate of $S\rbra{\rho}$.
By choosing appropriate values for the parameters such as $\delta_p, \varepsilon_p, \delta_Q, \delta_a, \varepsilon_H, k$ as in \cref{algo:von-main-intro}, we can obtain an $\varepsilon$-estimate of the von Neumann entropy $S\rbra{\rho}$ with sample and time complexity $\widetilde O\rbra{r^2/\varepsilon^5}$ (see \cref{thm:von-sample}). 

\subsection{R\'enyi entropy estimator} \label{sec:renyi-intro}

We also provide quantum estimators for the $\alpha$-R\'enyi entropy for every $\alpha \in \rbra{0, 1} \cup \rbra{1, +\infty}$ through the samplizer provided in \cref{thm:samplizer-intro}. 
As an illustrative example, we mainly introduce the estimator for $\alpha > 1$ in \cref{algo:renyi-main-intro}.
The idea for $0 < \alpha < 1$ is similar, which is presented in \cref{algo:renyi-lt1-main-intro}. 
The framework of our quantum estimators for the R\'enyi entropy of quantum states is recursive, which is inspired by the quantum query algorithm in \cite{LW19} for estimating the R\'enyi entropy of discrete probability distributions. 
We denote $P_\alpha\rbra{\rho} = \tr\rbra{\rho^\alpha}$. 

\begin{algorithm}[!htp]
    \caption{$\texttt{estimate\_R\'enyi\_gt1\_main}\rbra{\alpha, \varepsilon, \delta}$ --- \textit{quantum sample algorithm}}
    \label{algo:renyi-main-intro}
    \begin{algorithmic}[1]
    \renewcommand{\algorithmicrequire}{\textbf{Resources:}}
    \Require Access to independent samples of $N$-dimensional quantum state $\rho$ of rank $r$.
    \renewcommand{\algorithmicrequire}{\textbf{Input:}}
    \Require $\alpha > 1$, $\varepsilon \in \rbra{0, 1}$, and $\delta \in \rbra{0, 1}$.
    \Ensure $\widetilde S$ such that $\abs{\widetilde S - S_\alpha\rbra{\rho}} \leq \varepsilon$ with probability $\geq 1 - \delta$. 
    \vspace{1pt}
    \tikzmark{st}
    \Function{R\'enyi\_gt1\_subroutine}{$\alpha, P, \delta_p, \varepsilon_p, \delta_Q$} --- \textit{quantum query algorithm}

    \renewcommand{\algorithmicrequire}{\qquad \textbf{Resources:}}
    \Require Unitary oracle $U_A$ that is a block-encoding of $A$.
    \renewcommand{\algorithmicrequire}{\textbf{Input:}}

    \State $\beta \gets \min\cbra{\rbra{10P}^{1/\alpha}, 1/2}$, $c \gets \rbra{\alpha-1}/2$.

    \State Let $p\rbra{x}$ be a polynomial of degree $d_p = O\rbra*{\frac{1}{\delta_p}\log\rbra*{\frac{1}{\delta_p\varepsilon_p}}}$ such that $\abs*{p\rbra{x}} \leq \frac 1 2 \rbra*{\frac{\delta_p}{2\beta}}^{c}$ for $x \in \sbra{0, \delta_p}$, $\abs*{p\rbra{x}-\frac 1 4 \rbra*{\frac{x}{2\beta}}^{c}} \leq \varepsilon_p$ for $x \in \sbra{\delta_p, \beta}$, and $\abs*{p\rbra{x}} \leq \frac 1 2$ for $x \in \sbra{-1, 1}$ (by \cref{lemma:poly-approx-power}). \label{step:def-poly-pos-pow}

    \State Construct unitary operator $U_{p\rbra{A}}$ that is a $\rbra{1, a, \delta_Q}$-block-encoding of $p\rbra{A}$, using $O\rbra{d_p}$ queries to $U_A$ (by \cref{thm:qsvt}). 

    \State \Return $U_{p\rbra{A}}$.

    \EndFunction\tikzmark{ed}
    
    \begin{tikzpicture}[remember picture,overlay]
    \draw[draw=black,rounded corners=4pt,fill=black,opacity=0.08]
    ( $(left|-st)+(4pt-18pt,-4pt)$ ) rectangle ( $(left|-ed)+(16.5cm+2pt-18pt,-3pt)$ );
    \end{tikzpicture}
    \vspace{-12pt}

    \Function{estimate\_R\'enyi\_gt1\_promise}{$\alpha, P, \varepsilon, \delta$}

    \State $\beta \gets \min\cbra{\rbra{10P}^{1/\alpha}, 1/2}$, $m \gets \ceil{8\ln\rbra{1/\delta}}$, $\delta_p \gets \frac{1}{2}\rbra*{\frac{P\varepsilon}{40r}}^{1/\alpha}$.

    \State $\varepsilon_p \gets \frac{\rbra{4\beta}^{1-\alpha}P\varepsilon}{256}$, $\delta_Q \gets \frac{\rbra{4\beta}^{1-\alpha}P\varepsilon}{128r}$, $\delta_a \gets \frac{\rbra{4\beta}^{1-\alpha}P\varepsilon}{128}$, and $k \gets \ceil*{\frac{65536}{\rbra{4\beta}^{1-\alpha}P\varepsilon^2}}$.

    \For {$j = 1 \dots m$}
        \For {$i = 1 \dots k$}
            \State Let $\sigma = \mathsf{Samplize}_{\delta_a}\ave{\tikzmark{st}\texttt{R\'enyi\_gt1\_subroutine}\rbra{\alpha, P, \delta_p, \varepsilon_p, \delta_Q}\tikzmark{ed}}\sbra{\rho}\rbra{\rho \otimes \ket{0}\bra{0}^{\otimes a}}$.
            \State Measure $\sigma$ in the computational basis. 
            \State Let $X_i$ be $1$ if the outcome is $\ket{0}^{\otimes a}$, and $0$ otherwise. \label{step:def-X}
            \begin{tikzpicture}[remember picture,overlay]
            \draw[draw=black,rounded corners=4pt,fill=black,opacity=0.08]
            ( $(st)+(0pt,10pt)$ ) rectangle ( $(ed)+(0pt,-4pt)$ );
            \end{tikzpicture}
        \EndFor
        \State $\hat P_j \gets 16\rbra{4\beta}^{\alpha - 1} \sum_{i \in \sbra{k}} X_i / k$.
    \EndFor

    \State $\widetilde P \gets $ the median of $\hat P_j$ for $j \in \sbra{m}$. 

    \State \Return $\widetilde P$.

    \EndFunction

    \Function{estimate\_R\'enyi\_gt1}{$\alpha, \varepsilon, \delta$}

    \State $\lambda \gets 1 + 1 / \ln \rbra{r}$.
    \If {$\alpha \leq \lambda$}
        \State $P \gets e^{-1}$.
    \Else
        \State $P' \gets \texttt{estimate\_R\'enyi\_gt1}\rbra{\alpha/\lambda, 1/4, \delta/2}$. \label{step:recursive-call}
        \State $P \gets \rbra{4P'/5}^\lambda e^{-1}$.
    \EndIf
    
    \State \Return $\texttt{estimate\_R\'enyi\_gt1\_promise}\rbra{\alpha, P, \varepsilon, \delta/2}$. \label{step:call-promise}

    \EndFunction

    \State $\widetilde P \gets \texttt{estimate\_R\'enyi\_gt1}\rbra{\alpha, \rbra{\alpha-1}\varepsilon/2, \delta}$. \label{step:call-estimate-renyi-gt1}

    \State $\widetilde S \gets \frac{1}{1-\alpha}\ln\rbra{\widetilde P}$.
    
    \State \Return $\widetilde S$.
    
    \end{algorithmic}
\end{algorithm}

\subsubsection{The case of \texorpdfstring{$\alpha > 1$}{α > 1}}
In \cref{algo:renyi-main-intro}, two main functions are explained as follows. 
\begin{itemize}
    \item $\texttt{estimate\_R\'enyi\_gt1}\rbra{\alpha, \varepsilon, \delta}$: return an estimate $\widetilde P$ such that $\rbra{1-\varepsilon} \widetilde P \leq P_\alpha\rbra{\rho} \leq \rbra{1+\varepsilon} \widetilde P$ with probability $\geq 1 - \delta$.
    \item $\texttt{estimate\_R\'enyi\_gt1\_promise}\rbra{\alpha, P, \varepsilon, \delta}$: return an estimate $\widetilde P$ such that $\rbra{1-\varepsilon} \widetilde P \leq P_\alpha\rbra{\rho} \leq \rbra{1+\varepsilon} \widetilde P$ with probability $\geq 1 - \delta$, given a promise that $P \leq P_{\alpha}\rbra{\rho} \leq 10P$.
\end{itemize}
It can be seen that by letting $\widetilde P \gets \texttt{estimate\_R\'enyi\_gt1}\rbra{\alpha, \rbra{\alpha-1}\varepsilon/2, \delta}$ as in Line \ref{step:call-estimate-renyi-gt1} of \cref{algo:renyi-main-intro}, $\widetilde S \gets \frac{1}{1-\alpha}\ln\rbra{\widetilde P}$ is then an $\varepsilon$-estimate of $S_{\alpha}\rbra{\rho}$.

The main observation is that $\texttt{estimate\_R\'enyi\_gt1}\rbra{\alpha, \varepsilon, \delta}$ can be computed recursively. 
This is done by two steps:
\begin{enumerate}
    \item With probability $\geq 1 - \delta/2$, obtain an estimate $P$ such that $P \leq P_{\alpha}\rbra{\rho} \leq 10P$. This is done by reducing to another entropy estimation task with smaller $\alpha$ as in Line \ref{step:recursive-call} of \cref{algo:renyi-main-intro}. 
    \item With probability $\geq 1 - \delta/2$, obtain an estimate $\widetilde P$ such that $\rbra{1-\varepsilon} \widetilde P \leq P_\alpha\rbra{\rho} \leq \rbra{1+\varepsilon} \widetilde P$ by calling $\texttt{estimate\_R\'enyi\_gt1\_promise}\rbra{\alpha, P, \varepsilon, \delta/2}$ as in Line \ref{step:call-promise} of \cref{algo:renyi-main-intro}.
\end{enumerate}
To implement the function $\texttt{estimate\_R\'enyi\_gt1\_promise}\rbra{\alpha, P, \varepsilon, \delta}$, we first design a quantum query algorithm 
\[
\texttt{R\'enyi\_gt1\_subroutine}\rbra{\alpha, P, \delta_p, \varepsilon_p, \delta_Q}\sbra{U_A} = U_{p\rbra{A}},
\]
which implements a block-encoding $U_{p\rbra{A}}$ of $p\rbra{A}$, using queries to a block-encoding $U_A$ of $A$, where $p\rbra{\cdot}$ is a polynomial defined in Line \ref{step:def-poly-pos-pow} of \cref{algo:renyi-main-intro} that approximates the positive power function (up to some constant factor). 
If $\texttt{R\'enyi\_gt1\_subroutine}\rbra{\alpha, P, \delta_p, \varepsilon_p, \delta_Q}\sbra{U_\rho}$ can be implemented as desired for every quantum state $\rho$, then we can estimate $P_{\alpha}\rbra{\rho}$ by applying it on $\rho \otimes \ket{0}\bra{0}^{\otimes a}$. 
To see this, we provide the following lemma. 

\begin{lemma} [\cref{lemma:renyi-large-subroutine} informal] \label{lemma:renyi-intro}
    Suppose that $U_\rho$ is a block-encoding of $\rho/2$ where $\rho$ is a quantum state of rank $r$. 
    Let $U_{p\rbra{\rho/2}} = \textup{\texttt{R\'enyi\_gt1\_subroutine}}\rbra{\alpha, P, \delta_p, \varepsilon_p, \delta_Q}\sbra{U_\rho}$.
    Let random variable $X = 1$ if the measurement outcome of $U_{p\rbra{\rho/2}} \rbra{\rho \otimes \ket{0}\bra{0}^{\otimes a}} U_{p\rbra{\rho/2}}^\dag$
    in the computational basis (on the last $a$ qubits) is $\ket{0}^{\otimes a}$, and $0$ otherwise (as in Line \ref{step:def-X} of \cref{algo:renyi-main-intro}).
    Then,
    \[
    \abs*{ 16 \rbra{4\beta}^{\alpha - 1} \mathbb{E}\sbra{X} - P_{\alpha}\rbra*{\rho} } \leq 5 r \rbra{2\delta_p}^\alpha + 32 \rbra{4\beta}^{\alpha - 1} \rbra*{ \varepsilon_p + r\delta_Q }.
    \]
\end{lemma}

Using the samplizer provided in \cref{thm:samplizer-intro}, we are able to construct its ``samplized'' version 
\[
\mathsf{Samplize}_{\delta_a}\ave{\texttt{R\'enyi\_gt1\_subroutine}\rbra{\alpha, P, \delta_p, \varepsilon_p, \delta_Q}}\sbra{\rho},
\]
which only uses independent samples of the input quantum state $\rho$.
Let random variable $X' = 1$ if the measurement outcome of $\mathsf{Samplize}_{\delta_a}\ave{\texttt{R\'enyi\_gt1\_subroutine}\rbra{\alpha, P, \delta_p, \varepsilon_p, \delta_Q}}\sbra{\rho} \rbra{\rho \otimes \ket{0}\bra{0}^{\otimes a}}$ in the computational basis (on the last $a$ qubits) is $\ket{0}^{\otimes a}$, and $X' = 0$ otherwise, as analogous to \cref{lemma:von-intro}. 
It can be shown that $\abs{\mathbb{E}\sbra{X'} - \mathbb{E}\sbra{X}} \leq \delta_a$, which implies that
\[
\abs*{ 16 \rbra{4\beta}^{\alpha-1} \mathbb{E}\sbra{X'} - P_{\alpha}\rbra*{\rho} } \leq 5 r \rbra{2\delta_p}^\alpha + 16 \rbra{4\beta}^{\alpha-1} \rbra*{ 2\varepsilon_p + 2r\delta_Q + \delta_a }.
\]
Therefore, once an estimate $p$ of $\mathbb{E}\sbra{X'}$ is obtained, we can use $16\rbra{4\beta}^{\alpha-1} p$ as an estimate of $P_{\alpha}\rbra{\rho}$. By choosing appropriate values for the parameters such as $\delta_p, \varepsilon_p, \delta_Q, \delta_a, k$ as in \cref{algo:renyi-main-intro}, we can obtain an $\varepsilon$-estimate of the R\'enyi entropy $S_{\alpha}\rbra{\rho}$ with sample and time complexity $\widetilde O\rbra{r^{4-2/\alpha}/\varepsilon^{3+2/\alpha}}$ (see \cref{thm:estimate-renyi-gt1}). 

\subsubsection{The case of \texorpdfstring{$0 < \alpha < 1$}{0 < α < 1}}

\begin{algorithm}[!htp]
    \caption{$\texttt{estimate\_R\'enyi\_lt1\_main}\rbra{\alpha, \varepsilon, \delta}$ --- \textit{quantum sample algorithm}}
    \label{algo:renyi-lt1-main-intro}
    \begin{algorithmic}[1]
    \renewcommand{\algorithmicrequire}{\textbf{Resources:}}
    \Require Access to independent samples of $N$-dimensional quantum state $\rho$ of rank $r$.
    \renewcommand{\algorithmicrequire}{\textbf{Input:}}
    \Require $0 < \alpha < 1$, $\varepsilon \in \rbra{0, 1}$, and $\delta \in \rbra{0, 1}$.
    \Ensure $\widetilde S$ such that $\abs{\widetilde S - S_\alpha\rbra{\rho}} \leq \varepsilon$ with probability $\geq 1 - \delta$. 
    \vspace{1pt}
    \tikzmark{st}
    \Function{R\'enyi\_lt1\_subroutine}{$\alpha, P, \delta_p, \varepsilon_p, \delta_Q$} --- \textit{quantum query algorithm}

    \renewcommand{\algorithmicrequire}{\qquad \textbf{Resources:}}
    \Require Unitary oracle $U_A$ that is a block-encoding of $A$.
    \renewcommand{\algorithmicrequire}{\textbf{Input:}}

    \State Let $p\rbra{x}$ be a polynomial of degree $d_p = O\rbra*{\frac{1}{\delta_p}\log\rbra*{\frac{1}{\varepsilon_p}}}$ such that $\abs*{p\rbra{x} - \frac 1 4 \rbra*{\frac{x}{\delta_p}}^{\frac{\alpha - 1}{2}}} \leq \varepsilon_p$ for $x \in \sbra{\delta_p, 1}$, and $\abs*{p\rbra{x}} \leq \frac 1 2$ for $x \in \sbra{-1, 1}$ (by \cref{lemma:poly-approx-negative-power}). \label{step:def-poly-neg-pow}

    \State Construct unitary operator $U_{p\rbra{A}}$ that is a $\rbra{1, a, \delta_Q}$-block-encoding of $p\rbra{\rho}$, where $a = m + 2$, using $O\rbra{d_p}$ queries to $U_A$ (by \cref{thm:qsvt}). 

    \State \Return $U_{p\rbra{A}}$.

    \EndFunction\tikzmark{ed}
    
    \begin{tikzpicture}[remember picture,overlay]
    \draw[draw=black,rounded corners=4pt,fill=black,opacity=0.08] ( $(left|-st)+(4pt-18pt,-4pt)$ ) rectangle ( $(left|-ed)+(16.5cm+2pt-18pt,-3pt)$ );
    \end{tikzpicture}
    \vspace{-12pt}

    \Function{estimate\_R\'enyi\_lt1\_promise}{$\alpha, P, \varepsilon, \delta$}

    \State $m \gets \ceil{8\ln\rbra{1/\delta}}$, $\delta_p \gets \frac{1}{2}\rbra*{\frac{P\varepsilon}{40r}}^{1/\alpha}$.

    \State $\varepsilon_p \gets \frac{\rbra{2\delta_p}^{1-\alpha}P\varepsilon}{256}$, $\delta_Q \gets \frac{\rbra{2\delta_p}^{1-\alpha}P\varepsilon}{128r}$, $\delta_a \gets \frac{\rbra{2\delta_p}^{1-\alpha}P\varepsilon}{128}$, and $k \gets \ceil*{\frac{65536}{\rbra{2\delta_p}^{1-\alpha}P\varepsilon^2}}$.

    \For {$j = 1 \dots m$}
        \For {$i = 1 \dots k$}
            \State Let $\sigma = \mathsf{Samplize}_{\delta_a}\ave{\tikzmark{st}\texttt{R\'enyi\_lt1\_subroutine}\rbra{\alpha, P, \delta_p, \varepsilon_p, \delta_Q}\tikzmark{ed}}\sbra{\rho}\rbra{\rho \otimes \ket{0}\bra{0}^{\otimes a}}$.
            \State Let $X_i$ be $1$ if the outcome is $\ket{0}^{\otimes a}$, and $0$ otherwise. \label{step:def-X-lt1}
            \begin{tikzpicture}[remember picture,overlay]
            \draw[draw=black,rounded corners=4pt,fill=black,opacity=0.08] ( $(st)+(0pt,10pt)$ ) rectangle ( $(ed)+(0pt,-4pt)$ );
            \end{tikzpicture}
        \EndFor
        \State $\hat P_j \gets 16 \rbra{2\delta_p}^{\alpha - 1} \sum_{i \in \sbra{k}} X_i/k$.
    \EndFor

    \State $\widetilde P \gets $ the median of $\hat P_j$ for $j \in \sbra{m}$. 

    \State \Return $\widetilde P$.

    \EndFunction

    \Function{estimate\_R\'enyi\_lt1}{$\alpha, \varepsilon, \delta$}

    \State $\lambda \gets 1 - 1 / \ln \rbra{r}$.
    \If {$\alpha \geq \lambda$}
        \State $P \gets 1$.
    \Else
        \State $P' \gets \texttt{estimate\_R\'enyi\_lt1}\rbra{\alpha/\lambda, 1/4, \delta/2}$.
        \State $P \gets \rbra{4P'/5}^\lambda$.
    \EndIf
    
    \State \Return $\texttt{estimate\_R\'enyi\_lt1\_promise}\rbra{\alpha, P, \varepsilon, \delta/2}$.

    \EndFunction

    \State $\widetilde P \gets \texttt{estimate\_R\'enyi\_lt1}\rbra{\alpha, \rbra{1-\alpha}\varepsilon/2, \delta}$.

    \State $\widetilde S \gets \frac{1}{1-\alpha}\ln\rbra{\widetilde P}$.
    
    \State \Return $\widetilde S$.
    
    \end{algorithmic}
\end{algorithm}

Although the structure and the analysis of \cref{algo:renyi-lt1-main-intro} are similar to those of \cref{algo:renyi-main-intro}, we introduce them here for completeness and for noting the differences in detail. 
The two main functions are explained as follows. 
\begin{itemize}
    \item $\texttt{estimate\_R\'enyi\_lt1}\rbra{\alpha, \varepsilon, \delta}$: return an estimate $\widetilde P$ such that $\rbra{1-\varepsilon}P_\alpha\rbra{\rho} \leq \widetilde P \leq \rbra{1+\varepsilon}P_\alpha\rbra{\rho}$ with probability $\geq 1 - \delta$.
    \item $\texttt{estimate\_R\'enyi\_lt1\_promise}\rbra{\alpha, P, \varepsilon, \delta}$: return an estimate $\widetilde P$ such that $\rbra{1-\varepsilon}P_\alpha\rbra{\rho} \leq \widetilde P \leq \rbra{1+\varepsilon}P_\alpha\rbra{\rho}$ with probability $\geq 1 - \delta$, given a promise that $P \leq P_\alpha\rbra{\rho} \leq 10P$.
\end{itemize}
The key part is the implementation of the function $\texttt{estimate\_R\'enyi\_lt1\_promise}\rbra{\alpha, P, \varepsilon, \delta}$. 
To this end, we first design a quantum query algorithm 
\[
\texttt{R\'enyi\_lt1\_subroutine}\rbra{\alpha, P, \delta_p, \varepsilon_p, \delta_Q}\sbra{U_A} = U_{p\rbra{A}},
\]
which implements a block-encoding $U_{p\rbra{A}}$ of $p\rbra{A}$, using queries to a block-encoding $U_A$ of $A$, where $p\rbra{\cdot}$ is a polynomial defined in Line \ref{step:def-poly-neg-pow} of \cref{algo:renyi-lt1-main-intro} that approximates the negative power function (up to some constant factor). 
Similar to the analysis for $\alpha > 1$, if one can implement $\texttt{R\'enyi\_lt1\_subroutine}\rbra{\alpha, P, \delta_p, \varepsilon_p, \delta_Q}\sbra{U_\rho}$ for every quantum state $\rho$, then we can estimate $P_{\alpha}\rbra{\rho}$ by applying it on $\rho \otimes \ket{0}\bra{0}^{\otimes a}$. 
To see this, we provide the following lemma. 

\begin{lemma} [\cref{lemma:renyi-lt1-block-encoded} informal] \label{lemma:renyi-lt1-intro}
    Suppose that $U_\rho$ is a block-encoding of $\rho/2$ where $\rho$ is a quantum state of rank $r$. 
    Let $U_{p\rbra{\rho/2}} = \textup{\texttt{R\'enyi\_lt1\_subroutine}}\rbra{\alpha, P, \delta_p, \varepsilon_p, \delta_Q}\sbra{U_\rho}$.
    Let random variable $X = 1$ if the measurement outcome of $U_{p\rbra{\rho/2}} \rbra{\rho \otimes \ket{0}\bra{0}^{\otimes a}} U_{p\rbra{\rho/2}}^\dag$
    in the computational basis (on the last $a$ qubits) is $\ket{0}^{\otimes a}$, and $0$ otherwise (as in Line \ref{step:def-X-lt1} of \cref{algo:renyi-lt1-main-intro}).
    Then,
    \[
    \abs*{ 16 \rbra{2\delta_p}^{\alpha - 1} \mathbb{E}\sbra{X} - P_{\alpha}\rbra*{\rho} } \leq 5r\rbra{2\delta_p}^{\alpha} + 32\rbra{2\delta_p}^{\alpha-1}\rbra*{\varepsilon_p + r\delta_Q }.
    \]
\end{lemma}

Using the samplizer provided in \cref{thm:samplizer-intro}, we are able to construct its ``samplized'' version 
\[
\mathsf{Samplize}_{\delta_a}\ave{\texttt{R\'enyi\_lt1\_subroutine}\rbra{\alpha, P, \delta_p, \varepsilon_p, \delta_Q}}\sbra{\rho},
\]
which only uses independent samples of the input quantum state $\rho$.
Let random variable $X' = 1$ if the measurement outcome of $\mathsf{Samplize}_{\delta_a}\ave{\texttt{R\'enyi\_lt1\_subroutine}\rbra{\alpha, P, \delta_p, \varepsilon_p, \delta_Q}}\sbra{\rho} \rbra{\rho \otimes \ket{0}\bra{0}^{\otimes a}}$ in the computational basis (on the last $a$ qubits) is $\ket{0}^{\otimes a}$, and $X' = 0$ otherwise, as analogous to \cref{lemma:renyi-lt1-intro}. 
It can be shown that $\abs{\mathbb{E}\sbra{X'} - \mathbb{E}\sbra{X}} \leq \delta_a$, which implies that
\[
\abs*{ 16 \rbra{2\delta_p}^{\alpha-1} \mathbb{E}\sbra{X'} - P_{\alpha}\rbra*{\rho} } \leq 5r\rbra{2\delta_p}^{\alpha} + 16 \delta_p^{\alpha-1} \rbra*{ 2\varepsilon_p + 2r\delta_Q + \delta_a }.
\]
Therefore, once an estimate $p$ of $\mathbb{E}\sbra{X'}$ is obtained, we can use $16\rbra{2\delta_p}^{\alpha-1} p$ as an estimate of $P_{\alpha}\rbra{\rho}$. By choosing appropriate values for the parameters such as $\delta_p, \varepsilon_p, \delta_Q, \delta_a, k$ as in \cref{algo:renyi-lt1-main-intro}, we can obtain an $\varepsilon$-estimate of the R\'enyi entropy $S_{\alpha}\rbra{\rho}$ with sample and time complexity $\widetilde O\rbra{r^{4/\alpha-2}/\varepsilon^{1+4/\alpha}}$ (see \cref{thm:estimate-renyi-lt1}). 

\section{Samplizer}
\label{sec:samplizer}

In this section, we will develop a ``samplizer'' through which every quantum circuit family with block-encoded access to a quantum state $\rho$ can be implemented by a quantum channel family with sample access to $\rho$.
We will first introduce quantum query algorithms in \cref{sec:quantum-query-algorithm} and quantum sample algorithms in \cref{sec:quantum-sample-algorithm}.
A time-efficient samplizer will be given in \cref{sec:efficient-samplizer} and the optimality will be shown in \cref{sec:samplizer-optimality}.

\subsection{Quantum query algorithms} \label{sec:quantum-query-algorithm}

Quantum query is a way to model the resource of an unknown or partially known unitary operator.
As the concept of block-encoding was shown to have many applications in quantum computing (cf.\ \cite{GSLW19}),
in this paper, we consider the quantum query model and especially its special case called \textit{block-encoded access}, where matrices are block-encoded in quantum unitary oracles.

\subsubsection{Quantum circuit family}

In the following, we define the quantum query model that is commonly used to describe quantum query algorithms in terms of quantum circuit families.

\begin{definition} [Quantum circuit family]
\label{def:quantum-circuit-family}
    A quantum circuit family $C = \cbra{C\sbra{U}}$ with query access to a quantum unitary oracle $U$ with query complexity $Q$ consists of quantum circuits of the form
    \[
    C\sbra{U} = G_Q U_{Q} \dots G_{2} U_{2} G_1 U_{1} G_0,
    \]
    where 
    \begin{enumerate}
        \item $U_i$ is either (controlled-)$U$ or (controlled-)$U^\dag$;
        \item $G_i$ consists of one- and two-qubit quantum gates that do not depend on $U$.
    \end{enumerate}
    The (additional) time complexity of $C$ is the number of one- and two-qubit quantum gates in the circuit description of $C\sbra{U}$. 
\end{definition}

Intuitively, a quantum circuit with query access to a quantum unitary oracle is visualized in \cref{fig:def-block-encoded-sample-access}.
\cref{def:quantum-circuit-family} defines the standard quantum query input model. See \cref{sec:query-ex1,sec:query-ex2,sec:query-ex3} for examples.
The time complexity of quantum algorithms in this model is usually defined to be the sum of the number of queries to the oracle and the number of additional one- and two-qubit gates.
In \cref{def:quantum-circuit-family}, we consider the query complexity separately to simplify the statement of \cref{lemma:block-encoding-to-sample}.

\begin{figure} [!htp]
\centering
\begin{quantikz}
    \setwiretype{b}
    & \gate{G_0} \qw & \gate{U_1} & \gate{G_1} & \gate{U_2} & \gate{G_2} & \qw \midstick[2,brackets=none]{\(\cdots\)} & \gate{U_Q} & \gate{G_Q} & \qw \\
\end{quantikz}
\caption{Quantum circuit family with query access to a quantum unitary oracle.}
\label{fig:def-block-encoded-sample-access}
\end{figure}

\subsubsection{Block-encoded access} \label{sec:query-ex1}

Block-encoded access is a special type of quantum query model where the quantum oracle is a block-encoding of a matrix of interest. 

\begin{definition} [Quantum circuit family with block-encoded access]
\label{def:block-encoded-sample-access}
    Suppose $A$ is an $n$-qubit linear quantum operator with $\Abs{A}\leq 1$, and $m \geq 1$.
    A quantum circuit family $C = \cbra{C\sbra{U}}$ with $m$-ancilla block-encoded access to $A$ is a quantum circuit family with quantum query oracle $U$, where $U$ ranges over all $\rbra{n+m}$-qubit unitary operators that are $\rbra{1,m,0}$-block-encodings of $A$. 
\end{definition}

In the following, we will discuss the relationship between block-encoded access and some other quantum input models. 

\subsubsection{Query access to classical data} \label{sec:query-ex2}
A quantum oracle with access to a matrix $A$ is usually defined as 
\[
\mathcal{O}_{A} \colon \ket{i} \ket{j} \ket{0} \mapsto \ket{i} \ket{j} \ket{A_{i,j}}.
\]
This input model is widely considered in quantum algorithms, e.g., solving systems of linear equations \cite{HHL09} and Hamiltonian simulation \cite{LC19}. 
It was pointed out in \cite{GSLW19} that such query access to classical data can be efficiently converted to the block-encoded access, especially when $A$ is sparse.

\subsubsection{Purified access to quantum states} \label{sec:query-ex3}
Purified access assumes
a quantum oracle (circuit) 
\[
\mathcal{O}_\rho \colon \ket{0} \mapsto \ket{\rho}
\]
that prepares a purification of a mixed quantum state $\rho$ of interest in quantum complexity theory \cite{Wat02} and quantum property testing \cite{GL20}.
In many quantum algorithms, e.g., \cite{GL20,BKL+19,GHS21,WZC+23}, they used a unitary operator that is a block-encoding of the prepared quantum state, by the technique of purified density matrix \cite{LC19} (see also {\cite[Lemma 25]{GSLW19}}).

\subsection{Quantum sample algorithms} \label{sec:quantum-sample-algorithm}

In the following, we introduce the notion of quantum channel families with sample access to describe quantum sample algorithms.  

\begin{definition} [Quantum channel family with sample access]
\label{def:quantum-channel-family}
    A quantum channel family $\mathcal{E} = \cbra{\mathcal{E}\sbra{\rho}}$ with sample access to mixed quantum state $\rho$ with sample complexity $k$ is implemented by a unitary operator $W$ such that
    \[
    \mathcal{E}\sbra{\rho}\rbra{\varrho} = \tr_S \rbra*{ W \rbra*{ \underbrace{\rho^{\otimes k} \otimes \ket{0}\bra{0}^{\otimes \ell}}_{S} \otimes \varrho } W^\dag }.
    \]
    The time complexity of (the implementation of) $\mathcal{E}$ is the number of one- and two-qubit gates in the circuit description of $W$.
\end{definition}

The implementation of a quantum channel family with sample access to a mixed quantum state is visualized in \cref{fig:def-sample-access}. 
It can be seen that any quantum learning algorithm that takes independent samples of quantum states as input can be described by a quantum channel family with sample access defined by \cref{def:quantum-channel-family}.

\begin{figure} [!htp]
\centering
\begin{quantikz}
    \lstick[2]{$S$} \setwiretype{b} \rho^{\otimes k}\phantom{xx}          &[1em] \gate[3]{W}\qw   & \meterD{} \\
    \setwiretype{b} \ket{0}^{\otimes \ell}\phantom{x}      &[1em]  \qw             & \meterD{} \\
    \setwiretype{b} \varrho \phantom{xxxx}                   & [1em]   \qw        &\rstick{$\mathcal{E}\sbra{\rho}(\varrho)$}
\end{quantikz}
\caption{Quantum circuit with sample access to a mixed quantum state.}
\label{fig:def-sample-access}
\end{figure}

\subsection{An efficient samplizer} \label{sec:efficient-samplizer}

We provide an efficient construction of the samplizer.

\begin{theorem} [Samplizer] \label{lemma:block-encoding-to-sample}
    Suppose $C = \cbra{C\sbra{U}}$ is a quantum circuit family with $m$-ancilla block-encoded access to $\rho/2$ (see \cref{def:block-encoded-sample-access}) with query complexity $Q$, where $\rho$ ranges over $n$-qubit mixed quantum states. If $m \geq 4$, then, for every $\delta > 0$, there is a quantum channel family $\mathsf{Samplize}_\delta\ave{C}$ with sample access to $\rho$ (see \cref{def:quantum-channel-family}) with sample complexity $O\rbra*{\frac{Q^2}{\delta}\log^2\rbra*{\frac{Q}{\delta}}}$ 
    satisfying:
    for every $\rho$, there is a specific unitary operator $U_\rho$ that is a $\rbra{2, m, 0}$-block-encoding of $\rho$ such that 
    \[
    \Abs*{ \mathsf{Samplize}_{\delta}\ave*{C}\sbra{\rho} - \mathcal{C}\sbra*{U_\rho} }_{\diamond} \leq \delta.
    \]
    Moreover, if the time complexity of $C$ is $T$, then the time complexity of 
    $\mathsf{Samplize}_{\delta}\ave*{C}$ is 
    \[
    O\rbra*{T+\frac{Q^2}{\delta}n\log^2\rbra*{\frac{Q}{\delta}}}.
    \]
\end{theorem}

To this end, we need the quantum algorithmic tool called
density matrix exponentiation. 
Given only independent samples of a quantum state $\rho$, we can implement the unitary operator $e^{-i\rho t}$ by the method given in \cite{LMR14,KLL+17}, which is called density operator exponentiation or sample-based Hamiltonian simulation. 
Here, we use the version in \cite{GKP+24}. 

\begin{theorem} [Density matrix exponentiation, {\cite[Corollary 3]{GKP+24}}]
\label{thm:sample-based-hamiltonian-simulation}
    Suppose that $\rho$ is a mixed quantum state.
    For every $0 < \delta < 1$ and $t \geq 0$, it is sufficient to use $O\rbra{t^2/\delta}$ samples of $\rho$ to implement a quantum channel $\mathcal{E}$ such that
    \[
    \Abs*{\mathcal{E} - e^{-i\rho t}}_{\diamond} \leq \delta.
    \]
\end{theorem}

\begin{remark}
    For the case that $\rho = \ket{\psi}\bra{\psi}$ is (the density operator of) a pure state $\ket{\psi}$, \cref{thm:sample-based-hamiltonian-simulation} implies an approach for approximately implementing the reflection operator $R_{\psi} = I - 2\ket{\psi}\bra{\psi} = e^{i\ket{\psi}\bra{\psi}\pi}$ by using $\Theta\rbra{1/\delta}$ samples of the pure state $\ket{\psi}$, which was shown to be optimal in \cite[Theorem 4]{KLL+17}. 
    A different approach for implementing reflection operators from samples (i.e., the LMR protocol for pure states) was given in \cite{HL11}.
    Other approaches for implementing reflection operators were rediscovered in the literature, e.g., \cite[Lemma 42 in the full version]{ARU14} and \cite[Lemma 2]{Qia24}. 
    Recently, the optimal error (in the diamond norm distance) for approximately implementing the reflection operator $R_{\psi}$ was shown in \cite{SGSV24} in terms of the number of samples of the pure state $\ket{\psi}$. 
\end{remark}

Using density operator exponentiation, it is noted in \cite[Corollary 21]{GP22} how to implement a unitary operator that is a block-encoding of a given mixed quantum state. 
Here, we use a refined version given in \cite{WZ23b}, which will be used as a key tool for constructing the samplizer in \cref{sec:samplizer}. 

\begin{lemma} [{\cite[Lemma 2.21]{WZ23b}}] \label{lemma:block-encoding-of-quantum-states}
    For every $\delta \in \rbra{0, 1}$ and given access to independent samples of quantum states $\rho$, we can implement two quantum channels $\mathcal{E}$ and $\mathcal{E}^{\textup{inv}}$ using $O\rbra*{\frac{1}{\delta}\log^2\rbra*{\frac{1}{\delta}}}$ samples of $\rho$ such that $\Abs{\mathcal{E} - \mathcal{U}}_{\diamond} \leq \delta$ and $\Abs{\mathcal{E}^{\textup{inv}} - \mathcal{U}^{\textup{inv}}}_{\diamond} \leq \delta$, where $\mathcal{U} \colon \varrho \mapsto U\varrho U^\dag$, $\mathcal{U}^{\textup{inv}} \colon \varrho \mapsto U^\dag \varrho U$, and $U$ is a $\rbra{2, 4, 0}$-block-encoding of $\rho$. 
    Moreover, if $\rho$ is an $n$-qubit quantum state, then $\mathcal{E}$ and $\mathcal{E}^{\textup{inv}}$ use $O\rbra*{\frac{n}{\delta}\log^2\rbra*{\frac{1}{\delta}}}$ one- and two-qubit quantum states. 
\end{lemma}

Now we are ready to prove \cref{lemma:block-encoding-to-sample}.

\begin{proof}[Proof of \cref{lemma:block-encoding-to-sample}]
    The construction of the samplizer generalizes the proof of the quantum sample-to-query lifting theorem in \cite{WZ23b}.
    Suppose the quantum circuit $C\sbra{U}$ is described by 
    \[
    C\sbra{U} = G_Q \cdot U_{Q} \cdot \dots \cdot G_{2} \cdot U_{2} \cdot G_1 \cdot U_{1} \cdot G_0,
    \]
    where each $G_i$ consists of one- and two-qubit quantum gates and each $U_i$ is either (controlled-)$U$ or (controlled-)$U^\dag$. 

    Let $\varepsilon = \delta/Q$. 
    By \cref{lemma:block-encoding-of-quantum-states}, for every $n$-qubit quantum state $\rho$, there is an $\rbra{n+4}$-qubit unitary operator $U_\rho$ and we can implement two quantum channels $\mathcal{E}_\rho$ and $\mathcal{E}_\rho^{\textup{inv}}$ such that $\Abs{\mathcal{E}_\rho - \mathcal{U}_\rho}_{\diamond} \leq \varepsilon$ and $\Abs{\mathcal{E}_\rho^{\textup{inv}} - \mathcal{U}_\rho^{\textup{inv}}}_{\diamond} \leq \varepsilon$, where $\mathcal{U}_\rho \colon \varrho \mapsto U_\rho\varrho U_\rho^\dag$, $\mathcal{U}_\rho^{\textup{inv}} \colon \varrho \mapsto U_\rho^\dag \varrho U_\rho$, and $U_\rho$ is a $\rbra{2, 4, 0}$-block-encoding of $\rho$. 
    Moreover, each of $\mathcal{E}$ and $\mathcal{E}^{\textup{inv}}$ uses $O\rbra*{\frac{1}{\varepsilon}\log^2\rbra*{\frac{1}{\varepsilon}}}$ samples of $\rho$ and $O\rbra*{\frac{n}{\varepsilon}\log^2\rbra*{\frac{1}{\varepsilon}}}$ one- and two-qubit quantum gates. 

    Without loss of generality, we assume that $Q \geq 1$, $m \geq 4$ and $\delta \in \rbra{0, 1}$. 
    Now we construct a modified quantum circuit $\mathcal{C}'\sbra{\rho}$ (strictly speaking, $\mathcal{C}'$ is a quantum channel family) by replacing all the uses of (controlled-)$U_\rho$ and (controlled-)$U_\rho^\dag$ in the quantum circuit $C\sbra{U_\rho \otimes I^{\otimes \rbra{m-4}}}$ by the implementations of (controlled-)$\mathcal{E}_{\rho}$ and (controlled-)$\mathcal{E}_{\rho}^{\text{inv}}$, respectively.
    Then, we obtain a new quantum circuit (see \cref{fig:samplized-block-encoded-sample-access})
    \[
    \mathcal{C}'\sbra{\rho} = G_Q \circ \mathcal{E}_Q \circ \dots \circ G_2 \circ \mathcal{E}_2 \circ G_1 \circ \mathcal{E}_1 \circ G_0,
    \]
    where $\mathcal{E}_i$ is (controlled-)$\mathcal{E}_{\rho}$ if $U_i$ is (controlled-)$U$, and $\mathcal{E}_i$ is (controlled-)$\mathcal{E}_{\rho}^{\text{inv}}$ if $U_i$ is (controlled-)$U^\dag$.
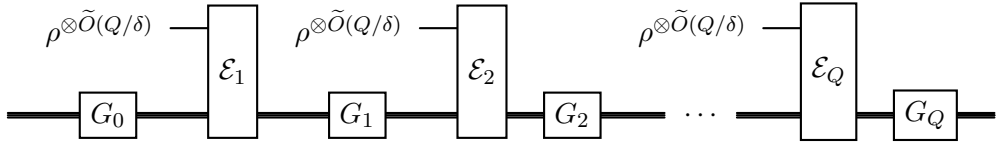
\begin{figure} [!htp]
\centering
\begin{quantikz}
    &  \wireoverride{n} \rho^{\otimes \widetilde O\rbra{Q/\delta}}\phantom{x} & \gate[2]{\mathcal{E}_1}  & \wireoverride{n}\rho^{\otimes \widetilde O\rbra{Q/\delta}}\phantom{x} & \gate[2]{\mathcal{E}_2} &\wireoverride{n} & \wireoverride{n}\rho^{\otimes \widetilde O\rbra{Q/\delta}}\phantom{x} & \gate[2]{\mathcal{E}_Q} & \wireoverride{n} \\
    \setwiretype{b} & \gate{G_0} &  &   \gate{G_1} &  &  \gate{G_2} & \midstick[2,brackets=none]{\(\cdots\)} &  &  \gate{G_Q} & \qw \\
\end{quantikz}
\caption{``Samplized'' quantum circuit for block-encoded access.}
\label{fig:samplized-block-encoded-sample-access}
\end{figure}

    It can be shown that $\mathcal{C}'\sbra{\rho}$ is a valid implementation of $\mathsf{Samplize}_{\delta}\ave*{C}\sbra{\rho}$. 
    To see this, we first note that
    \[
    \Abs*{\mathcal{C}'\sbra{\rho} - \mathcal{C}\sbra{U_\rho \otimes I^{\otimes \rbra{m-4}}}}_{\diamond} \leq Q\varepsilon = \delta.
    \]
    We also note that $U_\rho \otimes I^{\otimes \rbra{m-4}}$ is a $\rbra{2, m, 0}$-block-encoding of $\rho$. 

    Now we are going to analyze the complexity. There are $Q$ queries to $\mathcal{E}_\rho$ in $\mathcal{C}'\sbra{\rho}$, where each $\mathcal{E}_\rho$ uses $s = O\rbra*{\frac{1}{\varepsilon}\log^2\rbra*{\frac{1}{\varepsilon}}}$ samples of $\rho$. Therefore, the total number of used samples of $\rho$ is
    \[
    Q s = O\rbra*{\frac{Q^2}{\delta}\log^2\rbra*{\frac{Q}{\delta}}}.
    \]
    For the time complexity, we note that each $\mathcal{E}_\rho$ uses $t = O\rbra*{\frac{n}{\varepsilon}\log^2\rbra*{\frac{1}{\varepsilon}}}$ one- and two-qubit quantum gates. 
    Therefore, the total number of one- and two-qubit quantum gates is
    \[
    T + Qt = O\rbra*{T+\frac{Q^2}{\delta}n\log^2\rbra*{\frac{Q}{\delta}}}.
    \]
\end{proof}

\begin{remark}
    The requirement $m \geq 4$ in \cref{lemma:block-encoding-to-sample} is due to technical reasons.
    Nevertheless, in most cases this is not a problem, 
    given a $(2,m,0)$-block-encoding $U$ of $\rho$,
    we can always introduce redundant ancilla qubits and use $U \otimes I_4$ as a $\rbra{2, m+4, 0}$-block-encoding of $\rho$ in the design of quantum algorithms in order to fit the requirement.
\end{remark}

\begin{remark}
    Compared to the quantum sample-to-query lifting theorem in \cite{WZ23b} that focuses on the lower bound for quantum property testing, \cref{lemma:block-encoding-to-sample} extends their idea to a ``samplizer'' with the time efficiency considered, enabling us to use it as an algorithmic tool to obtain upper bounds on quantum sample complexity as well as quantum time complexity. 
    Notably, \cref{lemma:block-encoding-to-sample} reveals the dependence on the diamond norm distance $\delta$, which is an important parameter when designing time-efficient quantum algorithms. 
    In addition, we show that the $\delta$-dependence is optimal up to polylogarithmic factor in \cref{sec:samplizer-optimality}.
\end{remark}

\subsection{Optimality} \label{sec:samplizer-optimality}

Now we are going to show that our implementation of the samplizer is optimal up to a logarithmic factor. 
The strategy is to give a lower bound on the sample complexity of any samplizer by reducing it to sample-based Hamiltonian simulation \cite{LMR14,KLL+17}.
We first state the optimality of the samplizer in the following theorem.

\begin{theorem} [Optimality of samplizer]
\label{thm:optimality-samplizer}
    Suppose $C = \cbra{C\sbra{U}}$ is a quantum circuit family with $m$-ancilla block-encoded access to $\rho/2$ with query complexity $Q$ where $\rho$ is an $n$-qubit mixed quantum state. For every $\delta > 0$, it is necessary and sufficient to have sample complexity $\widetilde \Theta\rbra{Q^2/\delta}$ to implement a quantum channel family $\mathcal{E} = \cbra{\mathcal{E}\sbra{\rho}}$ satisfying: for every $\rho$, there is a unitary operator $U_\rho$ that is a $\rbra{2, m, 0}$-block-encoding of $\rho$ such that
    \[
    \Abs*{ \mathcal{E}\sbra{\rho} - \mathcal{C}\sbra{U_\rho} }_{\tr} \leq \delta.
    \]
\end{theorem}

We note that in \cref{thm:optimality-samplizer}, the distance measure between quantum channels is the trace norm distance $\Abs{\cdot}_{\tr}$ but not the diamond norm distance $\Abs{\cdot}_{\diamond}$ used in \cref{lemma:block-encoding-to-sample}.
From \cref{eq:tr-vs-diamond}, it can be seen that the condition required in \cref{thm:optimality-samplizer} is stricter than that in \cref{lemma:block-encoding-to-sample}. 
Specifically, let $\mathsf{S}_{\diamond}$ and $\mathsf{S}_{\tr}$ be the sample complexities for the samplizers with respect to the diamond norm distance and trace norm distance, respectively.
It is straightforward to see that $\mathsf{S}_{\tr} \leq \mathsf{S}_{\diamond}$ because every samplizer with respect to the diamond norm distance is also a samplizer with respect to the trace norm distance.
Then, \cref{thm:optimality-samplizer} implies that
\[
\underbrace{\Omega\rbra{Q^2/\delta} \leq \mathsf{S}_{\tr}}_{\text{\cref{thm:optimality-samplizer}}} \leq \underbrace{\mathsf{S}_{\diamond} \leq \widetilde O\rbra{Q^2/\delta}}_{\text{\cref{lemma:block-encoding-to-sample}}}.
\]
This means that our samplizer is optimal (up to a logarithmic factor) with respect to both the diamond norm distance and trace norm distance.

The upper bound in \cref{thm:optimality-samplizer} follows from the samplizer given in \cref{lemma:block-encoding-to-sample}. 
We only have to prove the lower bound, which is given in \cref{lemma:samplizer-lower-bound} as follows. 

\begin{lemma}
\label{lemma:samplizer-lower-bound}
    For every $0 < \delta \leq 1/36$ and $Q \geq \max\cbra{144c\pi\delta, 2c\ln\rbra{1/\delta}}$ for some constant $c > 0$, there is a quantum circuit family $C = \cbra{C\sbra{U}}$ with $m$-ancilla block-encoded access to $\rho/2$ with query complexity $Q$ where $\rho$ is an $n$-qubit mixed quantum state such that any implementation of a quantum channel family $\mathcal{E} = \cbra{\mathcal{E}\sbra{\rho}}$ requires $\Omega\rbra{Q^2/\delta}$ samples of $\rho$ to satisfy the following properties: for every $\rho$, there is a unitary operator $U_\rho$ that is a $\rbra{2, m, 0}$-block-encoding of $\rho$ such that
    \[
    \Abs*{ \mathcal{E}\sbra{\rho} - \mathcal{C}\sbra{U_\rho} }_{\tr} \leq \delta.
    \]
\end{lemma}

To prove lower bounds, we need the following result in \cite{KLL+17}, which involves the trace norm distance (a distance between quantum channels weaker than the diamond norm distance). 

\begin{theorem} [Lower bounds for density matrix exponentiation, {\cite[Theorem 2]{KLL+17}}] \label{thm:sample-based-hamiltonian-simulation-lb}
    Suppose that $\rho$ is a mixed quantum state.
    For every $0 < \delta \leq 1/6$ and $t \geq 6\pi \delta$, it is necessary to use $\Omega\rbra{t^2/\delta}$ samples of $\rho$ to implement a quantum channel $\mathcal{E}$ such that
    \[
    \Abs*{\mathcal{E} - e^{-i\rho t}}_{\tr} \leq \delta.
    \]
\end{theorem}

We choose Hamiltonian simulation as the hard instance for proving lower bounds.
To this end, we need the following quantum query algorithm for Hamiltonian simulation proposed in \cite{GSLW19}. 

\begin{theorem} [Hamiltonian simulation, {\cite[Corollary 32]{GSLW19}}]
\label{thm:hamiltonian-simulation}
    Suppose that $U$ is a unitary operator that is a $\rbra{1, a, 0}$-block-encoding of Hamiltonian $H$.
    For every $\varepsilon \in \rbra{0, 1/2}$ and $t \in \mathbb{R}$, it is necessary and sufficient to use 
    \[
    \Theta\rbra*{\abs*{t} + \frac{\log\rbra{1/\varepsilon}}{\log\rbra{e + \log\rbra{1/\varepsilon}/\abs*{t}}}}
    \]
    queries to $U$ to implement a unitary operator $V$ that is a $\rbra{1, a+2, \varepsilon}$-block-encoding of $e^{-iHt}$. 
\end{theorem}

Now we are ready to prove \cref{lemma:samplizer-lower-bound}. 

\begin{proof}[Proof of \cref{lemma:samplizer-lower-bound}]
    We choose $C = \cbra{C\sbra{U}}$ to be the quantum algorithm (circuit family) for Hamiltonian simulation given in \cref{thm:hamiltonian-simulation} such that, for every unitary operator $U$ that is a $\rbra{1, a, 0}$-block-encoding of $H$, $C\sbra{U}$ is a $\rbra{1, a+2, \varepsilon}$-block-encoding of $e^{-i2Ht}$ using 
    \[
    Q \leq c\rbra*{\abs*{2t} + \frac{\ln\rbra{1/\varepsilon}}{\ln\rbra{e + \ln\rbra{1/\varepsilon}/\abs*{2t}}}}
    \]
    queries to $U$ for some constant $c > 0$. 
    Here, we choose $\varepsilon = \delta$ and
    \[
    t = \frac{1}{2} \rbra*{ \frac{Q}{c} - \frac{\ln\rbra{1/\delta}}{\ln\rbra{e+c\ln\rbra{1/\delta}/Q}} }.
    \]
    It can be shown that $t \geq Q/\rbra{4c}$ under the given constraints. To see this, we only have to show that
    \[
    \frac{c\ln\rbra{1/\delta}}{Q} \leq \frac{1}{2} \ln\rbra*{ e + \frac{c\ln\rbra{1/\delta}}{Q} },
    \]
    which holds by noting that $x \leq \frac 1 2 \ln\rbra{e+x}$ for $0 < x \leq 1/2$ and $0 < c\ln\rbra{1/\delta}/Q \leq 1/2$.
    
    For every quantum state $\rho$, we are going to implement a quantum channel that is close to $e^{-i\rho t}$ through the quantum algorithm $C$. 
    Suppose we can implement a quantum channel family $\mathcal{E} = \cbra{\mathcal{E}\sbra{\rho}}$ using $\mathsf{S}$ samples of $\rho$ such that there is a unitary operator $U_\rho$ that is a $\rbra{1, a, 0}$-block-encoding of $\rho/2$ that satisfies
    \[
    \Abs{\mathcal{E}\sbra{\rho} - \mathcal{C}\sbra{U_{\rho}}}_{\tr} \leq \delta.
    \]
    It can be seen from \cref{thm:hamiltonian-simulation} that $C\sbra{U_{\rho}}$ uses $Q$ queries to $U_\rho$, and $C\sbra{U_{\rho}}$ is a $\rbra{1, a+2, \varepsilon}$-block-encoding of $e^{-i\rho t}$. 
    By one application of $\mathcal{E}\sbra{\rho}$, we can implement a quantum channel family $\mathcal{F} = \cbra{\mathcal{F}\sbra{\rho}}$ on the $n$-qubit subsystem such that
    \[
    \mathcal{F} \sbra{\rho} \rbra{\varrho} = \tr_{a+2} \rbra*{ \mathcal{E}\sbra{\rho} \rbra*{\varrho \otimes \ket{0}\bra{0}^{\otimes \rbra*{a+2}}} }.
    \]
    It can be shown (see \cref{lemma:F-eirhot}) that $\Abs{\mathcal{F}\sbra{\rho} - e^{-i\rho t}}_{\tr} \leq \delta + 5\varepsilon = 6\delta =: \delta'$. 
    Therefore, we have implemented a quantum channel $\mathcal{F}\sbra{\rho}$ that is close to $e^{-i\rho t}$. 
    By \cref{thm:sample-based-hamiltonian-simulation-lb}, noting that $\delta' = 6\delta \leq 1 / 6$ and $t \geq Q/\rbra{4c} \geq 6\pi \delta'$, we conclude that $\mathsf{S} = \Omega\rbra{t^2/\delta'} = \Omega\rbra{Q^2/\delta}$.
\end{proof}

To complete the proof of \cref{lemma:samplizer-lower-bound}, it remains to show the following technical lemma.

\begin{lemma}
\label{lemma:F-eirhot}
    In the proof of \cref{lemma:samplizer-lower-bound}, $\Abs{\mathcal{F}\sbra{\rho} - e^{-i\rho t}}_{\tr} \leq \delta + 5\varepsilon$.
\end{lemma}

To prove \cref{lemma:F-eirhot}, we need the following inequality.

\begin{lemma}
\label{lemma:eq-a-adag}
    Suppose $A$ is an operator with $\Abs{A} \leq 1$ and $U$ is a unitary operator. If $\Abs{A-U} \leq \varepsilon$ for some $\varepsilon \in \rbra{0, 1}$, then $\Abs{I - A^\dag A} \leq 3\varepsilon$. 
\end{lemma}
\begin{proof}
This is straightforward by the triangle inequality. 
\begin{align*}
    \Abs*{I - A^\dag A} 
    & = \Abs*{U^\dag U - A^\dag A} \\
    & \leq \Abs*{U^\dag U - U^\dag A} + \Abs*{U^\dag A - A^\dag A} \\
    & \leq \Abs*{U^\dag} \Abs*{U - A} + \Abs*{A} \Abs*{U^\dag - A^\dag} \\
    & \leq \Abs*{U - A} + \rbra*{\Abs*{U} + \Abs*{A - U}} \Abs*{U - A} \\
    & \leq 3\varepsilon.
\end{align*}
\end{proof}

Now we are ready to prove \cref{lemma:F-eirhot}.

\begin{proof} [Proof of \cref{lemma:F-eirhot}]
Let $V_j = \bra{j}_{a+2} C\sbra{U_{\rho}} \ket{0}_{a+2}$ for every $\ket{j}_{a+2}$, and
\[
\mathcal{F}'\rbra{\varrho} = \tr_{a+2}\rbra*{C\sbra{U_{\rho}} \rbra*{\varrho \otimes \ket{0}\bra{0}^{\otimes \rbra*{a+2}}} C\sbra{U_{\rho}}^\dag } = \sum_{j} V_j \varrho V_j^\dag.
\]
Note that $C\sbra{U_{\rho}}$ is a $\rbra{1, a+2, \varepsilon}$-block-encoding of $e^{-i\rho t}$, then
\[
\Abs*{V_0 - e^{-i\rho t}} = \Abs*{\bra{0}_{a+2} C\sbra*{U_\rho} \ket{0}_{a+2} - e^{-i\rho t}} \leq \varepsilon. 
\]
By \cref{lemma:eq-a-adag}, we have 
\begin{equation}
\label{eq:v0dag-v0}
    \Abs{I - V_0^\dag V_0} \leq 3\varepsilon.
\end{equation} 
We first split $\Abs{\mathcal{F}\sbra{\rho} - e^{-i\rho t}}_{\tr}$ into three terms by the triangle inequality, and then deal with them separately. 
\begin{align*}
    \Abs*{\mathcal{F}\sbra{\rho} - e^{-i \rho t}}_{\tr}
    & = \max_{\varrho} \Abs*{ \mathcal{F}\sbra{\rho}\rbra{\varrho} - e^{-i \rho t} \varrho e^{i \rho t} }_1 \\
    & \leq \max_{\varrho} \rbra*{ \Abs*{ \mathcal{F}\sbra{\rho}\rbra{\varrho} - \mathcal{F}'\rbra{\varrho}}_1 + \Abs*{\mathcal{F}'\rbra{\varrho} - V_0 \varrho V_0^\dag}_1 + \Abs*{ V_0 \varrho V_0^\dag - e^{-i \rho t} \varrho e^{i \rho t} }_1 }.
\end{align*}
For the first term $\Abs{ \mathcal{F}\sbra{\rho}\rbra{\varrho} - \mathcal{F}'\rbra{\varrho}}_1$, we note that 
\begin{align*}
\Abs*{\mathcal{F}\sbra{\rho}\rbra{\varrho} - \mathcal{F}'\rbra{\varrho}}_1
& = 
\Abs*{\tr_{a+2}\rbra*{\mathcal{E}\sbra{\rho}\rbra*{\varrho \otimes \ket{0}\bra{0}^{\otimes \rbra*{a+2}}}} - \tr_{a+2}\rbra*{\mathcal{C}\sbra*{U_\rho} \rbra*{\varrho \otimes \ket{0}\bra{0}^{\otimes \rbra*{a+2}}} }}_1 \\
& \leq \Abs*{\mathcal{E}\sbra{\rho}\rbra*{\varrho \otimes \ket{0}\bra{0}^{\otimes \rbra*{a+2}}} - \mathcal{C}\sbra*{U_\rho} \rbra*{\varrho \otimes \ket{0}\bra{0}^{\otimes \rbra*{a+2}}} }_1 \\
& \leq \Abs*{\mathcal{E}\sbra{\rho} - \mathcal{C}\sbra{U_\rho}}_{\tr} \leq \delta.
\end{align*}
For the second term $\Abs*{\mathcal{F}'\rbra{\varrho} - V_0 \varrho V_0^\dag}_1$, we note that
\begin{align*}
    \Abs*{\mathcal{F}'\rbra{\varrho} - V_0 \varrho V_0^\dag}_1
    = \Abs*{\sum_{j \neq 0} V_j \varrho V_j^\dag}_1 
    = \tr\rbra*{ \sum_{j \neq 0} V_j^\dag V_j \varrho } 
    = \tr\rbra*{ \rbra*{I - V_0^\dag V_0} \varrho } 
    \leq \Abs*{I - V_0^\dag V_0} 
    \leq 3\varepsilon,
\end{align*}
where the last inequality is from \cref{eq:v0dag-v0}.
For the third term $\Abs*{ V_0 \varrho V_0^\dag - e^{-i \rho t} \varrho e^{i \rho t} }_1$, we have
\begin{align}
    \Abs*{ V_0 \varrho V_0^\dag - e^{-i \rho t} \varrho e^{i \rho t} }_1
    & \leq \Abs*{ V_0 \varrho V_0^\dag - V_0 \varrho e^{i \rho t} }_1 + \Abs*{ V_0 \varrho e^{i \rho t} - e^{-i \rho t} \varrho e^{i \rho t} }_1 \nonumber \\
    & = \Abs*{ V_0 \varrho \rbra*{V_0^\dag - e^{i\rho t}} }_1 + \Abs*{ \rbra*{V_0 - e^{-i\rho t}} \varrho e^{i\rho t} }_1 \nonumber \\
    & \leq \Abs*{V_0} \Abs*{\varrho}_1 \Abs*{V_0^\dag - e^{i\rho t}} + \Abs*{V_0 - e^{-i\rho t}} \Abs*{\varrho}_1 \Abs*{e^{i\rho t}} \label{eq:holder} \\
    & \leq 2 \Abs*{V_0 - e^{-i\rho t}} \leq 2 \varepsilon, \nonumber
\end{align}
where \cref{eq:holder} uses the fact that $\Abs{ABC}_1 \leq \Abs{AB}_1 \Abs{C} \leq \Abs{A} \Abs{B}_1 \Abs{C}$ with H\"older's inequalities $\Abs{AB}_1 \leq \Abs{A}\Abs{B}_1$ and $\Abs{AB}_1 \leq \Abs{A}_1\Abs{B}$. 
Combining the three terms, we have
\[
\Abs*{\mathcal{F}\sbra{\rho} - e^{-i \rho t}}_{\tr} \leq \delta + 5\varepsilon.
\]
\end{proof}

\section{Von Neumann Entropy Estimator}
\label{sec:von-neumann}
In this section, we will analyze the quantum algorithm for estimating the von Neumann entropy given in \cref{algo:von-main-intro}.

\subsection{Subroutines with block-encoded access}

The function $\texttt{von\_Neumann\_subroutine}\rbra{\delta_p, \varepsilon_p, \delta_Q}\sbra{U_A}$ implemented in \cref{algo:von-main-intro} is a quantum query algorithm with block-encoded access. 
We restate it in \cref{algo:sub-von-block-encoded} with detailed constraints. 

\begin{algorithm}[!htp]
    \caption{$\texttt{von\_Neumann\_subroutine}\rbra{\delta_p, \varepsilon_p, \delta_Q}$ --- \textit{quantum query algorithm}}
    \label{algo:sub-von-block-encoded}
    \begin{algorithmic}[1]
    \Require $\delta_p \in (0, 1/3]$, $\varepsilon_p \in (0, 1/2]$, $\delta_Q \in \rbra{0, 1}$, and query access to unitary operator $U_A$ that is a $\rbra{1, m, 0}$-block-encoding of $N$-dimensional Hermitian operator $A$, where $m = O\rbra{\log\rbra{N}}$. 
    \Ensure The quantum circuit description of $U_{p\rbra{A}}$ with query access to $U_A$. 

    \State Let $p\rbra{x}$ be a polynomial of degree $d_p = O\rbra*{\frac{1}{\delta_p}\log\rbra*{\frac{1}{\varepsilon_p}}}$ (by \cref{lemma:poly-approx-ln}) such that
    \begin{align*}
        & \forall x \in \sbra*{-1, 1}, \quad \abs*{p\rbra*{x}} \leq \frac 1 2, \\
        & \forall x \in \sbra*{\delta_p, 1}, \quad \abs*{p\rbra{x} - \frac{\ln\rbra{1/x}}{4\ln\rbra{2/\delta_p}}} \leq \varepsilon_p.
    \end{align*}

    \State Construct unitary operator $U_{p\rbra{A}}$ (by \cref{thm:qsvt}) that is a $\rbra{1, a, \delta_Q}$-block-encoding of $p\rbra{A}$, where $a = m + 2$, using $O\rbra{d_p}$ queries to $U_A$. 

    \State \Return $U_{p\rbra{A}}$. 
    \end{algorithmic}
\end{algorithm}

\begin{lemma}
\label{lemma:von-neumann-block-encoded}
    For every $\delta_p \in (0, 1/3]$, $\varepsilon_p \in (0, 1/2]$, and $\delta_Q \in (0, 1)$, 
    \cref{algo:sub-von-block-encoded} will output
    the quantum circuit description of $U_{p\rbra{A}}$ in classical time $\poly\rbra{1/\delta_p, \log\rbra{1/\varepsilon_p}, \log\rbra{1/\delta_Q}}$, and $U_{p\rbra{A}}$ makes $Q$ queries to $U_A$ and $O\rbra{Q\log\rbra{N}}$ one- and two-qubit gates, where $Q = O\rbra*{\frac{1}{\delta_p}\log\rbra*{\frac{1}{\varepsilon_p}}}$.
    Moreover, if $A = \rho/2$ for an $N$-dimensional quantum state $\rho$ of rank $r$, then the Hadamard test (by \cref{thm:hadamard}) on $U_{p\rbra{\rho/2}}$ and $\rho$ will output $1$ with probability $p_a$ such that
    \[
        \abs*{\rbra*{4\rbra*{2p_a - 1}\ln\rbra*{\frac{2}{\delta_p}} - \ln\rbra{2}} - S\rbra*{\rho}} \leq 4 \rbra*{ 2r\delta_p + \varepsilon_p + r\delta_Q } \ln \rbra*{\frac{2}{\delta_p}}.
    \]
\end{lemma}

\begin{proof}
    Let
    \[
    f\rbra{x} = \frac{\ln\rbra{1/x}}{4\ln\rbra{2/\delta_p}}.
    \]
    By \cref{lemma:poly-approx-ln}, we can choose a polynomial $p\rbra{x}$ of degree $d_p = O\rbra*{\frac{1}{\delta_p}\log\rbra*{\frac{1}{\varepsilon_p}}}$ such that
    \begin{align*}
        & \forall x \in \sbra*{-1, 1}, \quad \abs*{p\rbra*{x}} \leq \frac 1 2, \\
        & \forall x \in \sbra*{\delta_p, 1}, \quad \abs*{p\rbra{x} - f\rbra*{x}} \leq \varepsilon_p.
    \end{align*}

    Suppose $U_A$ is a unitary operator and is a $\rbra{1, m, 0}$-block-encoding of $A$ with $m = O\rbra{\log\rbra{N}}$. 
    By \cref{thm:qsvt}, we can construct a unitary operator $U_{p\rbra{A}}$, that is a $\rbra{1, m+2, \delta_Q}$-block-encoding of $p\rbra{A}$, using $O\rbra{d_p}$ queries to $U_A$, where $\delta_Q \in \rbra{0, 1}$ is a parameter to be determined. 
    It should be noted that the description of $U_{p\rbra{A}}$ can be computed in classical time $\poly\rbra{d_p, \log\rbra{1/\delta_Q}}$.
    To make it clearer, let $a = m+2$, then 
    \begin{equation}
    \label{eq:error-delta-Q}
    \Abs*{\bra{0}_a U_{p\rbra{A}} \ket{0}_a - p\rbra{A}} \leq \delta_Q. 
    \end{equation}

    Now we consider the special case that $A = \rho/2$. 
    By the Hadamard test (\cref{thm:hadamard}), we can construct a quantum circuit $U_H$, using $1$ query to $U_{p\rbra{\rho/2}}$ and $1$ sample of $\rho$, that outputs $1$ with probability 
    \[
    p_a = \frac{1 + \operatorname{Re}\rbra{\tr\rbra{\bra{0}_a U_{p\rbra{\rho/2}} \ket{0}_a \rho}}}{2},
    \]
    and $0$ otherwise. 

    Now we are going to show the relationship between $p_a$ and $S\rbra{\rho}$.
    We first note that 
    \begin{equation}
    \label{eq:tr-f-S-rho}
    \tr\rbra*{f\rbra*{\frac{\rho}{2}}\rho} = \frac{1}{4\ln\rbra*{2/\delta_p}} \rbra*{S\rbra*{\rho}+\ln\rbra{2}}.
    \end{equation}
    It can be shown that (see \cref{lemma:von-neumann-error-analysis})
    \begin{equation}
    \label{eq:von-neumann-poly-error}
    \abs*{\tr\rbra*{p\rbra*{\frac{\rho}{2}}\rho} - \tr\rbra*{f\rbra*{\frac{\rho}{2}}\rho}} \leq 2 r \delta_p + \varepsilon_p.
    \end{equation}
    From \cref{eq:error-delta-Q} we know that
    \[
    \abs*{ \tr\rbra*{\bra{0}_a U_{p\rbra{\rho/2}} \ket{0}_a\rho} - \tr\rbra*{p\rbra*{\frac{\rho}{2}}\rho} } \leq r\delta_Q.
    \]
    Note that $\tr\rbra*{p\rbra{\rho/2}\rho}$ is a real number, then we only need to focus on the real part and the above inequality can be reduced to
    \begin{equation}
    \label{eq:von-neumann-reduce-real}
    \abs*{ \rbra*{2p_a-1} - \tr\rbra*{p\rbra*{\frac{\rho}{2}}\rho} } = \abs*{ \operatorname{Re}\rbra*{\tr\rbra*{\bra{0}_a U_{p\rbra{\rho/2}} \ket{0}_a\rho} } - \tr\rbra*{p\rbra*{\frac{\rho}{2}}\rho} } \leq r\delta_Q.
    \end{equation}
    By \cref{eq:tr-f-S-rho}, \cref{eq:von-neumann-poly-error} and \cref{eq:von-neumann-reduce-real}, we have
    \[
    \abs*{ \rbra*{2p_a-1} - \frac{1}{4\ln\rbra*{2/\delta_p}} \rbra*{S\rbra*{\rho} + \ln\rbra{2}} } \leq 2 r\delta_p + \varepsilon_p + r\delta_Q,
    \]
    which yields the proof.
\end{proof}

To complete the proof of \cref{lemma:von-neumann-block-encoded}, it remains to show the following technical lemma.

\begin{lemma}
\label{lemma:von-neumann-error-analysis}
    In the proof of \cref{lemma:von-neumann-block-encoded}, we have
    \[
    \abs*{\tr\rbra*{p\rbra*{\frac{\rho}{2}}\rho} - \tr\rbra*{f\rbra*{\frac{\rho}{2}}\rho}} \leq {2} r \delta_p + \varepsilon_p.
    \]
\end{lemma}
\begin{proof}
    Suppose $\rho = \sum_{i \in \sbra{N}} x_i \ket{\psi_i} \bra{\psi_i}$, where $x_i \geq 0$, $\sum_{i \in \sbra{N}} x_i = 1$, and $\cbra{\ket{\psi_i}}$ is an orthonormal basis. 
    Because $\rho$ is of rank $r$, we can assume that $x_i = 0$ for all $i \in \sbra{N} \setminus \sbra{r}$ without loss of generality. 
    Then, 
    \begin{align*}
    \abs*{\tr\rbra*{p\rbra*{\frac{\rho}{2}}\rho} - \tr\rbra*{f\rbra*{\frac{\rho}{2}}\rho}} 
    & \leq \sum_{i \in \sbra*{r}} \abs*{ x_i p\rbra*{{\frac{x_i}{2}}} - x_i f\rbra*{{\frac{x_i}{2}}} } \\
    & = \sum_{i \in \sbra*{r} \colon x_i \leq {2}\delta_p} \abs*{ x_i p\rbra*{{\frac{x_i}{2}}} - x_i f\rbra*{{\frac{x_i}{2}}} } + \sum_{i \in \sbra*{r} \colon x_i > {2}\delta_p} \abs*{ x_i p\rbra*{{\frac{x_i}{2}}} - x_i f\rbra*{{\frac{x_i}{2}}} }.
    \end{align*}
    For the first term, we note that $p\rbra{x} \leq 1/2$ for every $x \in \sbra{-1, 1}$ and $x \ln\rbra{{2}/x} \leq \delta_p \ln\rbra{{2}/\delta_p}$ for every $0 \leq x \leq \delta_p \leq 1/3$. Thus, we have
    \begin{align*}
        \sum_{i \in \sbra*{r} \colon x_i \leq {2}\delta_p} \abs*{ x_i p\rbra*{{\frac{x_i}{2}}} - x_i f\rbra*{{\frac{x_i}{2}}} }
        & \leq \sum_{i \in \sbra*{r} \colon x_i \leq {2}\delta_p} \rbra*{ \abs*{ x_i p\rbra*{{\frac{x_i}{2}}} } + \abs*{ x_i f\rbra*{{\frac{x_i}{2}}} } } \\
        & \leq \sum_{i \in \sbra*{r} \colon x_i \leq {2}\delta_p} \rbra*{ \frac 1 2 \cdot {2} \delta_p + \frac{\delta_p \ln\rbra{{2}/\delta_p}}{4\ln\rbra{2/\delta_p}} } \\
        & \leq {2} r \delta_p.
    \end{align*}
    For the last term, we note that $\abs{p\rbra{x} - f\rbra{x}} \leq \varepsilon_p$ for $x \in \sbra{\delta_p, 1}$. Thus, we have
    \begin{align*}
        \sum_{i \in \sbra*{r} \colon x_i > {2}\delta_p} \abs*{ x_i p\rbra*{{\frac{x_i}{2}}} - x_i f\rbra*{{\frac{x_i}{2}}} }
        \leq \sum_{i \in \sbra*{r} \colon x_i > {2}\delta_p} x_i \varepsilon_p 
        \leq \sum_{i \in \sbra*{r}} x_i \varepsilon_p = \varepsilon_p.
    \end{align*}
    Combining the above, we have
    \[
    \abs*{\tr\rbra*{p\rbra*{{\frac{x_i}{2}}}\rho} - \tr\rbra*{f\rbra*{{\frac{x_i}{2}}}\rho}} \leq {2} r \delta_p + \varepsilon_p.
    \]
\end{proof}

\subsection{Sample access}

We state the main function $\texttt{estimate\_von\_Neumann\_main}\rbra{\varepsilon, \delta}$ in \cref{algo:von-main-intro} again in \cref{algo:von-sample} with detailed constraints. 

\begin{algorithm}[!htp]
    \caption{$\texttt{estimate\_von\_Neumann\_main}\rbra{\varepsilon, \delta}$ --- \textit{quantum sample algorithm}}
    \label{algo:von-sample}
    \begin{algorithmic}[1]
    \Require Additive precision $\varepsilon \in \rbra{0, 1}$, error probability $\delta \in \rbra{0, 1}$, access to independent samples of $\rho$ of rank $r$, and subroutine $\texttt{von\_Neumann\_subroutine}\rbra{\delta_p, \delta_p, \delta_Q}$ defined by \cref{algo:sub-von-block-encoded}.
    \Ensure An estimate $\widetilde S$ of $S\rbra{\rho}$. 

    \State $\delta_p \gets \frac{\varepsilon}{128r\ln\rbra*{{32r}/{\varepsilon}}}$, $\varepsilon_p \gets \frac{\varepsilon}{32\ln\rbra{2/\delta_p}}$, $\delta_Q \gets \frac{\varepsilon}{32r\ln\rbra{2/\delta_p}}$, $\delta_a \gets \frac{\varepsilon}{64\ln\rbra{2/\delta_p}}$, $\varepsilon_H \gets \delta_a$, $k \gets \ceil*{\frac{1}{2\varepsilon_H^2} \ln\rbra*{\frac{2}{\delta}}}$.
    
    \For{$i = 1 \dots k$}
    \State Perform the Hadamard test on $\mathsf{Samplize}_{\delta_a}\ave{\texttt{von\_Neumann\_subroutine}\rbra{\delta_p, \varepsilon_p, \delta_Q}}\sbra{\rho}$ and $\rho$ (by \cref{thm:hadamard}). Let $X_i \in \cbra{0, 1}$ be the outcome.
    \EndFor

    \State $\widetilde S \gets 4 \rbra{ 2\sum_{i\in\sbra{k}}X_i/k - 1} \ln\rbra{2/\delta_p} -\ln\rbra{2}$.

    \State \Return $\widetilde S$.
    \end{algorithmic}
\end{algorithm}

\begin{theorem}
\label{thm:von-sample}
    \cref{algo:von-sample} with sample access to $N$-dimensional quantum state $\rho$ of rank $r$ that, with probability $\geq 1 - \delta$, estimates the von Neumann entropy $S\rbra{\rho}$ within additive error $\varepsilon$ with sample complexity $M$ and time complexity $O\rbra{M \log\rbra{N}}$, where
    \[
    M = O\rbra*{\frac{r^2}{\varepsilon^5} \log^7\rbra*{\frac{r}{\varepsilon}} \log^2\rbra*{\frac{\log\rbra*{r}}{\varepsilon}} \log\rbra*{\frac{1}{\delta}}}.
    \]
\end{theorem}

To analyze the correctness of the algorithms, we need Hoeffding's inequality as follows. 

\begin{theorem} [Hoeffding's inequality, {\cite[Theorem 2]{Hoe63}}]
\label{thm:hoeffding}
    Suppose $X_1, X_2, \dots, X_n$ are independent random variables with $a_i \leq X_i \leq b_i$. Let
    \[
    X = \sum_{i \in \sbra{n}} X_i.
    \]
    For every $t > 0$, we have
    \[
    \Pr \sbra*{ X - \mathbb{E}\sbra{X} \geq t } \leq \exp\rbra*{-\frac{2t^2}{\sum\limits_{i \in \sbra{n}} \rbra*{b_i-a_i}^2}}.
    \]
\end{theorem}

Now we are ready to prove \cref{thm:von-sample}. 

\begin{proof}[Proof of \cref{thm:von-sample}]
    We take $\delta_p \in (0, 1/3]$, $\varepsilon_p \in (0, 1/2]$, and $\delta_Q \in (0, 1)$ to be determined. 
    By \cref{lemma:von-neumann-block-encoded}, the Hadamard test on $\texttt{von\_Neumann\_subroutine}\rbra{\delta_p, \varepsilon_p, \delta_Q}\sbra{U_\rho}$ and $\rho$ using $1$ sample of $\rho$ and $Q$ queries to $U_{\rho}$ that outputs $1$ with probability $p_a$ such that 
    \begin{equation}
    \label{eq:von-pa-S}
        \abs*{ \rbra*{4\rbra*{2p_a - 1}\ln\rbra*{\frac{2}{\delta_p}} -\ln\rbra{2}} - S\rbra*{\rho}} \leq 4 \rbra*{ 2r\delta_p + \varepsilon_p + r\delta_Q } \ln \rbra*{\frac{2}{\delta_p}},
    \end{equation}
    with time complexity $O\rbra*{Q \log\rbra*{N}}$, where $Q = O\rbra*{\frac{1}{\delta_p}\log\rbra*{\frac{1}{\varepsilon_p}}}$.

    Let $\delta_a > 0$ to be determined. 
    By \cref{lemma:block-encoding-to-sample}, there is a unitary operator $U_\rho$ that is a block-encoding of $\rho$, 
    \[
    \begin{aligned}
    \big\|\mathsf{Samplize}_{\delta_a}\ave{&\texttt{von\_Neumann\_subroutine}\rbra{\delta_p, \varepsilon_p, \delta_Q}}\sbra{\rho} - \\
    &\texttt{von\_Neumann\_subroutine}\rbra{\delta_p, \varepsilon_p, \delta_Q}\sbra{U_\rho}\big\|_{\diamond} \leq \delta_a.
    \end{aligned}
    \]
    This means that the Hadamard test on $\mathsf{Samplize}_{\delta_a}\ave{\texttt{von\_Neumann\_subroutine}\rbra{\delta_p, \varepsilon_p, \delta_Q}}\sbra{\rho}$ and $\rho$ will output $1$ with probability $\widetilde p_a$, where
    \begin{equation}
    \label{eq:von-pa-tildepa}
        \abs*{p_a - \widetilde p_a} \leq \delta_a,
    \end{equation}
    with sample complexity 
    \[
    1 + O\rbra*{\frac{Q^2}{\delta_a} \log^2\rbra*{\frac{Q}{\delta_a}}} = O\rbra*{\frac{1}{\delta_a \delta_p^2} \log^2\rbra*{\frac{1}{\varepsilon_p}} \log^2\rbra*{\frac{1}{\delta_a\delta_p}\log\rbra*{\frac{1}{\varepsilon_p}}}}.
    \]
    We call by algorithm $\mathcal{A}$ the above procedure that outputs $1$ with probability $\widetilde p_a$.

    We repeat algorithm $\mathcal{A}$ for $k$ times and let $X_i \in \cbra{0, 1}$ be the output of the $i$-th repetition.  
    Let
    \begin{equation}
    \label{eq:von-tildeS-X}
    \widetilde S = 4\rbra{2X-1}\ln\rbra*{\frac{2}{\delta_p}} -\ln\rbra{2},
    \end{equation}
    where
    \[
    X = \frac 1 k \sum_{i \in \sbra{k}} X_i,
    \]
    We note that $\mathbb{E}\sbra{X} = \mathbb{E}\sbra{X_i} = \widetilde p_a$.
    By Hoeffding's inequality (\cref{thm:hoeffding}), for every $\varepsilon_H > 0$, 
    \begin{equation}
    \label{eq:von-X-tildepa}
    \Pr \sbra*{ \abs*{X - \widetilde p_a} \leq \varepsilon_H } \geq 1 - 2\exp\rbra*{-2k\varepsilon_H^2}. 
    \end{equation}
    From \cref{eq:von-pa-S}, \cref{eq:von-pa-tildepa}, \cref{eq:von-tildeS-X}, and \cref{eq:von-X-tildepa} we know that with probability $\geq 1 - 2\exp\rbra{-2k\varepsilon_H^2}$,
    \begin{equation}
    \label{eq:von-error-all}
    \abs*{\widetilde S - S\rbra*{\rho}} \leq 4 \rbra*{ 2r\delta_p + \varepsilon_p + r\delta_Q + 2\delta_a + 2\varepsilon_H } \ln \rbra*{\frac{2}{\delta_p}}.
    \end{equation}
    By taking $\delta_p = \frac{\varepsilon}{128r\ln\rbra*{{32r}/{\varepsilon}}}$, $\varepsilon_p = \frac{\varepsilon}{32\ln\rbra{2/\delta_p}}$, $\delta_Q = \frac{\varepsilon}{32r\ln\rbra{2/\delta_p}}$, and $\delta_a = \varepsilon_H = \frac{\varepsilon}{64\ln\rbra{2/\delta_p}}$, we have $\abs*{\widetilde S - S\rbra*{\rho}} \leq \varepsilon$ (see \cref{lemma:von-neumann-ln-bound}). 
    
    Now we are going to analyze the complexity. 
    We take $k = \ceil*{\frac{1}{2\varepsilon_H^2} \ln\rbra*{\frac{2}{\delta}}}$ to make \cref{eq:von-error-all} hold with probability $\geq 1 - \delta$.
    Our overall algorithm repeats algorithm $\mathcal{A}$ for $k$ times, thus the total sample complexity is
    \[
    k \cdot O\rbra*{\frac{1}{\delta_a \delta_p^2} \log^2\rbra*{\frac{1}{\varepsilon_p}} \log^2\rbra*{\frac{1}{\delta_a\delta_p}\log\rbra*{\frac{1}{\varepsilon_p}}}} = O\rbra*{\frac{r^2}{\varepsilon^5} \log^7\rbra*{\frac{r}{\varepsilon}} \log^2\rbra*{\frac{\log\rbra*{r}}{\varepsilon}} \log\rbra*{\frac{1}{\delta}}},
    \]
    and the time complexity is only a $\log\rbra{N}$ factor over the sample complexity, which is
    \[
    O\rbra*{\frac{r^2}{\varepsilon^5} \log^7\rbra*{\frac{r}{\varepsilon}} \log^2\rbra*{\frac{\log\rbra*{r}}{\varepsilon}} \log\rbra*{\frac{1}{\delta}} \log\rbra*{N}}.
    \]
\end{proof}

To complete the proof of \cref{thm:von-sample}, it remains to show the following technical lemma.

\begin{lemma}
\label{lemma:von-neumann-ln-bound}
    In the proof of \cref{thm:von-sample}, if we take $\delta_p = \frac{\varepsilon}{{128}r\ln\rbra*{{{32}r}/{\varepsilon}}}$, $\varepsilon_p = \frac{\varepsilon}{32\ln\rbra{2/\delta_p}}$, $\delta_Q = \frac{\varepsilon}{32r\ln\rbra{2/\delta_p}}$, and $\delta_a = \varepsilon_H = \frac{\varepsilon}{64\ln\rbra{2/\delta_p}}$, then \cref{eq:von-error-all} will become $\abs*{\widetilde S - S\rbra*{\rho}} \leq \varepsilon$.
\end{lemma}
\begin{proof}
    The idea is to show that each term in the right hand side of \cref{eq:von-error-all} is $\leq O\rbra{\varepsilon}$. Strictly speaking, we will show that
    \[
    \abs*{\widetilde S - S\rbra*{\rho}} \leq \frac{\varepsilon}{4} + \frac{\varepsilon}{8} + \frac{\varepsilon}{8} + \frac{\varepsilon}{8} + \frac{\varepsilon}{8} < \varepsilon.
    \]
    The most complicated part is to show that
    \[
    {8} r \delta_p \ln\rbra*{\frac{2}{\delta_p}} \leq \frac \varepsilon 4,
    \]
    which is
    \[
    \frac{\varepsilon}{{128}r \ln\rbra{{32}r/\varepsilon}} \ln\rbra*{\frac{{256}r\ln\rbra{{32}r/\varepsilon}}{\varepsilon}} \leq \frac{\varepsilon}{{32}r}.
    \]
    Let $x = \varepsilon/32r \in \rbra{0, 1}$. The above inequality becomes
    \[
    \frac{x}{4\ln\rbra{1/x}} \ln\rbra*{\frac{8\ln\rbra{1/x}}{x}} \leq x,
    \]
    which can be simplified to 
    \[
    x^3 \ln\rbra*{\frac{1}{x}} \leq \frac 1 8.
    \]
    The proof is completed by noting that $g\rbra{x} = x^3 \ln\rbra{1/x}$ takes the maximum value at $x = e^{-1/3}$, and thus
    \[
        g\rbra{x} \leq g\rbra*{e^{-1/3}} = \frac{1}{3e} \leq \frac 1 8.
    \]
\end{proof}

\section{R\'{e}nyi Entropy Estimator}
\label{sec:renyi}
In this section, we will present and analyze our quantum sample algorithms for estimating $\alpha$-R\'enyi entropy for $\alpha > 1$ and $0 < \alpha < 1$ separately. 
Both algorithms have a similar recursive structure inspired by the quantum algorithm for estimating R\'enyi entropy of probability distributions in \cite[Algorithm 4]{LW19}.
We will consider the case of $\alpha > 1$ in \cref{sec:alpha-gt1} and the case of $0 < \alpha < 1$ in \cref{sec:alpha-lt1}.

For convenience, we write $P_\alpha\rbra{\rho} = \tr\rbra{\rho^{\alpha}}$. 
Before the analysis, we recall the basic property of $P_\alpha\rbra{\rho}$ that will be often used as follows.

\begin{fact}
\label{fact:renyi}
    Suppose $\rho$ is a quantum state of rank $r$. Then, 
    \begin{equation*}
        \begin{matrix*}[l]
            \forall 0 < \alpha < 1, & 1 \leq P_{\alpha}\rbra{\rho} \leq r^{1-\alpha}, \\
            \forall \alpha > 1, & r^{1-\alpha} \leq P_{\alpha}\rbra{\rho} \leq 1.
        \end{matrix*}
    \end{equation*}
\end{fact}

\subsection{The case of \texorpdfstring{$\alpha > 1$}{α > 1}} \label{sec:alpha-gt1}

We formalize our quantum sample algorithm for estimating $\alpha$-R\'enyi entropy for $\alpha > 1$ in \cref{algo:renyi-main-intro} through the samplizer.

\subsubsection{Recursive framework}

Our algorithm has a recursive structure: the estimation of R\'{e}nyi entropy in the general case can be reduced to a special case with a promise that $P_{\alpha}\rbra{\rho}$ is upper and lower bounded.
The abstract algorithm is given in \cref{algo:esimate-large} (which restates the function $\texttt{estimate\_R\'enyi\_gt1}\rbra{\alpha, \varepsilon, \delta}$ in \cref{algo:renyi-main-intro}), where $\texttt{estimate\_R\'enyi\_gt1\_promise}\rbra{\alpha, P, \varepsilon, \delta}$ indicates an algorithm that, with probability $\geq 1 - \delta$, outputs an estimate $\widetilde P$ of $P_\alpha\rbra{\rho}$ such that $\rbra{1-\varepsilon}P_\alpha\rbra{\rho} \leq \widetilde P \leq \rbra{1+\varepsilon} P_\alpha\rbra{\rho}$, given a promise that $P \leq P_\alpha\rbra{\rho} \leq 10P$, where $P$ is the prior knowledge of $P_\alpha\rbra{\rho}$ given as input.

\begin{algorithm}[!htp]
    \caption{$\texttt{estimate\_R\'enyi\_gt1}\rbra{\alpha, \varepsilon, \delta}$ --- \textit{quantum sample algorithm}}
    \label{algo:esimate-large}
    \begin{algorithmic}[1]
    \Require $\alpha > 1$, $\varepsilon \in \rbra{0, 1}$, $\delta \in \rbra{0, 1}$, and access to $N$-dimensional quantum state $\rho$ of rank $r$.
    \Ensure $\widetilde P$ such that $\rbra{1 - \varepsilon} P_{\alpha}\rbra{\rho} \leq \widetilde P \leq \rbra{1 + \varepsilon} P_{\alpha}\rbra{\rho}$ with probability $\geq 1 - \delta$. 

    \State $\lambda \gets 1 + 1 / \ln \rbra{r}$.
    \If {$\alpha \leq \lambda$}
        \State $P \gets e^{-1}$.
    \Else
        \State $P' \gets \texttt{estimate\_R\'enyi\_gt1}\rbra{\alpha/\lambda, 1/4, \delta/2}$. \label{line:call-gt1}
        \State $P \gets \rbra{4P'/5}^\lambda e^{-1}$.
    \EndIf
    
    \State \Return $\texttt{estimate\_R\'enyi\_gt1\_promise}\rbra{\alpha, P, \varepsilon, \delta/2}$.
    \end{algorithmic}
\end{algorithm}

\begin{lemma}
\label{lemma:annealing}
    Suppose $\alpha > 1$, and $\texttt{\textup{estimate\_R\'enyi\_gt1\_promise}}\rbra{\alpha, P, \varepsilon, \delta}$ is a quantum algorithm with time complexity $T\rbra{\alpha, P, \varepsilon, \delta}$ that, with probability $\geq 1 - \delta$, outputs an estimate of $P_\alpha\rbra{\rho}$ such that $\rbra{1-\varepsilon}P_\alpha\rbra{\rho} \leq \widetilde P \leq \rbra{1+\varepsilon} P_\alpha\rbra{\rho}$, given a promise that $P \leq P_\alpha\rbra{\rho} \leq 10P$. Then, with probability $\geq 1 - \delta$, \cref{algo:esimate-large} outputs an estimate $\widetilde P$ of $P_\alpha\rbra{\rho}$ such that $\rbra{1-\varepsilon}P_\alpha\rbra{\rho} \leq \widetilde P \leq \rbra{1+\varepsilon} P_\alpha\rbra{\rho}$ with time complexity $Q\rbra{\alpha, \varepsilon, \delta}$, where
    \[
    Q\rbra*{\alpha, \varepsilon, \delta} = \begin{cases}
        T\rbra*{\alpha, e^{-1}, \varepsilon, \delta/2}, & 1 < \alpha \leq \lambda, \\
        Q\rbra*{\alpha/\lambda, 1/4, \delta/2} + \sup\limits_{P \in \sbra*{r^{1-\alpha}/10, 1}} T\rbra{\alpha, P, \varepsilon, \delta/2}, & \alpha > \lambda,
    \end{cases}
    \]
    in which $\lambda = 1 + 1/\ln\rbra{r}$ and $r$ is the rank of $\rho$. 
\end{lemma}

To prove \cref{lemma:annealing}, we need the following inequality.

\begin{lemma} [{\cite[Lemma 5.3]{LW19}}] \label{lemma:p-alpha-beta}
    Suppose $p_1, p_2, \dots, p_n$ is a probability distribution, i.e., $p_i \geq 0$ and $\sum_{i \in \sbra{n}} p_i = 1$. Then, for every $0 < \alpha < \beta$, we have 
    \[
    \rbra*{ \sum_{i \in \sbra*{n}} p_i^{\beta} }^{\frac{\alpha}{\beta}} \leq \sum_{i \in \sbra*{n}} p_i^{\alpha} \leq n^{1-\frac{\alpha}{\beta}} \rbra*{ \sum_{i \in \sbra*{n}} p_i^{\beta} }^{\frac{\alpha}{\beta}}.
    \]
\end{lemma}

Now we are ready to prove \cref{lemma:annealing}.

\begin{proof}[Proof of \cref{lemma:annealing}]
    Suppose $\rho = \sum_{i \in \sbra{N}} x_i \ket{\psi_i} \bra{\psi_i}$, where $x_i \geq 0 $, $\sum_{i \in \sbra{N}} x_i = 1$, and $\cbra{\ket{\psi_i}}$ is an orthonormal basis.
    \begin{itemize}
        \item For the basis case that $1 < \alpha \leq \lambda$, 
        \[
        P_{\alpha}\rbra*{\rho} \geq r^{1-\alpha} \geq r^{-1/\ln\rbra*{r}} = e^{-1}.
        \]
        On the other hand, $P_\alpha\rbra{\rho} \leq 1$. These together yield that $P = e^{-1} \leq P_{\alpha}\rbra{\rho} \leq 1 \leq 10 P$, which satisfies the required condition of $\texttt{\textup{estimate\_R\'enyi\_gt1\_promise}}\rbra{\alpha, P, \varepsilon, \delta/2}$. Therefore, \cref{algo:esimate-large} will output an estimate of $P_\alpha\rbra{\rho}$ within multiplicative error $\varepsilon$ with probability $\geq 1 - \delta/2 \geq 1 - \delta$, with time complexity $Q\rbra{\alpha, \varepsilon, \delta} = T\rbra{\alpha, P, \varepsilon, \delta/2}$. 
        \item Now we consider the case that $\alpha > \lambda$. Let $\alpha' = \alpha/\lambda < \alpha$. By \cref{lemma:p-alpha-beta}, we have
        \[
        \rbra*{P_{\alpha}\rbra*{\rho}}^{1/\lambda} = \rbra*{P_{\alpha}\rbra*{\rho}}^{\frac{\alpha'}{\alpha}} \leq P_{\alpha'}\rbra*{\rho} \leq r^{1-\frac{\alpha'}{\alpha}} \rbra*{P_{\alpha}\rbra*{\rho}}^{\frac{\alpha'}{\alpha}} = r^{1-1/\lambda} \rbra*{P_{\alpha}\rbra*{\rho}}^{1/\lambda}.
        \]
        By induction, after calling $\texttt{estimate\_R\'enyi\_gt1}\rbra{\alpha', \varepsilon', \delta'}$ in Line \ref{line:call-gt1} of \cref{algo:esimate-large}, where $\varepsilon' = 1/4$ and $\delta' = \delta/2$, we will obtain $P'$ such that $\rbra{1-\varepsilon'}P_{\alpha'}\rbra{\rho} \leq P' \leq \rbra{1+\varepsilon'}P_{\alpha'}\rbra{\rho}$ with probability $\geq 1 - \delta/2$. Thus, we have
        \[
        P_{\alpha}\rbra*{\rho} \geq \rbra*{r^{1/\lambda-1} P_{\alpha'}\rbra*{\rho}}^{\lambda} \geq r^{1-\lambda} \rbra*{\frac{P'}{1+\varepsilon'}}^{\lambda} = e^{-1} \rbra*{\frac{4P'}{5}}^{\lambda} = P.
        \]
        On the other hand, 
        \[
        P_{\alpha}\rbra*{\rho} \leq \rbra*{P_{\alpha'}\rbra*{\rho}}^{\lambda} \leq \rbra*{\frac{P'}{1-\varepsilon'}}^{\lambda} = \rbra*{\frac{4P'}{3}}^{\lambda} \leq 10 P.
        \]
        Therefore, $P \leq P_{\alpha}\rbra{\rho} \leq 10P$ with probability $\geq 1 - \delta/2$, i.e., the required condition of 
        \[
        \texttt{\textup{estimate\_R\'enyi\_gt1\_promise}}\rbra{\alpha, P, \varepsilon, \delta/2}
        \]
        holds with probability $\geq 1 - \delta/2$; then, \cref{algo:esimate-large} will output an estimate of $P_{\alpha}$ within multiplicative error $\varepsilon$ with probability $\geq \rbra{1-\delta/2}^2 \geq 1 - \delta$, with time complexity $Q\rbra{\alpha, \varepsilon, \delta} = Q\rbra{\alpha', \varepsilon', \delta'} + T\rbra{\alpha, P, \varepsilon, \delta/2}$. Without loss of generality, we may assume that it always holds that $r^{1-\alpha}/10 \leq P \leq 1$ during the execution of \cref{algo:esimate-large}. Thus, 
        \[
        Q\rbra*{\alpha, \varepsilon, \delta} = Q\rbra{\alpha', \varepsilon', \delta'} + \sup_{P \in \sbra*{r^{1-\alpha}/10, 1}} T\rbra{\alpha, P, \varepsilon, \delta/2}.
        \]
    \end{itemize}
    From the above, we have
    \[
    Q\rbra*{\alpha, \varepsilon, \delta} = \begin{cases}
        T\rbra*{\alpha, \varepsilon, e^{-1}, \delta/2}, & 1 < \alpha \leq \lambda \\
        Q\rbra*{\alpha/\lambda, 1/4, \delta/2} + \sup\limits_{P \in \sbra*{r^{1-\alpha}/10, 1}} T\rbra{\alpha, P, \varepsilon, \delta/2}, & \alpha > \lambda.
    \end{cases}
    \]
\end{proof}

On the basis of \cref{algo:esimate-large}, we only have to implement $\texttt{estimate\_R\'enyi\_gt1\_promise}\rbra{\alpha, P, \varepsilon, \delta}$ with sample access. For readability, we will first analysis the subroutine with block-encoded access that is used to implement it.

\subsubsection{Subroutines with block-encoded access}

The function $\texttt{R\'enyi\_gt1\_subroutine}\rbra{\alpha, P, \delta_p, \varepsilon_p, \delta_Q}\sbra{U_A}$ implemented in \cref{algo:renyi-main-intro} is with block-encoded access. 
We restate it in \cref{algo:sub-renyi-gt1} with detailed constraints. 

\begin{algorithm}[!htp]
    \caption{$\texttt{R\'enyi\_gt1\_subroutine}\rbra{\alpha, P, \delta_p, \varepsilon_p, \delta_Q}$ --- \textit{quantum query algorithm}}
    \label{algo:sub-renyi-gt1}
    \begin{algorithmic}[1]
    \Require $\alpha > 1$, $\delta_p \in (0, \beta]$, $\beta = \min\cbra{\rbra{10P}^{1/\alpha}, 1/2}$, $\varepsilon_p \in (0, 1/2]$, $\delta_Q \in \rbra{0, 1}$, and query access to unitary operator $U_A$ that is a $\rbra{1, m, 0}$-block-encoding of $N$-dimensional Hermitian operator $A$, where $m = O\rbra{\log\rbra{N}}$. 
    \Ensure The quantum circuit description of $U_{p\rbra{A}}$ with query access to $U_A$. 

    \State Let $p\rbra{x}$ be a polynomial of degree $d_p = O\rbra*{\frac{1}{\delta_p}\log\rbra*{\frac{1}{\delta_p\varepsilon_p}}}$ (by \cref{lemma:poly-approx-power}) such that
    \begin{align*}
        & \forall x \in \sbra{0, \delta_p}, \quad \abs*{p\rbra{x}} \leq \frac 1 2 \rbra*{\frac{\delta_p}{2\beta}}^{\frac{\alpha-1}{2}}, \\
        & \forall x \in \sbra{\delta_p, \beta}, \quad \abs*{p\rbra{x}-\frac 1 4 \rbra*{\frac{x}{2\beta}}^{\frac{\alpha-1}{2}}} \leq \varepsilon_p, \\
        & \forall x \in \sbra{-1, 1}, \quad \abs*{p\rbra{x}} \leq \frac 1 2.
    \end{align*}

    \State Construct unitary operator $U_{p\rbra{A}}$ (by \cref{thm:qsvt}) that is a $\rbra{1, a, \delta_Q}$-block-encoding of $p\rbra{A}$, where $a = m + 2$, using $O\rbra{d_p}$ queries to $U_A$. 

    \State \Return $U_{p\rbra{A}}$.
    \end{algorithmic}
\end{algorithm}

\begin{lemma}
\label{lemma:renyi-large-subroutine}
Suppose $\alpha > 1$ is a constant. 
For every $0 < \delta_p \leq \beta \leq 1/2$, $\varepsilon_p \in (0, 1/2]$, and $\delta_Q \in \rbra{0, 1}$,
\cref{algo:sub-renyi-gt1} will output the quantum circuit description of $
U_{p\rbra{A}}$ in classical time $\poly\rbra{1/\delta_p, \log\rbra{1/\varepsilon_p}, \log\rbra{1/\delta_Q}}$, and $U_{p\rbra{A}}$ makes $Q$ queries to $U_A$ and $O\rbra{Q\log\rbra{N}}$ one- and two-qubit gates, where $Q = O\rbra*{\frac{1}{\delta_p}\log\rbra*{\frac{1}{\delta_p\varepsilon_p}}}$.

Moreover, if $A = {\rho/2}$ for an $N$-dimensional quantum state $\rho$ of rank $r$ with a promise that $P \leq P_\alpha\rbra{\rho} \leq 10P$ for some $P > 0$ and $\beta = \min\cbra{\rbra*{10 P}^{1/\alpha}, 1/2}$, then measuring the quantum state $\sigma = U_{p\rbra{{\rho/2}}}\rbra*{\rho \otimes \ket{0}\bra{0}^{\otimes a}}U_{p\rbra{{\rho/2}}}^\dag$ in the computational basis will obtain the outcome $\ket{0}^{\otimes a}$ with probability $p_a$ such that
\[
    \abs*{ p_a - \frac 1 {16} \rbra*{{4}\beta}^{1-\alpha} P_{\alpha}\rbra*{\rho} } \leq \frac 5 {{8}} \rbra*{2\beta}^{1-\alpha} r \delta_p^\alpha + 2\varepsilon_p + 2r\delta_Q.
\]
\end{lemma}

\begin{proof}
    Suppose $\rho = \sum_{i \in \sbra{N}} x_i \ket{\psi_i} \bra{\psi_i}$ is of rank $r$, where $x_i \geq 0$, $\sum_{i \in \sbra{N}} x_i = 1$, and $\cbra{\ket{\psi_i}}$ is an orthonormal basis.
    Without loss of generality, we assume that $\sum_{i \in \sbra{r}} x_i = 1$.
    With the promise that $P \leq P_{\alpha}\rbra{\rho} \leq 10P$, we have 
    \[
    x_i \leq \rbra*{\sum_{i \in \sbra{N}} x_i^{\alpha}}^{\frac{1}{\alpha}} =  \rbra*{P_{\alpha}\rbra{\rho}}^{\frac 1 \alpha} \leq \rbra*{10 P}^{\frac 1 \alpha}.
    \]
    Let $\beta = \min\cbra{\rbra*{10 P}^{1/\alpha}, 1/2}$, and define
    \[
    f\rbra{x} = \frac 1 4 \rbra*{\frac{x}{2\beta}}^{\frac{\alpha-1}{2}}.
    \]
    By \cref{lemma:poly-approx-power}, we can choose a polynomial $p\rbra{x} \in \mathbb{R}\sbra{x}$ of degree $d_p = O\rbra*{\frac{1}{\delta_p}\log\rbra*{\frac{1}{\delta_p\varepsilon_p}}}$ such that 
    \begin{align*}
        & \forall x \in \sbra{0, \delta_p}, \quad \abs*{p\rbra{x}} \leq 2f\rbra{\delta_p}, \\
        & \forall x \in \sbra{\delta_p, \beta}, \quad \abs*{p\rbra{x}-f\rbra{x}} \leq \varepsilon_p, \\
        & \forall x \in \sbra{-1, 1}, \quad \abs*{p\rbra{x}} \leq \frac 1 2.
    \end{align*}
    Suppose $U_A$ is a unitary operator and is a $\rbra{1, m, 0}$-block-encoding of $A$ with $m = O\rbra{\log\rbra{N}}$. 
    By \cref{thm:qsvt}, we can construct a unitary operator $U_{p\rbra{A}}$, that is a $\rbra{1, m+2, \delta_Q}$-block-encoding of $p\rbra{A}$, using $O\rbra{d_p}$ queries to $U_{A}$, where $\delta_Q \in \rbra{0, 1}$ is a parameter to be determined. 
    It should be noted that the description of $U_{p\rbra{A}}$ can be computed classically in $\poly\rbra{d_p, \log\rbra{1/\delta_Q}}$ time.
    To make it clearer, let $a = m+2$, then 
    \begin{equation}
    \label{eq:Uprho-deltaQ}
    \Abs*{\bra{0}_a U_{p\rbra{A}} \ket{0}_a - p\rbra{A}} \leq \delta_Q. 
    \end{equation}

    Now we consider the special case that $A = {\rho/2}$. Applying $U_{p\rbra{\rho}}$ on quantum state $\rho \otimes \ket{0}_a \bra{0}_a$, we obtain a quantum state
    \[
    \sigma = U_{p\rbra{{\rho/2}}} \rbra*{\rho \otimes \ket{0}_a \bra{0}_a} U_{p\rbra{{\rho/2}}}^\dag.
    \]
    Note that the quantum state $\sigma$ is a $\rbra{1, a, 0}$-block-encoding of $\rbra{\bra{0}_a U_{p\rbra{{\rho/2}}} \ket{0}_a} \rho \rbra{\bra{0}_a U_{p\rbra{{\rho/2}}} \ket{0}_a}^\dag$.
    If we measure $\sigma$ in the computational basis, then the probability of measurement outcome $\ket{0}_a$ is 
    \begin{equation}
    \label{eq:large-def-pa}
    p_a = \tr\rbra*{\ket{0}_a \bra{0}_a \sigma} = \tr\rbra*{\rbra{\bra{0}_a U_{p\rbra{{\rho/2}}} \ket{0}_a} \rho \rbra{\bra{0}_a U_{p\rbra{{\rho/2}}} \ket{0}_a}^\dag}.
    \end{equation}
    We define our algorithm to output $1$ if the measurement outcome is $\ket{0}_a$, and $0$ otherwise.

    Now we are going to show the relationship between $p_a$ and $P_\alpha\rbra{\rho}$. We first note that
    \begin{equation}
    \label{eq:rhofrho}
    \tr\rbra*{\rho f\rbra*{{\frac{\rho}{2}}}^2} = \frac 1 {16} \rbra*{{4}\beta}^{1-\alpha} P_\alpha\rbra*{\rho}.
    \end{equation}
    From \cref{eq:Uprho-deltaQ}, we have
    \begin{equation}
    \label{eq:pa-rhoprho2-norm}
    \Abs*{\rbra{\bra{0}_a U_{p\rbra{{\rho/2}}} \ket{0}_a} \rho \rbra{\bra{0}_a U_{p\rbra{{\rho/2}}} \ket{0}_a}^\dag - \rho p\rbra*{{\frac{\rho}{2}}}^2} \leq 2\delta_Q,
    \end{equation}
    which means that $\sigma$ is a $\rbra{1, a, 2\delta_Q}$-block-encoding of $\rho p\rbra{{\rho/2}}^2$.
    From \cref{eq:large-def-pa} and \cref{eq:pa-rhoprho2-norm}, we have
    \begin{equation}
    \label{eq:pa-rhoprho2-abs}
        \abs*{ p_a - \tr\rbra*{\rho p\rbra*{{\frac{\rho}{2}}}^2} } \leq 2r\delta_Q.
    \end{equation}
    On the other hand, we have (see \cref{lemma:technical-1})
    \begin{equation}
    \label{eq:p-f-error}
    \abs*{ \tr\rbra*{\rho p\rbra*{{\frac{\rho}{2}}}^2} - \tr\rbra*{\rho f\rbra*{{\frac{\rho}{2}}}^2} } \leq \frac 5 {{8}} \rbra*{2\beta}^{1-\alpha} r \delta_p^\alpha + 2\varepsilon_p.
    \end{equation}
    From \cref{eq:rhofrho}, \cref{eq:pa-rhoprho2-abs}, and \cref{eq:p-f-error}, we have
    \[
        \abs*{ p_a - \frac 1 {16} \rbra*{{4}\beta}^{1-\alpha} P_{\alpha}\rbra*{\rho} } \leq \frac 5 {{8}} \rbra*{2\beta}^{1-\alpha} r \delta_p^\alpha + 2\varepsilon_p + 2r\delta_Q.
    \]
\end{proof}

To complete the proof of \cref{lemma:renyi-large-subroutine}, it remains to show the following technical lemma.

\begin{lemma}
\label{lemma:technical-1}
    In the proof of \cref{lemma:renyi-large-subroutine}, we have
    \[
    \abs*{ \tr\rbra*{\rho p\rbra*{{\frac{\rho}{2}}}^2} - \tr\rbra*{\rho f\rbra*{{\frac{\rho}{2}}}^2} } \leq \frac 5 {{8}} \rbra*{2\beta}^{1-\alpha} r \delta_p^\alpha + 2\varepsilon_p.
    \]
\end{lemma}

\begin{proof}
    To bound the error, we split it into two terms. 
\begin{align*}
    & \abs*{ \tr\rbra*{\rho p\rbra*{{\frac{\rho}{2}}}^2} - \tr\rbra*{\rho f\rbra*{{\frac{\rho}{2}}}^2} }
    = \abs*{ \sum_{i \in \sbra*{r}} \rbra*{ x_i p\rbra*{{\frac{x_i}{2}}}^2 - x_i f\rbra*{{\frac{x_i}{2}}}^2 } } \\
    & \qquad \leq \sum_{i \in \sbra*{r} \colon x_i \leq {2}\delta_p} \abs*{ x_i p\rbra*{{\frac{x_i}{2}}}^2 - x_i f\rbra*{{\frac{x_i}{2}}}^2 } + \sum_{i \in \sbra*{r} \colon x_i > {2}\delta_p} \abs*{ x_i p\rbra*{{\frac{x_i}{2}}}^2 - x_i f\rbra*{{\frac{x_i}{2}}}^2 }.
\end{align*}
For the first term, we only consider $x_i$'s such that $x_i \leq {2}\delta_p$. In this case, $\abs{p\rbra{{x_i/2}}} \leq 2f\rbra{\delta_p}$. Thus, we have
\begin{align*}
    \sum_{i \in \sbra*{r} \colon x_i \leq {2}\delta_p} \abs*{ x_i p\rbra*{{\frac{x_i}{2}}}^2 - x_i f\rbra*{{\frac{x_i}{2}}}^2 }
    & \leq \sum_{i \in \sbra*{r} \colon x_i \leq {2}\delta_p} \abs{x_i} \rbra*{ p\rbra*{{\frac{x_i}{2}}}^2 + f\rbra*{{\frac{x_i}{2}}}^2 } \\
    & \leq \sum_{i \in \sbra*{r} \colon x_i \leq {2}\delta_p} {2} \delta_p \rbra*{ 4 f\rbra*{\delta_p}^2 + f\rbra*{\delta_p}^2 } \\
    & \leq {10} r \delta_p f\rbra{\delta_p}^2 \\
    & = \frac 5 {{8}} \rbra*{2\beta}^{1-\alpha} r \delta_p^\alpha.
\end{align*}
For the last term, we only consider $x_i$'s such that $x_i > {2}\delta_p$. In this case, 
\begin{align*}
    \sum_{i \in \sbra*{r} \colon x_i > {2}\delta_p} \abs*{ x_i p\rbra*{{\frac{x_i}{2}}}^2 - x_i f\rbra*{{\frac{x_i}{2}}}^2 }
    & = \sum_{i \in \sbra*{r} \colon x_i > {2}\delta_p} \abs*{x_i} \abs*{p\rbra*{{\frac{x_i}{2}}} + f\rbra*{{\frac{x_i}{2}}}} \abs*{p\rbra*{{\frac{x_i}{2}}} - f\rbra*{{\frac{x_i}{2}}}} \\
    & \leq \sum_{i \in \sbra*{r} \colon x_i > {2}\delta_p} 2\abs*{x_i} \varepsilon_p \\
    & \leq \sum_{i \in \sbra*{r}} 2 x_i \varepsilon_p \\
    & = 2\varepsilon_p.
\end{align*}
Combining both cases, we have
\[
\abs*{ \tr\rbra*{\rho p\rbra*{\frac{\rho}{2}}^2} - \tr\rbra*{\rho f\rbra*{\frac{\rho}{2}}^2} } \leq \frac 5 {{8}} \rbra*{2\beta}^{1-\alpha} r \delta_p^\alpha + 2\varepsilon_p.
\]
\end{proof}

\subsubsection{Sample access}

We state the function $\texttt{estimate\_R\'enyi\_gt1\_promise}\rbra{\alpha, P, \varepsilon, \delta}$ in \cref{algo:renyi-main-intro} again in \cref{algo:esimate-large-promise} with detailed constraints. 

\begin{algorithm}[!htp]
    \caption{$\texttt{estimate\_R\'enyi\_gt1\_promise}\rbra{\alpha, P, \varepsilon, \delta}$ --- \textit{quantum sample algorithm}}
    \label{algo:esimate-large-promise}
    \begin{algorithmic}[1]
    \Require $\alpha > 1$, $\varepsilon \in \rbra{0, 1}$, $\delta \in \rbra{0, 1}$, access to $N$-dimensional quantum state $\rho$ of rank $r$, promise that $P \leq P_{\alpha}\rbra{\rho} \leq 10P$, and $\texttt{R\'enyi\_gt1\_subroutine}\rbra{\alpha, P, \delta_p, \varepsilon_p, \delta_Q}$ defined by \cref{algo:sub-renyi-gt1}.
    \Ensure $\widetilde P$ such that $\rbra{1 - \varepsilon} P_{\alpha}\rbra{\rho} \leq \widetilde P \leq \rbra{1 + \varepsilon} P_{\alpha}\rbra{\rho}$ with probability $\geq 1 - \delta$. 

    \State $\beta \gets \min\cbra{\rbra{10P}^{1/\alpha}, 1/2}$, $m \gets \ceil{8\ln\rbra{1/\delta}}$, $\delta_p \gets {\frac{1}{2}} \rbra*{\frac{P\varepsilon}{40r}}^{1/\alpha}$.

    \State $\varepsilon_p \gets \frac{\rbra{{4}\beta}^{1-\alpha}P\varepsilon}{256}$, $\delta_Q \gets \frac{\rbra{{4}\beta}^{1-\alpha}P\varepsilon}{128r}$, $\delta_a \gets \frac{\rbra{{4}\beta}^{1-\alpha}P\varepsilon}{128}$, and $k \gets \ceil*{\frac{65536}{\rbra{{4}\beta}^{1-\alpha}P\varepsilon^2}}$.

    \For {$j = 1 \dots m$}
        \For {$i = 1 \dots k$}
            \State Prepare $\sigma = \mathsf{Samplize}_{\delta_a}\ave{\texttt{R\'enyi\_gt1\_subroutine}\rbra{\alpha, P, \delta_p, \varepsilon_p, \delta_Q}}\sbra{\rho}\rbra{\rho \otimes \ket{0}\bra{0}^{\otimes a}}$.
            \State Measure $\sigma$ in the computational basis. 
            \State Let $X_i$ be $1$ if the outcome is $\ket{0}^{\otimes a}$, and $0$ otherwise. 
        \EndFor
        \State $\hat P_j \gets 16\rbra{{4}\beta}^{\alpha - 1} \sum_{i \in \sbra{k}} X_i$.
    \EndFor

    \State $\widetilde P \gets $ the median of $\hat P_j$ for $j \in \sbra{m}$. 

    \State \Return $\widetilde P$.
    
    \end{algorithmic}
\end{algorithm}

\begin{lemma}
\label{lemma:renyi-gt1-promise}
    Suppose $\alpha > 1$ is a constant, and $\rho$ is an $N$-dimensional quantum state of rank $r$ with a promise that $P \leq P_\alpha\rbra{\rho} \leq 10P$ for some $P > 0$.
    Then, \cref{algo:esimate-large-promise} with sample access to $\rho$, with probability $\geq 1 - \delta$, outputs an estimate $\widetilde P$ of $P_\alpha\rbra{\rho}$ such that $\rbra{1-\varepsilon}P_\alpha\rbra{\rho} \leq \widetilde P \leq \rbra{1+\varepsilon} P_\alpha\rbra{\rho}$, with sample complexity $M$ and time complexity $T\rbra{\alpha, P, \varepsilon, \delta} = O\rbra{M\log\rbra{N}}$, where
    \[
    M = O\rbra*{\frac{r^{\frac{2}{\alpha}}}{P^{\frac{4}{\alpha}}\varepsilon^{3+\frac{2}{\alpha}}} \log^4\rbra*{\frac{r}{P\varepsilon}} \log\rbra*{\frac 1 \delta}}.
    \]
\end{lemma}

To analyze the correctness of the algorithms, we need Chebyshev's inequality as follows. 

\begin{theorem} [Chebyshev's inequality, {\cite[Page 233]{Fel68}}]
\label{thm:chebyshev}
    Suppose $X$ is a real random variable with expectation $\mathbb{E}\sbra{X}$ and variance $\operatorname{Var}\sbra{X}$. Then, for every $\varepsilon > 0$, 
    \[
    \Pr \sbra*{ \abs*{X - \mathbb{E}\sbra*{X}} \geq \varepsilon } \leq \frac{\operatorname{Var}\sbra*{X}}{\varepsilon^2}.
    \]
\end{theorem}

Now we are ready to prove \cref{lemma:renyi-gt1-promise}. 

\begin{proof}[Proof of \cref{lemma:renyi-gt1-promise}]
    Let $\beta = \min\cbra{\rbra*{10 P}^{1/\alpha}, 1/2}$. 
    We take $\delta_p \in (0, \beta]$, $\varepsilon_p \in (0, 1/2]$, and $\delta_Q \in \rbra{0, 1}$ to be determined. By \cref{lemma:renyi-large-subroutine}, measuring the quantum state 
    \[
    \sigma = \texttt{R\'enyi\_gt1\_subroutine}\rbra{\alpha, P, \delta_p, \varepsilon_p, \delta_Q}\sbra{U_\rho}\rbra{\rho \otimes \ket{0}_a\bra{0}_a}
    \]
    in the computational basis will obtain the outcome $\ket{0}_a$ with probability $p_a$ such that
    \begin{equation}
    \label{eq:pa-Pa}
        \abs*{ p_a - \frac 1 {16} \rbra*{{4}\beta}^{1-\alpha} P_{\alpha}\rbra*{\rho} } \leq \frac 5 {{8}} \rbra*{2\beta}^{1-\alpha} r \delta_p^\alpha + 2\varepsilon_p + 2r\delta_Q,
    \end{equation}
    with query complexity $Q$ and time complexity $O\rbra{Q\log\rbra{N}}$, where $Q = O\rbra*{\frac{1}{\delta_p}\log\rbra*{\frac{1}{\delta_p\varepsilon_p}}}$.

    Let $\delta_a > 0$ to be determined. 
    By \cref{lemma:block-encoding-to-sample}, there is a unitary operator $U_\rho$ that is a block-encoding of $\rho$ such that
    \[
    \begin{aligned}
    \big\|\mathsf{Samplize}_{\delta_a}\ave{&\texttt{R\'enyi\_gt1\_subroutine}\rbra{\alpha, P, \delta_p, \varepsilon_p, \delta_Q}}\sbra{\rho} - \\
    &\texttt{R\'enyi\_gt1\_subroutine}\rbra{\alpha, P, \delta_p, \varepsilon_p, \delta_Q}\sbra{U_\rho}\big\|_{\diamond} \leq \delta_a.
    \end{aligned}
    \]
    This means that measuring 
    \[
    \widetilde\sigma = \mathsf{Samplize}_{\delta_a}\ave{\texttt{R\'enyi\_gt1\_subroutine}\rbra{\alpha, P, \delta_p, \varepsilon_p, \delta_Q}}\sbra{\rho}\rbra{\rho \otimes \ket{0}_a\bra{0}_a}
    \]
    in the computational basis will yield the outcome $\ket{0}_a$ with probability $\widetilde p_a$, where
    \begin{equation}
    \label{eq:pa-tildepa}
        \abs*{p_a - \widetilde p_a} \leq \delta_a,
    \end{equation}
    with sample complexity 
    \[
    1 + O\rbra*{\frac{Q^2}{\delta_a} \log^2\rbra*{\frac{Q}{\delta_a}}} = O\rbra*{\frac{1}{\delta_a \delta_p^2} \log^2\rbra*{\frac{1}{\delta_p\varepsilon_p}} \log^2\rbra*{\frac{1}{\delta_a\delta_p}\log\rbra*{\frac{1}{\delta_p\varepsilon_p}}}}.
    \]
    We call by algorithm $\mathcal{A}$ the above procedure that outputs $1$ with probability $\widetilde p_a$.

    We repeat algorithm $\mathcal{A}$ for $k$ times and let $X_i \in \cbra{0, 1}$ be the output of the $i$-th repetition. Let
    \begin{equation}
    \label{eq:tildeP-X}
    \hat P = 16\rbra*{{4}\beta}^{\alpha-1} X,
    \end{equation}
    where
    \[
    X = \frac 1 k \sum_{i \in \sbra{k}} X_i.
    \]
    We note that $\mathbb{E}\sbra{X} = \mathbb{E}\sbra{X_i} = \widetilde p_a$, and $\operatorname{Var}\sbra{X} = \operatorname{Var}\sbra{X_i}/k = \rbra{\widetilde p_a - \widetilde p_a^2}/k$.
    By Chebyshev's inequality (\cref{thm:chebyshev}), for every $\varepsilon_C > 0$, 
    \begin{equation}
    \label{eq:X-tildepa}
        \Pr \rbra*{ \abs*{X - \widetilde p_a} \geq \varepsilon_C } \leq \frac{\widetilde p_a - \widetilde p_a^2}{k \varepsilon_C^2}.
    \end{equation}
    From \cref{eq:pa-Pa}, \cref{eq:pa-tildepa}, \cref{eq:tildeP-X}, and \cref{eq:X-tildepa} we know that with probability $\geq 1 - \rbra{\widetilde p_a - \widetilde p_a^2}/{k \varepsilon_C^2}$,
    \[
    \abs*{\hat P - P_{\alpha}\rbra*{\rho}} \leq 5r\rbra{{2}\delta_p}^{\alpha} + 16\rbra*{{4}\beta}^{\alpha-1}\rbra*{2\varepsilon_p + 2r\delta_Q + \delta_a + \varepsilon_C}.
    \]
    By taking $\delta_p = {\frac{1}{2}} \rbra*{\frac{P\varepsilon}{40r}}^{1/\alpha}$, $\varepsilon_p = \frac{\rbra{{4}\beta}^{1-\alpha}P\varepsilon}{256}$, $\delta_Q = \frac{\rbra{{4}\beta}^{1-\alpha}P\varepsilon}{128r}$, $\delta_a = \varepsilon_C = \frac{\rbra{{4}\beta}^{1-\alpha}P\varepsilon}{128}$, and $k = \frac{65536}{\rbra{{4}\beta}^{1-\alpha}P\varepsilon^2}$, we have that, with probability $\geq 3/4$, it holds that $\abs*{\hat P - P_\alpha\rbra{\rho}} \leq P\varepsilon$, which implies
    \[
        \rbra*{1 - \varepsilon} P_{\alpha}\rbra{\rho} \leq \hat P \leq \rbra*{1 + \varepsilon} P_{\alpha}\rbra{\rho}
    \]
    by noting that $P \leq P_{\alpha}\rbra{\rho} \leq 10P$. In other words, by choosing the values of $\delta_p, \varepsilon_p, \delta_Q, \delta_a, \varepsilon_C, k$ as above, we have
    \[
    \Pr\sbra*{ \rbra*{1 - \varepsilon} P_{\alpha}\rbra{\rho} \leq \hat P \leq \rbra*{1 + \varepsilon} P_{\alpha}\rbra{\rho} } \geq \frac{3}{4}.
    \]
    We denote the algorithm described above  as $\mathcal{B}$. 

    Finally, we need to amplify the success probability from $3/4$ to $1-\delta$. To achieve this, we repeat algorithm $\mathcal{B}$ for $m = \ceil{8\ln\rbra{1/\delta}}$ times, and write $\hat P_i$ to be the $i$-th estimate for $i \in \sbra{m}$. 
    Let $\widetilde P$ be the median of all $\hat P_i$ for $i \in \sbra{m}$, then it holds that
    \[
    \Pr\sbra*{ \rbra*{1 - \varepsilon} P_{\alpha}\rbra{\rho} \leq \widetilde P \leq \rbra*{1 + \varepsilon} P_{\alpha}\rbra{\rho} } \geq 1 - \delta.
    \]
    To see this, we define $Z_i$ to be $0$ if $\rbra{1-\varepsilon} P_{\alpha}\rbra{\rho} \leq \hat P_i \leq \rbra{1+\varepsilon} P_{\alpha}\rbra{\rho}$ and $1$ otherwise. 
    Then, $\mathbb{E}\sbra{Z_i} \leq 1/4$. Let $Z = \sum_{i \in \sbra{m}} Z_i$, and note that $\mathbb{E}\sbra{Z} \leq m/4$. 
    The median $\widetilde P$ of $\hat P_i$ satisfies $\rbra{1-\varepsilon} P_{\alpha}\rbra{\rho} \leq \widetilde P \leq \rbra{1+\varepsilon} P_{\alpha}\rbra{\rho}$ with probability (by \cref{thm:hoeffding})
    \[
    \geq 1 - \Pr\sbra*{ Z \geq \frac m 2 } \geq 1 - \Pr \sbra*{ Z - \mathbb{E}\sbra*{Z} \geq \frac m 4 } \geq 1 - \exp\rbra*{\frac{m}{8}} \geq 1 - \delta.
    \]

    Now we are going to analyze the complexity. We repeat algorithm $\mathcal{B}$ for $m$ times, and each repetition of $\mathcal{B}$ repeats algorithm $\mathcal{A}$ for $k$ times.
    Thus, the total sample complexity is
    \[
    mk \cdot O\rbra*{\frac{1}{\delta_a \delta_p^2} \log^2\rbra*{\frac{1}{\delta_p\varepsilon_p}} \log^2\rbra*{\frac{1}{\delta_a\delta_p}\log\rbra*{\frac{1}{\delta_p\varepsilon_p}}}} = O\rbra*{\frac{r^{\frac{2}{\alpha}}}{P^{\frac{4}{\alpha}}\varepsilon^{3+\frac{2}{\alpha}}} \log^4\rbra*{\frac{r}{P\varepsilon}} \log\rbra*{\frac 1 \delta}},
    \]
    and the time complexity is only a $\log\rbra{N}$ factor over the sample complexity, which is
    \[
    O\rbra*{\frac{r^{\frac{2}{\alpha}}}{P^{\frac{4}{\alpha}}\varepsilon^{3+\frac{2}{\alpha}}} \log^4\rbra*{\frac{r}{P\varepsilon}} \log\rbra*{\frac 1 \delta} \log\rbra{N}}.
    \]
\end{proof}

Finally, we will use the estimate of $\widetilde P$ of $P_\alpha\rbra{\rho}$ with multiplicative error to estimate $S_\alpha\rbra{\rho}$ with additive error. 
This process is simple and stated in \cref{algo:esimate-gt1}, and we include it here for completeness. 

\begin{algorithm}[!htp]
    \caption{$\texttt{estimate\_R\'enyi\_gt1\_main}\rbra{\alpha, \varepsilon, \delta}$ --- \textit{quantum sample algorithm}}
    \label{algo:esimate-gt1}
    \begin{algorithmic}[1]
    \Require $\alpha > 1$, $\varepsilon \in \rbra{0, 1}$, $\delta \in \rbra{0, 1}$, sample access to quantum state $\rho$ of rank $r$, and $\texttt{estimate\_R\'enyi\_gt1}\rbra{\alpha, \varepsilon, \delta}$ defined by \cref{algo:esimate-large}.
    \Ensure $\widetilde S$ such that $\abs{\widetilde S - S_{\alpha}\rbra{\rho}} \leq \varepsilon$ with probability $\geq 1 - \delta$. 

    \State $\widetilde P \gets \texttt{estimate\_R\'enyi\_gt1}\rbra{\alpha, \rbra{\alpha-1}\varepsilon/2, \delta}$.

    \State $\widetilde S \gets \frac{1}{1-\alpha}\ln\rbra{\widetilde P}$.
    
    \State \Return $\widetilde S$.
    \end{algorithmic}
\end{algorithm}

\begin{theorem}
\label{thm:estimate-renyi-gt1}
    Suppose $\alpha > 1$ is a constant, and $N$-dimensional quantum state $\rho$ is of rank $r$.
    Then, \cref{algo:esimate-gt1}, with probability $\geq 1 - \delta$, outputs an estimate $\widetilde S$ of $S_\alpha\rbra{\rho}$ within additive error $\varepsilon$, using $M$ samples of $\rho$ and $O\rbra{M \log\rbra{N}}$ one- and two-qubit quantum gates, where
    \[
    M = O\rbra*{r^{4-\frac{2}{\alpha}} \rbra*{\log^6\rbra*{r} + \frac{1}{\varepsilon^{3+\frac{2}{\alpha}} } \log^4\rbra*{\frac{r}{\varepsilon}}} \log\rbra*{\frac 1 \delta}}.
    \]
\end{theorem}

\begin{proof}
    By \cref{lemma:renyi-gt1-promise}, we can implement the procedure $\texttt{estimate\_R\'enyi\_gt1\_promise}\rbra{\alpha, \varepsilon, P, \delta}$ required in \cref{algo:esimate-large}.
    Then, by \cref{lemma:annealing}, \cref{algo:esimate-large} returns an estimate $\widetilde P$ of $P_{\alpha}\rbra{\rho}$ such that
    \[
    \rbra{1 - \epsilon} P_\alpha\rbra{\rho} \leq \widetilde P \leq \rbra{1 + \epsilon} P_\alpha\rbra{\rho}
    \]
    with probability $\geq 1 - \delta$.
    By taking $\widetilde S = \frac{1}{1-\alpha} \ln \rbra{\widetilde P}$, we have that
    \[
    S_\alpha\rbra{\rho} - \frac{\ln\rbra{1 + \epsilon}}{\alpha - 1} \leq \widetilde S \leq S_\alpha\rbra{\rho} + \frac{\ln\rbra{1 - \epsilon}}{1 - \alpha}.
    \]
    By letting $\epsilon = \rbra{\alpha - 1}\varepsilon/2$, we have $S_\alpha\rbra{\rho} - \varepsilon \leq \widetilde S \leq S_\alpha\rbra{\rho} + \varepsilon$, as required.

    Now we are going to analyze the complexity. 
    Let $\lambda = 1 + 1/\ln\rbra{r}$. 
    Combining \cref{lemma:annealing} and \cref{lemma:renyi-gt1-promise}, we have the recurrence relation:
    if $1 < \alpha \leq \lambda$, then 
    \[
    Q\rbra*{\alpha, \varepsilon, \delta} = O \rbra*{ \frac{ r^{\frac{2}{\alpha}} } { \varepsilon^{3+\frac{2}{\alpha}} } \log^4\rbra*{\frac{r}{\varepsilon}} \log\rbra*{\frac 1 \delta} \log\rbra*{N} };
    \]
    and if $\alpha > \lambda$, then
    \begin{align*}
        Q\rbra*{\alpha, \varepsilon, \delta} 
        = Q\rbra*{\frac{\alpha}{\lambda}, \frac{1}{4}, \frac{\delta}{2}} + O \rbra*{ \frac{ r^{4-\frac{2}{\alpha}} } { \varepsilon^{3+\frac{2}{\alpha}} } \log^4\rbra*{\frac{r}{\varepsilon}} \log\rbra*{\frac 1 \delta} \log\rbra*{N} }.
    \end{align*}
    For the special case of $\varepsilon = 1/4$, we have: if $1 < \alpha \leq \lambda$, then
    \[
    Q\rbra*{\alpha, \frac{1}{4}, \delta} = O \rbra*{ r^{\frac{2}{\alpha}} \log^4\rbra*{r} \log\rbra*{\frac 1 \delta} \log\rbra*{N} };
    \]
    and if $\alpha > \lambda$, then
    \begin{align*}
        Q\rbra*{\alpha, \frac{1}{4}, \delta} = Q\rbra*{\frac{\alpha}{\lambda}, \frac{1}{4}, \frac{\delta}{2}} + O \rbra*{ r^{4 - \frac{2}{\alpha}} \log^4\rbra*{r} \log\rbra*{\frac 1 \delta} \log\rbra*{N} },
    \end{align*}
    where both cases satisfy that
    \begin{align*}
        Q\rbra*{\alpha, \frac{1}{4}, \delta} 
        & \leq \sum_{k=0}^{\floor{\log_\lambda\rbra{\alpha}}} O\rbra*{  r^{4-\frac{2}{\alpha}} \log^4\rbra*{r} \log\rbra*{\frac {2^k} \delta} \log\rbra*{N} } \\
        & = O\rbra*{ r^{4-\frac{2}{\alpha}} \log^6\rbra*{r} \log\rbra*{\frac 1 \delta} \log\rbra*{N} }.
    \end{align*}
    
    Finally, we have
    \[
    Q\rbra{\alpha, \varepsilon, \delta} = O\rbra*{ r^{4-\frac{2}{\alpha}} \rbra*{\log^6\rbra*{r} + \frac{1}{\varepsilon^{3+\frac{2}{\alpha}} } \log^4\rbra*{\frac{r}{\varepsilon}}} \log\rbra*{\frac 1 \delta} \log\rbra*{N} }.
    \]
\end{proof}

\subsection{The case of \texorpdfstring{$0 < \alpha < 1$}{0 < α < 1}} \label{sec:alpha-lt1}

We first provide an overview of our quantum algorithm (see \cref{algo:renyi-lt1-main-intro}) for estimating $\alpha$-R\'enyi entropy for $0 < \alpha < 1$ as that for $\alpha > 1$ in \cref{algo:renyi-main-intro}.
Then, we will analyze it in details.

We first explain two functions (similar to the case of $\alpha > 1$):

\begin{itemize}
    \item $\texttt{estimate\_R\'enyi\_lt1}\rbra{\alpha, \varepsilon, \delta}$: return an estimate $\widetilde P$ such that $\rbra{1-\varepsilon}P_\alpha\rbra{\rho} \leq \widetilde P \leq \rbra{1+\varepsilon}P_\alpha\rbra{\rho}$ with probability $\geq 1 - \delta$.
    \item $\texttt{estimate\_R\'enyi\_lt1\_promise}\rbra{\alpha, P, \varepsilon, \delta}$: return an estimate $\widetilde P$ such that $\rbra{1-\varepsilon}P_\alpha\rbra{\rho} \leq \widetilde P \leq \rbra{1+\varepsilon}P_\alpha\rbra{\rho}$ with probability $\geq 1 - \delta$, with the promise that $P \leq P_\alpha\rbra{\rho} \leq 10P$.
\end{itemize}

The function $\texttt{estimate\_R\'enyi\_lt1}\rbra{\alpha, \varepsilon, \delta}$ recursively estimates $P_\alpha\rbra{\rho}$ by reducing it to a special case of estimating $P_\alpha\rbra{\rho}$ with a promise that $P \leq P_\alpha\rbra{\rho} \leq 10P$ for some given $P > 0$, where the promised problem is solved by $\texttt{estimate\_R\'enyi\_lt1\_promise}\rbra{\alpha, P, \varepsilon, \delta}$ (see \cref{sec:recursive-lt1} for the correctness of the reduction).

\subsubsection{Recursive framework} \label{sec:recursive-lt1}

For the case of $0 < \alpha < 1$, we first provide an abstract algorithm (\cref{algo:esimate-small}) with a structure similar to the one (\cref{algo:esimate-large}) for $\alpha > 1$.
Here, $\texttt{\textup{estimate\_R\'enyi\_lt1\_promise}}\rbra{\alpha, P, \varepsilon, \delta}$ is a quantum algorithm that, with probability $\geq 1 - \delta$, outputs an estimate of $P_\alpha\rbra{\rho}$ such that $\rbra{1-\varepsilon}P_\alpha\rbra{\rho} \leq \widetilde P \leq \rbra{1+\varepsilon} P_\alpha\rbra{\rho}$, given a promise that $P \leq P_\alpha\rbra{\rho} \leq 10P$, where $P$ is the prior knowledge of $P_\alpha\rbra{\rho}$ given as input.

\begin{algorithm}[!htp]
    \caption{$\texttt{estimate\_R\'enyi\_lt1}\rbra{\alpha, \varepsilon, \delta}$ --- \textit{quantum sample access}}
    \label{algo:esimate-small}
    \begin{algorithmic}[1]
    \Require $0 < \alpha < 1$, $\varepsilon \in \rbra{0, 1}$, $\delta \in \rbra{0, 1}$, and access to $N$-dimensional quantum state $\rho$ of rank $r$.
    \Ensure $\widetilde P$ such that $\rbra{1 - \varepsilon} P_{\alpha}\rbra{\rho} \leq \widetilde P \leq \rbra{1 + \varepsilon} P_{\alpha}\rbra{\rho}$ with probability $\geq 1 - \delta$. 

    \State $\lambda \gets 1 - 1 / \ln \rbra{r}$.
    \If {$\alpha \geq \lambda$}
        \State $P \gets 1$.
    \Else
        \State $P' \gets \texttt{estimate\_R\'enyi\_lt1}\rbra{\alpha/\lambda, 1/4, \delta/2}$. \label{line:call-small}
        \State $P \gets \rbra{4P'/5}^\lambda$.
    \EndIf
    
    \State \Return $\texttt{estimate\_R\'enyi\_lt1\_promise}\rbra{\alpha, P, \varepsilon, \delta/2}$.
    \end{algorithmic}
\end{algorithm}

\begin{lemma}
\label{lemma:annealing-small}
    Suppose $0 < \alpha < 1$, and $\texttt{\textup{estimate\_R\'enyi\_lt1\_promise}}\rbra{\alpha, P, \varepsilon, \delta}$ is a quantum algorithm with time complexity $T\rbra{\alpha, P, \varepsilon, \delta}$ that, with probability $\geq 1 - \delta$, outputs an estimate of $P_\alpha\rbra{\rho}$ such that $\rbra{1-\varepsilon}P_\alpha\rbra{\rho} \leq \widetilde P \leq \rbra{1+\varepsilon} P_\alpha\rbra{\rho}$, given a promise that $P \leq P_\alpha\rbra{\rho} \leq 10P$. Then, with probability $\geq 1 - \delta$, \cref{algo:esimate-small} outputs an estimate $\widetilde P$ of $P_\alpha\rbra{\rho}$ such that $\rbra{1-\varepsilon}P_\alpha\rbra{\rho} \leq \widetilde P \leq \rbra{1+\varepsilon} P_\alpha\rbra{\rho}$ with time complexity $Q\rbra{\alpha, \varepsilon, \delta}$, where
    \[
    Q\rbra*{\alpha, \varepsilon, \delta} = \begin{cases}
        T\rbra*{\alpha, \varepsilon, 1, \delta/2}, & \lambda \leq \alpha < 1, \\
        Q\rbra*{\alpha/\lambda, 1/4, \delta/2} + \sup\limits_{P \in \sbra*{1/10, r^{1-\alpha}}} T\rbra{\alpha, \varepsilon, P, \delta/2}, & 0 < \alpha < \lambda,
    \end{cases}
    \]
    in which $\lambda = 1 - 1/\ln\rbra{r}$ and $r$ is the rank of $\rho$. 
\end{lemma}
\begin{proof}
    Suppose $\rho = \sum_{i \in \sbra{N}} x_i \ket{\psi_i} \bra{\psi_i}$, where $x_i \geq 0 $, $\sum_{i \in \sbra{N}} x_i = 1$, and $\cbra{\ket{\psi_i}}$ is an orthonormal basis.
    \begin{itemize}
        \item For the basis case that $\lambda \leq \alpha < 1$, 
        \[
        P_{\alpha}\rbra*{\rho} \leq r^{1-\alpha} \leq r^{1/\ln\rbra*{r}} = e.
        \]
        On the other hand, $P_\alpha\rbra{\rho} \geq 1$. These together yield that $P = 1 \leq P_{\alpha}\rbra{\rho} \leq e \leq 10 P$, which satisfies the required condition of $\texttt{estimate\_R\'enyi\_lt1\_promise}\rbra{\alpha, \varepsilon, P, \delta/2}$. Therefore, \cref{algo:esimate-small} will output an estimate of $P_\alpha\rbra{\rho}$ within multiplicative error $\varepsilon$ with probability $\geq 1 - \delta/2 \geq 1 - \delta$, with time complexity $Q\rbra{\alpha, \varepsilon, \delta} = T\rbra{\alpha, \varepsilon, P, \delta/2}$. 
        \item Now we consider the case that $\alpha < \lambda$. Let $\alpha' = \alpha/\lambda > \alpha$. By \cref{lemma:p-alpha-beta}, we have
        \[
        \rbra*{P_{\alpha'}\rbra{\rho}}^{\lambda} = \rbra*{P_{\alpha'}\rbra{\rho}}^{\frac{\alpha}{\alpha'}} \leq P_{\alpha}\rbra{\rho} \leq r^{1-\frac{\alpha}{\alpha'}} \rbra*{P_{\alpha'}\rbra{\rho}}^{\frac{\alpha}{\alpha'}} = r^{1-\lambda} \rbra*{P_{\alpha'}\rbra{\rho}}^{\lambda}.
        \]
        By induction, after calling $\texttt{estimate\_R\'enyi\_lt1}\rbra{\alpha', \varepsilon', \delta'}$ in Line \ref{line:call-small} of \cref{algo:esimate-small}, where $\varepsilon' = 1/4$ and $\delta' = \delta/2$, we will obtain $P'$ such that $\rbra{1-\varepsilon'}P_{\alpha'}\rbra{\rho} \leq P' \leq \rbra{1+\varepsilon'}P_{\alpha'}\rbra{\rho}$ with probability $\geq 1 - \delta/2$. Thus, we have
        \[
        P_{\alpha}\rbra*{\rho} \geq \rbra*{P_{\alpha'}\rbra*{\rho}}^{\lambda} \geq \rbra*{\frac{P'}{1+\varepsilon'}}^{\lambda} = \rbra*{\frac{4P'}{5}}^{\lambda} = P.
        \]
        On the other hand, 
        \[
        P_{\alpha}\rbra*{\rho} \leq r^{1-\lambda} \rbra*{P_{\alpha'}\rbra*{\rho}}^{\lambda} \leq e\rbra*{\frac{P'}{1-\varepsilon'}}^{\lambda} = e\rbra*{\frac{4P'}{3}}^{\lambda} \leq 10 P.
        \]
        Therefore, $P \leq P_{\alpha}\rbra{\rho} \leq 10P$ with probability $\geq 1 - \delta/2$, i.e., the required condition of 
        \[
        \texttt{estimate\_R\'enyi\_lt1\_promise}\rbra{\alpha, \varepsilon, P, \delta/2}
        \]
        holds with probability $\geq 1 - \delta/2$; then, \cref{algo:esimate-small} will output an estimate of $P_{\alpha}$ within multiplicative error $\varepsilon$ with probability $\geq \rbra{1-\delta/2}^2 \geq 1 - \delta$, with time complexity $Q\rbra{\alpha, \varepsilon, \delta} = Q\rbra{\alpha', \varepsilon', \delta'} + T\rbra{\alpha, \varepsilon, P, \delta/2}$. 
        Without loss of generality, we may assume that it always holds that $1/10 \leq P \leq r^{1-\alpha}$ during the execution of \cref{algo:esimate-large}. Thus, 
        \[
        Q\rbra*{\alpha, \varepsilon, \delta} = Q\rbra{\alpha', \varepsilon', \delta'} + \sup_{P \in \sbra*{1/10, r^{1-\alpha}}} T\rbra{\alpha, \varepsilon, P, \delta/2}.
        \]
    \end{itemize}
    From the above, we have
    \[
    Q\rbra*{\alpha, \varepsilon, \delta} = \begin{cases}
        T\rbra*{\alpha, \varepsilon, 1, \delta/2}, & \lambda \leq \alpha < 1, \\
        Q\rbra*{\alpha/\lambda, 1/4, \delta/2} + \sup\limits_{P \in \sbra*{1/10, r^{1-\alpha}}} T\rbra{\alpha, \varepsilon, P, \delta/2}, & 0 < \alpha < \lambda.
    \end{cases}
    \]
\end{proof}

\subsubsection{Subroutines with block-encoded access}

The function $\texttt{R\'enyi\_lt1\_subroutine}\rbra{\alpha, P, \delta_p, \varepsilon_p, \delta_Q}$ implemented in \cref{algo:renyi-lt1-main-intro} is with block-encoded access. 
We restate it in \cref{algo:sub-renyi-lt1} with detailed constraints. 

\begin{algorithm}[!htp]
    \caption{$\texttt{R\'enyi\_lt1\_subroutine}\rbra{\alpha, P, \delta_p, \varepsilon_p, \delta_Q}$ --- \textit{quantum query algorithm}}
    \label{algo:sub-renyi-lt1}
    \begin{algorithmic}[1]
    \Require $0 < \alpha < 1$, $\delta_p \in (0, 1/2)$, $\varepsilon_p \in (0, 1/2)$, $\delta_Q \in \rbra{0, 1}$, and query access to unitary operator $U_A$ that is a $\rbra{1, m, 0}$-block-encoding of $N$-dimensional Hermitian operator $A$, where $m = O\rbra{\log\rbra{N}}$. 
    \Ensure The quantum circuit description of $U_{p\rbra{A}}$ with query access to $U_A$. 

    \State Let $p\rbra{x}$ be a polynomial of degree $d_p = O\rbra*{\frac{1}{\delta_p}\log\rbra*{\frac{1}{\varepsilon_p}}}$ (by \cref{lemma:poly-approx-negative-power}) such that
    \begin{align*}
        & \forall x \in \sbra*{-1, 1}, \quad \abs*{p\rbra*{x}} \leq \frac 1 2, \\
        & \forall x \in \sbra*{\delta_p, 1}, \quad \abs*{p\rbra*{x} - \frac 1 4 \rbra*{\frac{x}{\delta_p}}^{\frac{\alpha-1}{2}}} \leq \varepsilon_p.
    \end{align*}

    \State Construct unitary operator $U_{p\rbra{A}}$ (by \cref{thm:qsvt}) that is a $\rbra{1, a, \delta_Q}$-block-encoding of $p\rbra{A}$, where $a = m + 2$, using $O\rbra{d_p}$ queries to $U_A$. 

    \State \Return $U_{p\rbra{A}}$.
    \end{algorithmic}
\end{algorithm}

\begin{lemma}
\label{lemma:renyi-lt1-block-encoded}
Suppose $0 < \alpha < 1$ is a constant. 
For every $\delta_p \in \rbra{0, 1/2}$, $\varepsilon_p \in \rbra{0, 1/2}$, and $\delta_Q \in \rbra{0, 1}$,
\cref{algo:sub-renyi-lt1} will output the quantum circuit description of $
U_{p\rbra{A}}$ in classical time $\poly\rbra{1/\delta_p, \log\rbra{1/\varepsilon_p}, \log\rbra{1/\delta_Q}}$, and $U_{p\rbra{A}}$ makes $Q$ queries to $U_A$ and $O\rbra{Q\log\rbra{N}}$ one- and two-qubit gates, where $Q = O\rbra*{\frac{1}{\delta_p}\log\rbra*{\frac{1}{\varepsilon_p}}}$.

Moreover, if $A = {\rho/2}$ for an $N$-dimensional quantum state $\rho$ of rank $r$ with a promise that $P \leq P_\alpha\rbra{\rho} \leq 10P$ for some $P > 0$, then measuring the quantum state $\sigma = U_{p\rbra{{\rho/2}}}\rbra*{\rho \otimes \ket{0}\bra{0}^{\otimes a}}U_{p\rbra{{\rho/2}}}^\dag$ in the computational basis will obtain the outcome $\ket{0}^{\otimes a}$ with probability $p_a$ such that
\[
\abs*{ p_a - \frac{1}{16} \rbra{{2}\delta_p}^{1-\alpha} P_{\alpha}\rbra*{\rho} } \leq \frac{5}{{8}} r \delta_p + 2 \varepsilon_p + 2r\delta_Q.
\]
\end{lemma}
\begin{proof}
    Let 
    \[
    f\rbra{x} = \frac 1 4 \rbra*{\frac{x}{\delta_p}}^{\frac{\alpha-1}{2}}. 
    \]
    By \cref{lemma:poly-approx-negative-power}, there is a polynomial $p\rbra{x}$ of degree $d_p = O\rbra*{\frac{1}{\delta_p}\log\rbra*{\frac{1}{\varepsilon_p}}}$ such that
    \begin{align*}
        & \forall x \in \sbra*{-1, 1}, \quad \abs*{p\rbra*{x}} \leq \frac 1 2, \\
        & \forall x \in \sbra*{\delta_p, 1}, \quad \abs*{p\rbra*{x} - f\rbra{x}} \leq \varepsilon_p.
    \end{align*}
    In the following, the construction is similar to that of \cref{algo:sub-renyi-gt1}. 
    Now suppose $U_A$ is a unitary operator and is a $\rbra{1, m, 0}$-block-encoding of $A$ with $m = O\rbra{\log\rbra{N}}$. 
    By \cref{thm:qsvt}, we can construct a unitary operator $U_{p\rbra{A}}$, that is a $\rbra{1, m+2, \delta_Q}$-block-encoding of $p\rbra{A}$, using $O\rbra{d_p}$ queries to $U$, where $\delta_Q \in \rbra{0, 1}$ is a parameter to be determined. 
    To make it clearer, let $a = m+2$, then $\Abs*{\bra{0}_a U_{p\rbra{A}} \ket{0}_a - p\rbra{A}} \leq \delta_Q$. 

    Now we consider the special case that $A = {\rho/2}$. Applying $U_{p\rbra{{\rho/2}}}$ on quantum state $\rho \otimes \ket{0}_a \bra{0}_a$, we obtain a quantum state $\sigma = U_{p\rbra{{\rho/2}}} \rbra*{\rho \otimes \ket{0}_a \bra{0}_a} U_{p\rbra{{\rho/2}}}^\dag$.
    Note that the quantum state $\sigma$ is a $\rbra{1, a, 0}$-block-encoding of $\rbra{\bra{0}_a U_{p\rbra{{\rho/2}}} \ket{0}_a} \rho \rbra{\bra{0}_a U_{p\rbra{{\rho/2}}} \ket{0}_a}^\dag$.
    If we measure $\sigma$ in the computational basis, then the probability of measurement outcome $\ket{0}_a$ is $p_a = \tr\rbra*{\ket{0}_a \bra{0}_a \sigma}$.

    Now we are going to show the relationship between $p_a$ and $P_\alpha\rbra{\rho}$. 
    We first note that
    \begin{equation}
    \label{eq:renyi-lt1-f-P-alpha}
    \tr\rbra*{\rho f\rbra*{{\frac{\rho}{2}}}^2} = \frac 1 {16} \rbra{{2}\delta_p}^{1-\alpha} P_\alpha\rbra*{\rho}.
    \end{equation}
    Similar to the proof of \cref{lemma:renyi-large-subroutine}, we have
    \begin{equation}
    \label{eq:renyi-lt1-pa-prho}
        \abs*{ p_a - \tr\rbra*{\rho p\rbra*{{\frac{\rho}{2}}}^2} } \leq 2r\delta_Q.
    \end{equation}
    On the other hand, we have (see \cref{lemma:renyi-tech-lt1})
    \begin{equation}
    \label{eq:renyi-lt1-p-f-error}
    \abs*{ \tr\rbra*{\rho p\rbra*{{\frac{\rho}{2}}}^2} - \tr\rbra*{\rho f\rbra*{{\frac{\rho}{2}}}^2} } \leq \frac 5 {{8}} r \delta_p + 2\varepsilon_p.
    \end{equation}
    From \cref{eq:renyi-lt1-f-P-alpha}, \cref{eq:renyi-lt1-pa-prho}, and \cref{eq:renyi-lt1-p-f-error}, we have
    \[
    \abs*{ p_a - \frac{1}{16} \rbra{{2}\delta_p}^{1-\alpha} P_{\alpha}\rbra*{\rho} } \leq \frac{5}{{8}} r \delta_p + 2 \varepsilon_p + 2r\delta_Q.
    \]
\end{proof}

To complete the proof of \cref{lemma:renyi-lt1-block-encoded}, it remains to show the following technical lemma.

\begin{lemma}
\label{lemma:renyi-tech-lt1}
    In the proof of \cref{lemma:renyi-lt1-block-encoded}, we have
    \[
    \abs*{ \tr\rbra*{\rho p\rbra*{{\frac{\rho}{2}}}^2} - \tr\rbra*{\rho f\rbra*{{\frac{\rho}{2}}}^2} } \leq \frac 5 {{8}} r \delta_p + 2\varepsilon_p.
    \]
\end{lemma}

\begin{proof}
    To bound the error, we split it into two terms. 
\begin{align*}
    & \abs*{ \tr\rbra*{\rho p\rbra*{{\frac{\rho}{2}}}^2} - \tr\rbra*{\rho f\rbra*{{\frac{\rho}{2}}}^2} }
    = \abs*{ \sum_{i \in \sbra*{r}} \rbra*{ x_i p\rbra*{{\frac{x_i}{2}}}^2 - x_i f\rbra*{{\frac{x_i}{2}}}^2 } } \\
    & \qquad \leq \sum_{i \in \sbra*{r} \colon x_i \leq {2}\delta_p} \abs*{ x_i p\rbra*{{\frac{x_i}{2}}}^2 - x_i f\rbra*{{\frac{x_i}{2}}}^2 } + \sum_{i \in \sbra*{r} \colon x_i > {2}\delta_p} \abs*{ x_i p\rbra*{{\frac{x_i}{2}}}^2 - x_i f\rbra*{{\frac{x_i}{2}}}^2 }.
\end{align*}
For the first term, we only consider $x_i$'s such that $x_i \leq {2}\delta_p$. In this case, $\abs{p\rbra{x_i}} \leq 1/2$. Thus, we have
\begin{align*}
    \sum_{i \in \sbra*{r} \colon x_i \leq {2}\delta_p} \abs*{ x_i p\rbra*{{\frac{x_i}{2}}}^2 - x_i f\rbra*{{\frac{x_i}{2}}}^2 }
    & \leq \sum_{i \in \sbra*{r} \colon x_i \leq {2}\delta_p} \rbra*{ x_i p\rbra*{{\frac{x_i}{2}}}^2 + x_i f\rbra*{{\frac{x_i}{2}}}^2 } \\
    & \leq \sum_{i \in \sbra*{r} \colon x_i \leq {2}\delta_p} \rbra*{ \frac 1 4 x_i + \frac 1 {16} \rbra{{2}\delta_p}^{1-\alpha} x_i^\alpha } \\
    & \leq \frac 5 {{8}} r \delta_p.
\end{align*}
For the last term, we only consider $x_i$'s such that $x_i > {2}\delta_p$. In this case, 
\begin{align*}
    \sum_{i \in \sbra*{r} \colon x_i > {2}\delta_p} \abs*{ x_i p\rbra*{{\frac{x_i}{2}}}^2 - x_i f\rbra*{{\frac{x_i}{2}}}^2 }
    & = \sum_{i \in \sbra*{r} \colon x_i > {2}\delta_p} \abs*{x_i} \abs*{p\rbra*{{\frac{x_i}{2}}} + f\rbra*{{\frac{x_i}{2}}}} \abs*{p\rbra*{{\frac{x_i}{2}}} - f\rbra*{{\frac{x_i}{2}}}} \\
    & \leq \sum_{i \in \sbra*{r} \colon x_i > {2}\delta_p} 2\abs*{x_i} \varepsilon_p \\
    & \leq \sum_{i \in \sbra*{r}} 2 x_i \varepsilon_p \\
    & = 2\varepsilon_p.
\end{align*}
Combining both cases, we have
\[
\abs*{ \tr\rbra*{\rho p\rbra*{{\frac{\rho}{2}}}^2} - \tr\rbra*{\rho f\rbra*{{\frac{\rho}{2}}}^2} } \leq \frac 5 {{8}} r \delta_p + 2\varepsilon_p.
\]
\end{proof}

\subsubsection{Sample access}

We state the function $\texttt{estimate\_R\'enyi\_lt1\_promise}\rbra{\alpha, P, \varepsilon, \delta}$ in \cref{algo:renyi-lt1-main-intro} again in \cref{algo:esimate-small-promise} with detailed constraints. 

\begin{algorithm}[!htp]
    \caption{$\texttt{estimate\_R\'enyi\_lt1\_promise}\rbra{\alpha, P, \varepsilon, \delta}$ --- \textit{quantum sample algorithm}}
    \label{algo:esimate-small-promise}
    \begin{algorithmic}[1]
    \Require $0 < \alpha < 1$, $\varepsilon \in \rbra{0, 1}$, $\delta \in \rbra{0, 1}$, access to $N$-dimensional quantum state $\rho$ of rank $r$, promise that $P \leq P_{\alpha}\rbra{\rho} \leq 10P$, and $\texttt{R\'enyi\_lt1\_subroutine}\rbra{\alpha, P, \delta_p, \varepsilon_p, \delta_Q}$ defined by \cref{algo:sub-renyi-lt1}.
    \Ensure $\widetilde P$ such that $\rbra{1 - \varepsilon} P_{\alpha}\rbra{\rho} \leq \widetilde P \leq \rbra{1 + \varepsilon} P_{\alpha}\rbra{\rho}$ with probability $\geq 1 - \delta$. 

    \State $m \gets \ceil{8\ln\rbra{1/\delta}}$, $\delta_p \gets {\frac{1}{2}}\rbra*{\frac{P\varepsilon}{40r}}^{1/\alpha}$.

    \State $\varepsilon_p \gets \frac{\rbra{{2}\delta_p}^{1-\alpha}P\varepsilon}{256}$, $\delta_Q \gets \frac{\rbra{{2}\delta_p}^{1-\alpha}P\varepsilon}{128r}$, $\delta_a \gets \frac{\rbra{{2}\delta_p}^{1-\alpha}P\varepsilon}{128}$, and $k \gets \ceil*{\frac{65536}{\rbra{{2}\delta_p}^{1-\alpha}P\varepsilon^2}}$.

    \For {$j = 1 \dots m$}
        \For {$i = 1 \dots k$}
            \State Prepare $\sigma = \mathsf{Samplize}_{\delta_a}\ave{\texttt{R\'enyi\_lt1\_subroutine}\rbra{\alpha, P, \delta_p, \varepsilon_p, \delta_Q}}\sbra{\rho}\rbra{\rho \otimes \ket{0}\bra{0}^{\otimes a}}$.
            \State Measure $\sigma$ in the computational basis. 
            \State Let $X_i$ be $1$ if the outcome is $\ket{0}^{\otimes a}$, and $0$ otherwise. 
        \EndFor
        \State $\hat P_j \gets 16 \rbra{{2}\delta_p}^{\alpha - 1} \sum_{i \in \sbra{k}} X_i / k$.
    \EndFor

    \State $\widetilde P \gets $ the median of $\hat P_j$ for $j \in \sbra{m}$. 

    \State \Return $\widetilde P$.
    
    \end{algorithmic}
\end{algorithm}

\begin{lemma}
\label{lemma:renyi-lt1-promise}
    Suppose $0 < \alpha < 1$ is a constant, and $\rho$ is an $N$-dimensional quantum state of rank $r$ with a promise that $P \leq P_\alpha\rbra{\rho} \leq 10P$ for some $P > 0$.
    Then, \cref{algo:esimate-small-promise} with sample access to $\rho$, with probability $\geq 1 - \delta$, outputs an estimate $\widetilde P$ of $P_\alpha\rbra{\rho}$ such that $\rbra{1-\varepsilon}P_\alpha\rbra{\rho} \leq \widetilde P \leq \rbra{1+\varepsilon} P_\alpha\rbra{\rho}$, with sample complexity $M$ and time complexity $T\rbra{\alpha, \varepsilon, P, \delta} = O\rbra{M\log\rbra{N}}$, where
    \[
    M = O\rbra*{\frac{r^{\frac{4}{\alpha}-2}}{P^{\frac{4}{\alpha}}\varepsilon^{1+\frac{4}{\alpha}}} \log^4\rbra*{\frac{r}{P\varepsilon}} \log\rbra*{\frac 1 \delta}}.
    \]
\end{lemma}

\begin{proof}
    We take $\delta_p \in (0, 1/2)$, $\varepsilon_p \in (0, 1/2)$, and $\delta_Q \in (0, 1)$ to be determined. 
    By \cref{lemma:renyi-lt1-block-encoded}, measuring the quantum state 
    \[
    \sigma = \texttt{R\'enyi\_lt1\_subroutine}\rbra{\alpha, P, \delta_p, \varepsilon_p, \delta_Q}\sbra{U_\rho}\rbra{\rho \otimes \ket{0}_a\bra{0}_a}
    \]
    in the computational basis will obtain the outcome $\ket{0}_a$ with probability $p_a$ such that 
    \begin{equation}
    \label{eq:renyi-lt1-sub}
    \abs*{ p_a - \frac{1}{16} \rbra{{2}\delta_p}^{1-\alpha} P_{\alpha}\rbra*{\rho} } \leq \frac{5}{{8}} r \delta_p + 2 \varepsilon_p + 2r\delta_Q,
    \end{equation}
    with query complexity $Q$ and time complexity $O\rbra*{Q \log\rbra*{N}}$, where $Q = O\rbra*{\frac{1}{\delta_p}\log\rbra*{\frac{1}{\varepsilon_p}}}$.

    Let $\delta_a > 0$ to be determined. 
    By \cref{lemma:block-encoding-to-sample}, there is a unitary operator $U_\rho$ that is a block-encoding of $\rho$ such that
    \[
    \begin{aligned}
    \big\|\mathsf{Samplize}_{\delta_a}\ave{&\texttt{R\'enyi\_lt1\_subroutine}\rbra{\alpha, P, \delta_p, \varepsilon_p, \delta_Q}}\sbra{\rho} - \\
    &\texttt{R\'enyi\_lt1\_subroutine}\rbra{\alpha, P, \delta_p, \varepsilon_p, \delta_Q}\sbra{U_\rho}\big\|_{\diamond} \leq \delta_a.
    \end{aligned}
    \]
    This means that measuring 
    \[
    \widetilde\sigma = \mathsf{Samplize}_{\delta_a}\ave{\texttt{R\'enyi\_lt1\_subroutine}\rbra{\alpha, P, \delta_p, \varepsilon_p, \delta_Q}}\sbra{\rho}\rbra{\rho \otimes \ket{0}_a\bra{0}_a}
    \]
    in the computational basis will obtain the outcome $\ket{0}_a$ with probability $\widetilde p_a$, where
    \begin{equation}
    \label{eq:pa-tildepa-lt1}
        \abs*{p_a - \widetilde p_a} \leq \delta_a,
    \end{equation}
    with sample complexity 
    \[
    1 + O\rbra*{\frac{Q^2}{\delta_a} \log^2\rbra*{\frac{Q}{\delta_a}}} = O\rbra*{\frac{1}{\delta_a \delta_p^2} \log^2\rbra*{\frac{1}{\varepsilon_p}} \log^2\rbra*{\frac{1}{\delta_a\delta_p}\log\rbra*{\frac{1}{\varepsilon_p}}}}.
    \]
    We call by algorithm $\mathcal{A}$ the above procedure that outputs $1$ with probability $\widetilde p_a$.

    We repeat algorithm $\mathcal{A}$ for $k$ times and let $X_i \in \cbra{0, 1}$ be the output of the $i$-th repetition. Let
    \begin{equation}
    \label{eq:renyi-lt1-tilde-P}
    \widetilde P = 16\rbra{{2}\delta_p}^{\alpha-1} X,
    \end{equation}
    where
    \[
    X = \frac 1 k \sum_{i \in \sbra{k}} X_i.
    \]
    We note that $\mathbb{E}\sbra{X} = \mathbb{E}\sbra{X_i} = \widetilde p_a$, and $\operatorname{Var}\sbra{X} = \operatorname{Var}\sbra{X_i}/k = \rbra{\widetilde p_a - \widetilde p_a^2}/k$.
    By Chebyshev's inequality (\cref{thm:chebyshev}), for every $\varepsilon_C > 0$, 
    \begin{equation}
    \label{eq:lt1-X-tildepa}
        \Pr \rbra*{ \abs*{X - \widetilde p_a} \geq \varepsilon_C } \leq \frac{\widetilde p_a - \widetilde p_a^2}{k \varepsilon_C^2}.
    \end{equation}
    From \cref{eq:renyi-lt1-sub}, \cref{eq:pa-tildepa-lt1}, \cref{eq:renyi-lt1-tilde-P}, and \cref{eq:lt1-X-tildepa} we know that with probability $\geq 1 - \rbra{\widetilde p_a - \widetilde p_a^2}/{k \varepsilon_C^2}$,
    \[
    \abs*{\hat P - P_{\alpha}\rbra*{\rho}} \leq 5r\rbra{{2}\delta_p}^{\alpha} + 16\rbra{{2}\delta_p}^{\alpha-1}\rbra*{2\varepsilon_p + 2r\delta_Q + \delta_a + \varepsilon_C}.
    \]
    By taking $\delta_p = {\frac{1}{2}}\rbra*{\frac{P\varepsilon}{40r}}^{1/\alpha}$, $\varepsilon_p = \frac{\rbra{{2}\delta_p}^{1-\alpha}P\varepsilon}{256}$, $\delta_Q = \frac{\rbra{{2}\delta_p}^{1-\alpha}P\varepsilon}{128r}$, $\delta_a = \varepsilon_C = \frac{\rbra{{2}\delta_p}^{1-\alpha}P\varepsilon}{128}$, and $k = \frac{65536}{\rbra{{2}\delta_p}^{1-\alpha}P\varepsilon^2}$, we have that, with probability $\geq 3/4$, it holds that $\abs*{\hat P - P_\alpha\rbra{\rho}} \leq P\varepsilon$, which implies
    \[
        \rbra*{1 - \varepsilon} P_{\alpha}\rbra{\rho} \leq \hat P \leq \rbra*{1 + \varepsilon} P_{\alpha}\rbra{\rho}
    \]
    by noting that $P \leq P_{\alpha}\rbra{\rho} \leq 10P$. In other words, by choosing the values of $\delta_p, \varepsilon_p, \delta_Q, \delta_a, \varepsilon_C, k$ as above, we have
    \[
    \Pr\sbra*{ \rbra*{1 - \varepsilon} P_{\alpha}\rbra{\rho} \leq \hat P \leq \rbra*{1 + \varepsilon} P_{\alpha}\rbra{\rho} } \geq \frac{3}{4}.
    \]
    We denote the algorithm described above  as $\mathcal{B}$. 

    Finally, we need to amplify the success probability from $3/4$ to $1-\delta$. To achieve this, we repeat algorithm $\mathcal{B}$ for $m = \ceil{8\ln\rbra{1/\delta}}$ times, and write $\hat P_i$ to be the $i$-th estimate for $i \in \sbra{m}$. 
    Let $\widetilde P$ be the median of all $\hat P_i$ for $i \in \sbra{m}$, then it holds that
    \[
    \Pr\sbra*{ \rbra*{1 - \varepsilon} P_{\alpha}\rbra{\rho} \leq \widetilde P \leq \rbra*{1 + \varepsilon} P_{\alpha}\rbra{\rho} } \geq 1 - \delta.
    \]
    The proof of this part is the same as that of \cref{lemma:renyi-gt1-promise}.

    Now we are going to analyze the complexity. We repeat algorithm $\mathcal{B}$ for $m$ times, and each repetition of $\mathcal{B}$ repeats algorithm $\mathcal{A}$ for $k$ times.
    Thus, the total sample complexity is
    \[
    mk \cdot O\rbra*{\frac{1}{\delta_a \delta_p^2} \log^2\rbra*{\frac{1}{\varepsilon_p}} \log^2\rbra*{\frac{1}{\delta_a\delta_p}\log\rbra*{\frac{1}{\varepsilon_p}}}} = O\rbra*{\frac{r^{\frac{4}{\alpha}-2}}{P^{\frac{4}{\alpha}}\varepsilon^{1+\frac{4}{\alpha}}} \log^4\rbra*{\frac{r}{P\varepsilon}} \log\rbra*{\frac 1 \delta}},
    \]
    and the time complexity is only a $\log\rbra{N}$ factor over the sample complexity, which is
    \[
    O\rbra*{\frac{r^{\frac{4}{\alpha}-2}}{P^{\frac{4}{\alpha}}\varepsilon^{1+\frac{4}{\alpha}}} \log^4\rbra*{\frac{r}{P\varepsilon}} \log\rbra*{\frac 1 \delta} \log\rbra{N}}.
    \]
\end{proof}

Finally, we will use the estimate of $\widetilde P$ of $P_\alpha\rbra{\rho}$ with multiplicative error to estimate $S_\alpha\rbra{\rho}$ with additive error. 
This process is simple and stated in \cref{algo:esimate-lt1}, and we include it here for completeness. 

\begin{algorithm}[!htp]
    \caption{$\texttt{estimate\_R\'enyi\_lt1\_main}\rbra{\alpha, \varepsilon, \delta}$ --- \textit{quantum sample algorithm}}
    \label{algo:esimate-lt1}
    \begin{algorithmic}[1]
    \Require $0 < \alpha < 1$, $\varepsilon \in \rbra{0, 1}$, $\delta \in \rbra{0, 1}$, sample access to quantum state $\rho$ of rank $r$, and $\texttt{estimate\_R\'enyi\_lt1}\rbra{\alpha, \varepsilon, \delta}$ defined by \cref{algo:esimate-small}.
    \Ensure $\widetilde S$ such that $\abs{\widetilde S - S_{\alpha}\rbra{\rho}} \leq \varepsilon$ with probability $\geq 1 - \delta$. 

    \State $\widetilde P \gets \texttt{estimate\_R\'enyi\_lt1}\rbra{\alpha, \rbra{1-\alpha}\varepsilon/2, \delta}$.

    \State $\widetilde S \gets \frac{1}{1-\alpha}\ln\rbra{\widetilde P}$.
    
    \State \Return $\widetilde S$.
    \end{algorithmic}
\end{algorithm}

\begin{theorem}
\label{thm:estimate-renyi-lt1}
    Suppose $0 < \alpha < 1$ is a constant, and $N$-dimensional quantum state $\rho$ is of rank $r$.
    Then, \cref{algo:esimate-lt1}, with probability $\geq 1 - \delta$, outputs an estimate $\widetilde S$ of $S_\alpha\rbra{\rho}$ within additive error $\varepsilon$, using $M$ samples of $\rho$ and $O\rbra{M \log\rbra{N}}$ one- and two-qubit quantum gates, where
    \[
    M = O\rbra*{ r^{\frac{4}{\alpha}-2} \rbra*{\log^6\rbra*{r} + \frac{1}{\varepsilon^{1+\frac{4}{\alpha}} } \log^4\rbra*{\frac{r}{\varepsilon}}} \log\rbra*{\frac 1 \delta} }.
    \]
\end{theorem}

\begin{proof}
    By \cref{lemma:renyi-lt1-promise}, we can implement the procedure $\texttt{estimate\_R\'enyi\_lt1\_promise}\rbra{\alpha, \varepsilon, P, \delta}$ required in \cref{algo:esimate-small}.
    Then, by \cref{lemma:annealing-small}, \cref{algo:esimate-small} returns an estimate $\widetilde P$ of $P_{\alpha}\rbra{\rho}$ such that
    \[
    \rbra{1 - \epsilon} P_\alpha\rbra{\rho} \leq \widetilde P \leq \rbra{1 + \epsilon} P_\alpha\rbra{\rho}
    \]
    with probability $\geq 1 - \delta$.
    By taking $\widetilde S = \frac{1}{1-\alpha} \ln \rbra{\widetilde P}$, we have that
    \[
    S_\alpha\rbra{\rho} + \frac{\ln\rbra{1 - \epsilon}}{1 - \alpha} \leq \widetilde S \leq S_\alpha\rbra{\rho} + \frac{\ln\rbra{1 + \epsilon}}{1 - \alpha}.
    \]
    By letting $\epsilon = \rbra{1-\alpha}\varepsilon/2$, we have $S_\alpha\rbra{\rho} - \varepsilon \leq \widetilde S \leq S_\alpha\rbra{\rho} + \varepsilon$, as required.

    Now we are going to analyze the complexity. 
    Let $\lambda = 1 - 1/\ln\rbra{r}$. 
    Combining \cref{lemma:annealing-small} and \cref{lemma:renyi-lt1-promise}, we have the recurrence relation: if $\lambda \leq \alpha < 1$, then
    \[
    Q\rbra*{\alpha, \varepsilon, \delta} = O\rbra*{\frac{r^{\frac{4}{\alpha}-2}}{\varepsilon^{1+\frac{4}{\alpha}}} \log^4\rbra*{\frac{r}{\varepsilon}} \log\rbra*{\frac 1 \delta} \log\rbra*{N}};
    \]
    and if $0 < \alpha < \lambda$, then
    \begin{align*}
        Q\rbra*{\alpha, \varepsilon, \delta}
        = Q\rbra*{\frac{\alpha}{\lambda}, \frac{1}{4}, \frac{\delta}{2}} + O\rbra*{\frac{r^{\frac{4}{\alpha}-2}}{\varepsilon^{1+\frac{4}{\alpha}}} \log^4\rbra*{\frac{r}{\varepsilon}} \log\rbra*{\frac 1 \delta} \log\rbra*{N}}.
    \end{align*}
    For the special case of $\varepsilon = 1/4$, we have: if $\lambda \leq \alpha < 1$, then
    \[
    Q\rbra*{\alpha, \frac{1}{4}, \delta} = O \rbra*{ r^{\frac{4}{\alpha}-2} \log^4\rbra*{r} \log\rbra*{\frac 1 \delta} \log\rbra*{N} };
    \]
    and if $0 < \alpha < \lambda$, then 
    \begin{align*}
        Q\rbra*{\alpha, \frac{1}{4}, \delta} = Q\rbra*{\frac{\alpha}{\lambda}, \frac{1}{4}, \frac{\delta}{2}} + O \rbra*{ r^{\frac{4}{\alpha}-2} \log^4\rbra*{r} \log\rbra*{\frac 1 \delta} \log\rbra*{N} },
    \end{align*}
    where both cases satisfy that 
    \begin{align*}
        Q\rbra*{\alpha, \frac{1}{4}, \delta} 
        & \leq \sum_{k=0}^{\floor{\log_\lambda\rbra{\alpha}}} O\rbra*{  r^{\frac{4}{\alpha}-2} \log^4\rbra*{r} \log\rbra*{\frac {2^k} \delta} \log\rbra*{N} } \\
        & = O\rbra*{ r^{\frac{4}{\alpha}-2} \log^6\rbra*{r} \log\rbra*{\frac 1 \delta} \log\rbra*{N} }.
    \end{align*}
    
    Finally, we have
    \[
    Q\rbra{\alpha, \varepsilon, \delta} = O\rbra*{ r^{\frac{4}{\alpha}-2} \rbra*{\log^6\rbra*{r} + \frac{1}{\varepsilon^{1+\frac{4}{\alpha}} } \log^4\rbra*{\frac{r}{\varepsilon}}} \log\rbra*{\frac 1 \delta} \log\rbra*{N} }.
    \]
\end{proof}

\section{Sample Lower Bounds for Entropy Estimation}
\label{sec:sample-lower-bound-entropy}

In this section, we give sample lower bounds for estimating $\alpha$-R\'enyi entropies for all $\alpha > 0$ (where $\alpha = 1$ means the von Neumann entropy), by reducing to the mixedness testing problem considered in \cite{OW21} and the distinguishing problem of a special distribution used in \cite{AOST17,AISW20}.
The main result is collected in \cref{thm:entropy-estimation-sample-lower-bound}. 

\begin{theorem} [Theorems \ref{thm:sample-lower-bound-von-neumann}, \ref{thm:sample-lower-bound-renyi-gt1}, and \ref{thm:sample-lower-bound-renyi-lt1} combined]
\label{thm:entropy-estimation-sample-lower-bound}
    Suppose $\alpha > 0$ is a constant. 
    Every quantum algorithm with sample access for estimating the $\alpha$-R\'enyi entropy $S_{\alpha}\rbra{\rho}$ of $N$-dimensional quantum state $\rho$ within additive error requires $\Omega\rbra{N/\varepsilon + N^{1/\alpha-1}/\varepsilon^{1/\alpha}}$ independent samples of $\rho$.
\end{theorem}

\subsection{Von Neumann entropy}

\begin{theorem}
\label{thm:sample-lower-bound-von-neumann}
    Every quantum algorithm for estimating the von Neumann entropy $S\rbra{\rho}$ of $N$-dimensional quantum state $\rho$ within additive error $\varepsilon$ requires $\Omega\rbra{N/\varepsilon}$ independent samples of $\rho$. 
\end{theorem}

\begin{proof}
    We consider the problem of mixedness testing of quantum states: given an $N$-dimensional quantum state $\rho$, distinguish whether it is maximally mixed, or it is $\varepsilon$-away (in trace distance) from being maximally mixed. 
    Here, a quantum state is said to be maximally mixed if its eigenvalues are $1/N$. 
    It was shown in \cite[Theorem 1.10]{OW21} that the sample complexity of mixedness testing is $\Theta\rbra{N/\varepsilon^2}$.
    In the following, we will reduce the problem of estimating von Neumann entropy to mixedness testing. 

    Suppose the eigenvalues of $\rho$ are $x_1, x_2, \dots, x_N$. 
    If $\rho$ is maximally mixed, then $x_i = 1/N$ and thus $S\rbra{\rho} = \ln{N}$. 
    If $\rho$ is $\varepsilon$-away (in trace distance) from being maximally mixed, let $S = \set{ i \in \sbra{N} }{ x_i \geq 1 / N }$ and $T = \sbra{N} \setminus S$.
    We note that
    \[
    \sum_{i \in S} \rbra*{x_i - \frac 1 N} = \sum_{i \in T} \rbra*{\frac 1 N - x_i} = \varepsilon.
    \]
    Let $\varphi\rbra{x} = -x\ln\rbra{x}$, and $\varphi''\rbra{x} = -1/x < 0$ for all $x > 0$. By Jensen's inequality, we have
    \begin{align*}
        S\rbra*{\rho}
        & = \sum_{i \in \sbra{N}} \varphi\rbra*{x_i} \\
        & = \sum_{i \in S} \varphi\rbra*{x_i} + \sum_{i \in T} \varphi\rbra*{x_i}\\
        & \leq \abs*{S} \cdot \varphi\rbra*{ \frac{1}{\abs*{S}} \sum_{i \in S} x_i } + \abs*{T} \cdot \varphi\rbra*{ \frac{1}{\abs*{T}} \sum_{i \in T} x_i } \\
        & = \abs*{S} \cdot \varphi\rbra*{ \frac 1 N + \frac \varepsilon {\abs*{S}} } + \abs*{T} \cdot \varphi\rbra*{ \frac 1 N - \frac \varepsilon {\abs*{T}} } \\
        & = \ln\rbra*{N} - \rbra*{ \rbra*{z + \varepsilon} \ln\rbra*{1 + \frac \varepsilon z} + \rbra*{1-z-\varepsilon} \ln\rbra*{1 - \frac{\varepsilon}{1-z}} },
    \end{align*}
    where $z = \abs{S}/N$. 
    The valid range of $z$ is $\frac{\varepsilon}{N-1} \leq z \leq 1-\varepsilon$, and in this range we have (see \cref{lemma:log-eq}) that
    \[
    \rbra*{z + \varepsilon} \ln\rbra*{1 + \frac \varepsilon z} + \rbra*{1-z-\varepsilon} \ln\rbra*{1 - \frac{\varepsilon}{1-z}} \geq \varepsilon^2. 
    \]
    We have $S\rbra{\rho} \leq \ln\rbra{N} - \varepsilon^2$. 
    
    By letting $\epsilon = \varepsilon^2$, every quantum algorithm for estimating $S\rbra{\rho}$ within additive $\epsilon$ will lead to a quantum algorithm with the same sample complexity for mixedness testing with a promise that either $\rho$ is maximally mixed or $\rho$ is $\varepsilon$-away (in trace distance) from being maximally mixed.
    The latter problem has sample lower bound $\Omega\rbra{N/\varepsilon^2} = \Omega\rbra{N/\epsilon}$. 
\end{proof}

To complete the proof of \cref{thm:sample-lower-bound-von-neumann}, it remains to show the following technical lemma.

\begin{lemma}
\label{lemma:log-eq}
    Let $\varepsilon \in \rbra{0, 1/2}$. 
    For $0 < z \leq 1-\varepsilon$, we have
    \[
    \rbra*{z + \varepsilon} \ln\rbra*{1 + \frac \varepsilon z} + \rbra*{1-z-\varepsilon} \ln\rbra*{1 - \frac{\varepsilon}{1-z}} \geq \varepsilon^2. 
    \]
\end{lemma}
\begin{proof}
    Suppose $0 < z < 1$. 
    Let
    \[
    f\rbra{\varepsilon} = \rbra*{z + \varepsilon} \ln\rbra*{1 + \frac \varepsilon z} + \rbra*{1-z-\varepsilon} \ln\rbra*{1 - \frac{\varepsilon}{1-z}} - \varepsilon^2. 
    \]
    We only have to show that $f\rbra{\varepsilon} \geq 0$ for every $0 < \varepsilon \leq \min\cbra{1-z,1/2}$.

    For the case of $0 < \varepsilon < \min\cbra{1-z,1/2}$,
    \[
    f'\rbra{\varepsilon} = \ln\rbra*{1+\frac{\varepsilon}{z}} - \ln\rbra*{1 - \frac{\varepsilon}{1-z}} - 2\varepsilon,
    \]
    \[
    f''\rbra{\varepsilon} = \frac{1}{z+\varepsilon} + \frac{1}{1-\rbra*{z+\varepsilon}} - 2 > 0.
    \]
    Hence, $f'\rbra{\varepsilon}$ is increasing, then $f'\rbra{\varepsilon} \geq f'\rbra{0} = 0$. 
    This further implies that $f\rbra{\varepsilon}$ is increasing, and thus $f\rbra{\varepsilon} \geq f\rbra{0} = 0$. 

    For the limiting case that $\varepsilon = 1 - z$, let
    \[
    g\rbra{z} = \lim_{\varepsilon \to \rbra{1-z}^-} f\rbra{\varepsilon} = \ln\rbra*{\frac 1 z}-\rbra{1-z}^2.
    \]
    Note that 
    \[
    g'\rbra{z} = -2z-\frac 1 z + 2 < 0,
    \]
    i.e., $g\rbra{z}$ is decreasing. Therefore, $g\rbra{z} \geq g\rbra{1} = 0$.
\end{proof}

\subsection{R\'enyi entropy for \texorpdfstring{$\alpha > 1$}{α > 1}}

\begin{theorem}
\label{thm:sample-lower-bound-renyi-gt1}
    Suppose $\alpha > 1$ is a constant. 
    Every quantum algorithm for estimating the $\alpha$-R\'enyi entropy $S_\alpha\rbra{\rho}$ of $N$-dimensional quantum state $\rho$ within additive error $\varepsilon$ requires $\Omega\rbra{N/\varepsilon}$ independent samples of $\rho$. 
\end{theorem}

\begin{proof}
    For every $\alpha > 1$, it holds that $S_\alpha\rbra{\rho} \leq S\rbra{\rho}$ for every quantum state $\rho$ (cf.\ \cite{BS93}).
    Suppose the eigenvalues of $\rho$ are $x_1, x_2, \dots, x_N$. 
    If $\rho$ is maximally mixed, then $x_i = 1/N$ and thus $S_{\alpha}\rbra{\rho} = \ln{N}$.
    Similar to the proof of \cref{thm:sample-lower-bound-von-neumann}, we have: 
    If $\rho$ is $\varepsilon$-away (in trace distance) from being maximally mixed, then $S_\alpha\rbra{\rho} \leq S\rbra{\rho} \leq \ln\rbra{N} - \varepsilon^2$. 
    Thus, we can obtain the same sample lower bound as \cref{thm:sample-lower-bound-von-neumann}.
\end{proof}

\subsection{R\'enyi entropy for \texorpdfstring{$0 < \alpha < 1$}{0 < α < 1}}

\begin{theorem}
\label{thm:sample-lower-bound-renyi-lt1}
    Suppose $0 < \alpha < 1$ is a constant. 
    Every quantum algorithm for estimating the $\alpha$-R\'enyi entropy $S_\alpha\rbra{\rho}$ of $N$-dimensional quantum state $\rho$ within additive error $\varepsilon$ requires $\Omega\rbra{N/\varepsilon + N^{1/\alpha - 1}/\varepsilon^{1/\alpha}}$ independent samples of $\rho$. 
\end{theorem}

To prove the lower bounds in \cref{thm:sample-lower-bound-renyi-lt1}, we need the following lower bound for quantum state discrimination.
Quantum state discrimination is a task for distinguishing between two quantum states. 
The success probability of any protocol for quantum state discrimination can be upper bounded in terms of the trace distance between the two quantum states, which is originated from \cite{Hel67,Hol73}.
We state it as follows.

\begin{theorem} [Quantum state discrimination, cf.\ {\cite[Section 9.1.4]{Wil13}}]
\label{thm:HH-measurement}
    Suppose that $\rho_0$ and $\rho_1$ are two quantum states. 
    Let $\varrho$ be a quantum state such that $\varrho = \rho_0$ or $\varrho = \rho_1$ with equal probability. 
    For any POVM $\Lambda = \cbra{\Lambda_0, \Lambda_1}$,
    the success probability of distinguishing the two cases is bounded by
    \[
    p_{\mathrm{succ}} = \frac 1 2 \tr\rbra*{\Lambda_0\rho_0} + \frac 1 2 \tr\rbra*{\Lambda_1\rho_1} \leq \frac{1}{2}\rbra*{1+\frac{1}{2}\Abs{\rho_0-\rho_1}_1}. 
    \]
\end{theorem}

We also need the following inequality of the $\alpha$-R\'enyi entropy for $0 < \alpha < 1$.

\begin{lemma} [Lemma 32 of the full version of \cite{AISW20}]
\label{lemma:ineq-renyi-lt1}
    Suppose $p_1, p_2, \dots, p_n$ is a probability distribution, i.e., $p_i \geq 0$ and $\sum_{i \in \sbra{n}} p_i = 1$, such that $\sum_{i \in \sbra{n}} \abs{p_i - 1/n} = 2\varepsilon$ for some $\varepsilon > 0$. 
    Then, for $0 < \alpha < 1$, 
    \[
    \sum_{i \in \sbra{n}} p_i^{\alpha} \leq \rbra*{ 1 - \alpha\rbra*{1-\alpha} \varepsilon^2 } n^{1-\alpha}.
    \]
\end{lemma}

Now we are ready to prove \cref{thm:sample-lower-bound-renyi-lt1}.

\begin{proof}[Proof of \cref{thm:sample-lower-bound-renyi-lt1}]
    The proof is split into two parts. The first part shows a sample lower bound $\Omega\rbra{N/\varepsilon}$ and the second part shows a sample lower bound $\Omega\rbra{N^{1/\alpha-1}/\varepsilon^{1/\alpha}}$; combining the both yields the proof. 

    We first show a lower bound $\Omega\rbra{N/\varepsilon}$. 
    Suppose the eigenvalues of $\rho$ are $x_1, x_2, \dots, x_N$. 
    If $\rho$ is maximally mixed, then $x_i = 1/N$ and thus $S_{\alpha}\rbra{\rho} = \ln{N}$.
    Similar to the proof of \cref{thm:sample-lower-bound-von-neumann}, we have: 
    If $\rho$ is $\varepsilon$-away (in trace distance) from being maximally mixed, then by \cref{lemma:ineq-renyi-lt1}, we have
    \[
    \sum_{i \in \sbra{N}} x_i^\alpha \leq \rbra*{ 1 - \alpha\rbra*{1-\alpha} \varepsilon^2 } N^{1-\alpha},
    \]
    and thus
    \begin{align*}
        S_\alpha\rbra{\rho}
        & = \frac{1}{1-\alpha} \ln \rbra*{ \sum_{i \in \sbra{N}} x_i^\alpha } \\
        & \leq \ln\rbra*{N} + \frac{1}{1-\alpha} \ln\rbra*{ 1 - \alpha\rbra*{1-\alpha}\varepsilon^2 } \\
        & \leq \ln\rbra*{N} - \alpha \varepsilon^2.
    \end{align*}
    Following the proof of \cref{thm:sample-lower-bound-von-neumann}, we can obtain the same sample lower bound.

    Then, we consider the problem of distinguishing two $N$-dimensional quantum states $\rho$ and $\sigma$ such that
    \begin{align*}
    \rho & = \diag\rbra*{1 - \delta, \frac{\delta}{N-1}, \dots, \frac{\delta}{N-1}}, \\
    \sigma & = \diag\rbra*{1, 0, \dots, 0},
    \end{align*}
    where $\delta = \rbra*{\frac{2\varepsilon}{\rbra*{N-1}^{1-\alpha}}}^{1/\alpha}$ and $N > 2 \varepsilon^{1-\alpha} + 1$.
    Such construction was ever used in analyzing the sample complexity of estimating classical R\'enyi entropy \cite{AOST17} and quantum R\'enyi entropy \cite{AISW20} as well as the quantum query complexity \cite{LWZ22}.
    Note that $\delta < \varepsilon$. 
    It can be seen that $S_{\alpha}\rbra{\sigma} = 0$ and 
    \begin{align*}
        S_{\alpha}\rbra{\rho}
        & = \frac{1}{1-\alpha} \ln\rbra*{ \rbra{1-\delta}^{\alpha} + \delta^{\alpha} \rbra{N-1}^{1-\alpha} } \\
        & \geq \frac{1}{1-\alpha} \ln\rbra*{ 1 - \delta + 2\varepsilon } \\
        & \geq \frac{1}{1-\alpha} \ln\rbra*{ 1 + \varepsilon } \geq \frac \varepsilon 2.
    \end{align*}
    Suppose a quantum algorithm for estimating $\alpha$-R\'enyi entropy within additive error $\varepsilon$ uses $\mathsf{S}$ samples, then it can distinguish the two quantum states $\rho$ and $\sigma$ with probability $p_{\text{succ}} \geq 2/3$. 
    On the other hand, by \cref{thm:HH-measurement}, the success probability of the quantum hypothesis testing experiment \cite{Hel67,Hol73} is upper bounded by
    \[
    p_{\text{succ}} \leq \frac{1 + \frac{1}{2}\Abs*{\rho^{\otimes \mathsf{S}} - \sigma^{\otimes \mathsf{S}}}_1}{2},
    \]
    where
    \[
    \frac{1}{2} \Abs*{\rho^{\otimes \mathsf{S}} - \sigma^{\otimes \mathsf{S}}}_1 \leq \sqrt{1 - F\rbra{\rho, \sigma}^{2\mathsf{S}}} = \sqrt{1 - \rbra{1-\delta}^{\mathsf{S}}}.
    \]
    Finally, we obtain that $\mathsf{S} = \Omega\rbra{1/\delta} = \Omega\rbra{N^{1/\alpha - 1}/\varepsilon^{1/\alpha}}$.
\end{proof}

\section*{Acknowledgment}

The authors would like to thank John Wright for valuable comments and sharing their results \cite{BMW16} on von Neumann entropy estimation, 
thank Zhengfeng Ji for pointing out the related work \cite{ARU14},
thank Masahito Hayashi for helpful discussions and sharing the related work \cite{Hay24}, 
thank Aaron B.\ Wagner for explaining the EYD algorithms proposed in \cite{AISW20}, thank Luowen Qian for pointing out the related work \cite{HL11,SGSV24}, and thank anonymous reviewers for constructive suggestions on the organization of this paper and pointing out a mistake in an earlier version of this paper.
Qisheng Wang also thanks Fran\c{c}ois Le Gall for helpful discussions.

The work of Qisheng Wang was supported in part by the Engineering and Physical Sciences Research Council under Grant \mbox{EP/X026167/1} and in part by the Ministry of Education, Culture, Sports, Science and Technology
(MEXT) Quantum Leap Flagship Program (MEXT Q-LEAP) under Grant \mbox{JPMXS0120319794}. 
The work of Zhicheng Zhang was supported by the Sydney Quantum Academy, NSW, Australia, and the Australian Research Council Discovery Project under Grant DP220102059.

\addcontentsline{toc}{section}{References}
\bibliographystyle{unsrturl}
\bibliography{main}

\begin{thebibliography}{10}

\bibitem{WZ24b}
Qisheng Wang and Zhicheng Zhang.
\newblock Time-efficient quantum entropy estimator via samplizer.
\newblock In {\em Proceedings of the 32nd Annual European Symposium on Algorithms}, pages 101:1--101:15, 2024.
\newblock \href {https://doi.org/10.4230/LIPIcs.ESA.2024.101} {\path{doi:10.4230/LIPIcs.ESA.2024.101}}.

\bibitem{NC10}
Michael~A. Nielsen and Isaac~L. Chuang.
\newblock {\em Quantum Computation and Quantum Information}.
\newblock Cambridge University Press, 2010.
\newblock \href {https://doi.org/10.1017/CBO9780511976667} {\path{doi:10.1017/CBO9780511976667}}.

\bibitem{Wil13}
Mark~M. Wilde.
\newblock {\em Quantum Information Theory}.
\newblock Cambridge University Press, 2013.
\newblock \href {https://doi.org/10.1017/9781316809976} {\path{doi:10.1017/9781316809976}}.

\bibitem{Hay16}
Masahito Hayashi.
\newblock {\em Quantum Information Theory: Mathematical Foundation}.
\newblock Cambridge University Press, 2016.
\newblock \href {https://doi.org/10.1007/978-3-662-49725-8} {\path{doi:10.1007/978-3-662-49725-8}}.

\bibitem{Wat18}
John Watrous.
\newblock {\em The Theory of Quantum Information}.
\newblock Cambridge University Press, 2018.
\newblock \href {https://doi.org/10.1017/9781316848142} {\path{doi:10.1017/9781316848142}}.

\bibitem{Sch95}
Benjamin Schumacher.
\newblock Quantum coding.
\newblock {\em Physical Review A}, 51(4):2738, 1995.
\newblock \href {https://doi.org/10.1103/PhysRevA.51.2738} {\path{doi:10.1103/PhysRevA.51.2738}}.

\bibitem{JS94}
Richard Jozsa and Benjamin Schumacher.
\newblock A new proof of the quantum noiseless coding theorem.
\newblock {\em Journal of Modern Optics}, 41(12):2343--2349, 1994.
\newblock \href {https://doi.org/10.1080/09500349414552191} {\path{doi:10.1080/09500349414552191}}.

\bibitem{Lo95}
Hoi-Kwong Lo.
\newblock Quantum coding theorem for mixed states.
\newblock {\em Optics Communications}, 119(5-6):552--556, 1995.
\newblock \href {https://doi.org/10.1016/0030-4018(95)00406-X} {\path{doi:10.1016/0030-4018(95)00406-X}}.

\bibitem{HHHH09}
Ryszard Horodecki, Pawe{\l} Horodecki, Micha{\l} Horodecki, and Karol Horodecki.
\newblock Quantum entanglement.
\newblock {\em Reviews of Modern Physics}, 81(2):865, 2009.
\newblock \href {https://doi.org/10.1103/RevModPhys.81.865} {\path{doi:10.1103/RevModPhys.81.865}}.

\bibitem{Laf16}
Nicolas Laflorencie.
\newblock Quantum entanglement in condensed matter systems.
\newblock {\em Physics Reports}, 646:1--59, 2016.
\newblock \href {https://doi.org/10.1016/j.physrep.2016.06.008} {\path{doi:10.1016/j.physrep.2016.06.008}}.

\bibitem{FIK08}
F.~Franchini, A.~R. Its, and V.~E. Korepin.
\newblock Renyi entropy of the {XY} spin chain.
\newblock {\em Journal of Physics A: Mathematical and Theoretical}, 41(2):025302, 2008.
\newblock \href {https://doi.org/10.1088/1751-8113/41/2/025302} {\path{doi:10.1088/1751-8113/41/2/025302}}.

\bibitem{HGKM10}
Matthew~B. Hastings, Iv{\'{a}}n Gonz{\'{a}}lez, Ann~B. Kallin, and Roger~G. Melko.
\newblock Measuring {Renyi} entanglement entropy in quantum {Monte} {Carlo} simulations.
\newblock {\em Physical Review Letters}, 104(15):157201, 2010.
\newblock \href {https://doi.org/10.1103/PhysRevLett.104.157201} {\path{doi:10.1103/PhysRevLett.104.157201}}.

\bibitem{IMP+15}
Rajibul Islam, Ruichao Ma, Philipp~M. Preiss, M.~Eric Tai, Alexander Lukin, Matthew Rispoli, and Markus Greiner.
\newblock Measuring entanglement entropy in a quantum many-body system.
\newblock {\em Nature}, 528(7580):77--83, 2015.
\newblock \href {https://doi.org/10.1038/nature15750} {\path{doi:10.1038/nature15750}}.

\bibitem{AISW20}
Jayadev Acharya, Ibrahim Issa, Nirmal~V. Shende, and Aaron~B. Wagner.
\newblock Estimating quantum entropy.
\newblock {\em IEEE Journal on Selected Areas in Information Theory}, 1(2):454--468, 2020.
\newblock \href {https://doi.org/10.1109/JSAIT.2020.3015235} {\path{doi:10.1109/JSAIT.2020.3015235}}.

\bibitem{BMW16}
Mohammad Bavarian, Saeed Mehraban, and John Wright.
\newblock Learning entropy.
\newblock A manuscript on von Neumann entropy estimation, private communication, 2016.

\bibitem{ARS88}
Robert Alicki, S{\l}awomir Rudnicki, and S{\l}awomir Sadowski.
\newblock Symmetry properties of product states for the system of $n$ $n$-level atoms.
\newblock {\em Journal of Mathematical Physics}, 29(5):1158--1162, 1988.
\newblock \href {https://doi.org/10.1063/1.527958} {\path{doi:10.1063/1.527958}}.

\bibitem{KW01}
M.~Keyl and R.~F. Werner.
\newblock Estimating the spectrum of a density operator.
\newblock {\em Physical Review A}, 64(5):052311, 2001.
\newblock \href {https://doi.org/10.1103/PhysRevA.64.052311} {\path{doi:10.1103/PhysRevA.64.052311}}.

\bibitem{CHW07}
Andrew~M. Childs, Aram~W. Harrow, and Pawe{\l} Wocjan.
\newblock Weak {Fourier-Schur} sampling, the hidden subgroup problem, and the quantum collision problem.
\newblock In {\em Proceedings of the 24th Annual Symposium on Theoretical Aspects of Computer Science}, pages 598--609, 2007.
\newblock \href {https://doi.org/10.1007/978-3-540-70918-3_51} {\path{doi:10.1007/978-3-540-70918-3_51}}.

\bibitem{MdW16}
Ashley Montanaro and Ronald de~Wolf.
\newblock A survey of quantum property testing.
\newblock In {\em Theory of Computing Library}, number~7 in Graduate Surveys, pages 1--81. University of Chicago, 2016.
\newblock \href {https://doi.org/10.4086/toc.gs.2016.007} {\path{doi:10.4086/toc.gs.2016.007}}.

\bibitem{Wri22}
John Wright.
\newblock Private communication, 2022.

\bibitem{KS16}
Yasuhito Kawano and Hiroshi Sekigawa.
\newblock Quantum {Fourier} transform over symmetric groups --- improved result.
\newblock {\em Journal of Symbolic Computation}, 75:219--243, 2016.
\newblock \href {https://doi.org/10.1016/j.jsc.2015.11.016} {\path{doi:10.1016/j.jsc.2015.11.016}}.

\bibitem{VV11}
Gregory Valiant and Paul Valiant.
\newblock Estimating the unseen: an $n/\log(n)$-sample estimator for entropy and support size, shown optimal via new {CLTs}.
\newblock In {\em Proceedings of the 43rd Annual ACM Symposium on Theory of Computing}, pages 685--694, 2011.
\newblock \href {https://doi.org/10.1145/1993636.1993727} {\path{doi:10.1145/1993636.1993727}}.

\bibitem{VV17}
Gregory Valiant and Paul Valiant.
\newblock Estimating the unseen: improved estimators for entropy and other properties.
\newblock {\em Journal of the ACM}, 64(6):37:1--37:41, 2017.
\newblock \href {https://doi.org/10.1145/3125643} {\path{doi:10.1145/3125643}}.

\bibitem{JVHW15}
Jiantao Jiao, Kartik Venkat, Yanjun Han, and Tsachy Weissman.
\newblock Minimax estimation of functionals of discrete distributions.
\newblock {\em IEEE Transactions on Information Theory}, 61(5):2835--2885, 2015.
\newblock \href {https://doi.org/10.1109/TIT.2015.2412945} {\path{doi:10.1109/TIT.2015.2412945}}.

\bibitem{JVHW17}
Jiantao Jiao, Kartik Venkat, Yanjun Han, and Tsachy Weissman.
\newblock Maximum likelihood estimation of functionals of discrete distributions.
\newblock {\em IEEE Transactions on Information Theory}, 63(10):6774--6798, 2017.
\newblock \href {https://doi.org/10.1109/TIT.2017.2733537} {\path{doi:10.1109/TIT.2017.2733537}}.

\bibitem{WY16}
Yihong Wu and Pengkun Yang.
\newblock Minimax rates of entropy estimation on large alphabets via best polynomial approximation.
\newblock {\em IEEE Transactions on Information Theory}, 62(6):3702--3720, 2016.
\newblock \href {https://doi.org/10.1109/TIT.2016.2548468} {\path{doi:10.1109/TIT.2016.2548468}}.

\bibitem{AOST17}
Jayadev Acharya, Alon Orlitsky, Ananda~Theertha Suresh, and Himanshu Tyagi.
\newblock Estimating {Renyi} entropy of discrete distributions.
\newblock {\em IEEE Transactions on Information Theory}, 63(1):38--56, 2017.
\newblock \href {https://doi.org/10.1109/TIT.2016.2620435} {\path{doi:10.1109/TIT.2016.2620435}}.

\bibitem{WH19}
Jingxiang Wu and Timothy~H. Hsieh.
\newblock Variational thermal quantum simulation via thermofield double states.
\newblock {\em Physical Review Letters}, 123(22):220502, 2019.
\newblock \href {https://doi.org/10.1103/PhysRevLett.123.220502} {\path{doi:10.1103/PhysRevLett.123.220502}}.

\bibitem{CLW20}
Anirban~N. Chowdhury, Guang~Hao Low, and Nathan Wiebe.
\newblock A variational quantum algorithm for preparing quantum {Gibbs} states.
\newblock ArXiv preprints, 2020.
\newblock \href {https://arxiv.org/abs/2002.00055} {\path{arXiv:2002.00055}}.

\bibitem{WLW21}
Youle Wang, Guangxi Li, and Xin Wang.
\newblock Variational quantum {Gibbs} state preparation with a truncated {Taylor} series.
\newblock {\em Physical Review Applied}, 16(5):054035, 2021.
\newblock \href {https://doi.org/10.1103/PhysRevApplied.16.054035} {\path{doi:10.1103/PhysRevApplied.16.054035}}.

\bibitem{AAKS21}
Anurag Anshu, Srinivasan Arunachalam, Tomotaka Kuwahara, and Mehdi Soleimanifar.
\newblock Sample-efficient learning of interacting quantum systems.
\newblock {\em Nature Physics}, 17(8):931--935, 2021.
\newblock \href {https://doi.org/10.1038/s41567-021-01232-0} {\path{doi:10.1038/s41567-021-01232-0}}.

\bibitem{vN32}
John von Neumann.
\newblock {\em Mathematische Grundlagen der Quantenmechanik (Mathematical Foundations of Quantum Mechanics)}.
\newblock Springer, 1932.
\newblock URL: \url{http://eudml.org/doc/203794}.

\bibitem{GSLW19}
Andr\'{a}s Gily\'{e}n, Yuan Su, Guang~Hao Low, and Nathan Wiebe.
\newblock Quantum singular value transformation and beyond: exponential improvements for quantum matrix arithmetics.
\newblock In {\em Proceedings of the 51st Annual ACM SIGACT Symposium on Theory of Computing}, pages 193--204, 2019.
\newblock \href {https://doi.org/10.1145/3313276.3316366} {\path{doi:10.1145/3313276.3316366}}.

\bibitem{WZW22}
Youle Wang, Benchi Zhao, and Xin Wang.
\newblock Quantum algorithms for estimating quantum entropies.
\newblock {\em Physical Review Applied}, 19(4):044041, 2023.
\newblock \href {https://doi.org/10.1103/PhysRevApplied.19.044041} {\path{doi:10.1103/PhysRevApplied.19.044041}}.

\bibitem{Ren61}
Alfr{\'{e}}d R{\'{e}}nyi.
\newblock On measures of entropy and information.
\newblock In {\em Proceedings of the Fourth Berkeley Symposium on Mathematical Statistics and Probability}, pages 547--561, 1961.
\newblock URL: \url{https://projecteuclid.org/ebook/Download?urlid=bsmsp/1200512181&isFullBook=false}.

\bibitem{CAS+22}
Pedro C.~S. Costa, Dong An, Yuval~R. Sanders, Yuan Su, Ryan Babbush, and Dominic~W. Berry.
\newblock Optimal scaling quantum linear-systems solver via discrete adiabatic theorem.
\newblock {\em PRX Quantum}, 3(4):040303, 2022.
\newblock \href {https://doi.org/10.1103/PRXQuantum.3.040303} {\path{doi:10.1103/PRXQuantum.3.040303}}.

\bibitem{LMR14}
Seth Lloyd, Masoud Mohseni, and Patrick Rebentrost.
\newblock Quantum principal component analysis.
\newblock {\em Nature Physics}, 10(9):631--633, 2014.
\newblock \href {https://doi.org/10.1038/nphys3029} {\path{doi:10.1038/nphys3029}}.

\bibitem{KLL+17}
Shelby Kimmel, Cedric Yen-Yu Lin, Guang~Hao Low, Maris Ozols, and Theodore~J. Yoder.
\newblock Hamiltonian simulation with optimal sample complexity.
\newblock {\em npj Quantum Information}, 3(1):1--7, 2017.
\newblock \href {https://doi.org/10.1038/s41534-017-0013-7} {\path{doi:10.1038/s41534-017-0013-7}}.

\bibitem{GKP+24}
Byeongseon Go, Hyukjoon Kwon, Siheon Park, Dhrumil Patel, and Mark~M. Wilde.
\newblock Sample-based {Hamiltonian} and {Lindbladian} simulation: Non-asymptotic analysis of sample complexity.
\newblock {\em Quantum Science and Technology}, 10(4):045058, 2025.
\newblock \href {https://doi.org/10.1088/2058-9565/ae075b} {\path{doi:10.1088/2058-9565/ae075b}}.

\bibitem{GP22}
Andr\'{a}s {Gily\'{e}n} and Alexander Poremba.
\newblock Improved quantum algorithms for fidelity estimation.
\newblock ArXiv preprints, 2022.
\newblock \href {https://arxiv.org/abs/2203.15993} {\path{arXiv:2203.15993}}.

\bibitem{WZ23}
Qisheng Wang and Zhicheng Zhang.
\newblock Fast quantum algorithms for trace distance estimation.
\newblock {\em IEEE Transactions on Information Theory}, 70(4):2720--2733, 2024.
\newblock \href {https://doi.org/10.1109/TIT.2023.3321121} {\path{doi:10.1109/TIT.2023.3321121}}.

\bibitem{WZ23b}
Qisheng Wang and Zhicheng Zhang.
\newblock Quantum lower bounds by sample-to-query lifting.
\newblock {\em SIAM Journal on Computing}, 54(5):1294--1334, 2025.
\newblock \href {https://doi.org/10.1137/24M1638616} {\path{doi:10.1137/24M1638616}}.

\bibitem{LC19}
Guang~Hao Low and Isaac~L. Chuang.
\newblock Hamiltonian simulation by qubitization.
\newblock {\em Quantum}, 3:163, 2019.
\newblock \href {https://doi.org/10.22331/q-2019-07-12-163} {\path{doi:10.22331/q-2019-07-12-163}}.

\bibitem{ARU14}
Andris Ambainis, Ansis Rosmanis, and Dominique Unruh.
\newblock Quantum attacks on classical proof systems: the hardness of quantum rewinding.
\newblock In {\em Proceedings of the 55th IEEE Annual Symposium on Foundations of Computer Science}, pages 474--483, 2014.
\newblock \href {https://doi.org/10.1109/FOCS.2014.57} {\path{doi:10.1109/FOCS.2014.57}}.

\bibitem{WZ25}
Qisheng Wang and Zhicheng Zhang.
\newblock Sample-optimal quantum estimators for pure-state trace distance and fidelity via samplizer.
\newblock ArXiv preprints, 2024.
\newblock \href {https://arxiv.org/abs/2410.21201} {\path{arXiv:2410.21201}}.

\bibitem{AJL09}
Dorit Aharonov, Vaughan Jones, and Zeph Landau.
\newblock A polynomial quantum algorithm for approximating the {Jones} polynomial.
\newblock {\em Algorithmica}, 55(3):395--421, 2009.
\newblock \href {https://doi.org/10.1007/s00453-008-9168-0} {\path{doi:10.1007/s00453-008-9168-0}}.

\bibitem{OW21}
Ryan O'Donnell and John Wright.
\newblock Quantum spectrum testing.
\newblock {\em Communications in Mathematical Physics}, 387(1):1--75, 2021.
\newblock \href {https://doi.org/10.1007/s00220-021-04180-1} {\path{doi:10.1007/s00220-021-04180-1}}.

\bibitem{OW17}
Ryan O'Donnell and John Wright.
\newblock Efficient quantum tomography {II}.
\newblock In {\em Proceedings of the 49th Annual ACM Symposium on Theory of Computing}, pages 962--974, 2017.
\newblock \href {https://doi.org/10.1145/3055399.3055454} {\path{doi:10.1145/3055399.3055454}}.

\bibitem{GL20}
Andr{\'a}s Gily{\'e}n and Tongyang Li.
\newblock Distributional property testing in a quantum world.
\newblock In {\em Proceedings of the 11th Innovations in Theoretical Computer Science Conference}, pages 25:1--25:19, 2020.
\newblock \href {https://doi.org/10.4230/LIPIcs.ITCS.2020.25} {\path{doi:10.4230/LIPIcs.ITCS.2020.25}}.

\bibitem{GHS21}
Tom Gur, Min-Hsiu Hsieh, and Sathyawageeswar Subramanian.
\newblock Sublinear quantum algorithms for estimating von {Neumann} entropy.
\newblock ArXiv preprints, 2021.
\newblock \href {https://arxiv.org/abs/2111.11139} {\path{arXiv:2111.11139}}.

\bibitem{SH21}
Sathyawageeswar Subramanian and Min-Hsiu Hsieh.
\newblock Quantum algorithm for estimating $\alpha$-{Renyi} entropies of quantum states.
\newblock {\em Physical Review A}, 104(2):022428, 2021.
\newblock \href {https://doi.org/10.1103/PhysRevA.104.022428} {\path{doi:10.1103/PhysRevA.104.022428}}.

\bibitem{LWZ22}
Xinzhao Wang, Shengyu Zhang, and Tongyang Li.
\newblock A quantum algorithm framework for discrete probability distributions with applications to {R\'{e}nyi} entropy estimation.
\newblock {\em IEEE Transactions on Information Theory}, 70(5):3399--3426, 2024.
\newblock \href {https://doi.org/10.1109/TIT.2024.3382037} {\path{doi:10.1109/TIT.2024.3382037}}.

\bibitem{WGL+22}
Qisheng Wang, Ji~Guan, Junyi Liu, Zhicheng Zhang, and Mingsheng Ying.
\newblock New quantum algorithms for computing quantum entropies and distances.
\newblock {\em IEEE Transactions on Information Theory}, 70(8):5653--5680, 2024.
\newblock \href {https://doi.org/10.1109/TIT.2024.3399014} {\path{doi:10.1109/TIT.2024.3399014}}.

\bibitem{GH20}
Alexandru Gheorghiu and Matty~J. Hoban.
\newblock Estimating the entropy of shallow circuit outputs is hard.
\newblock ArXiv preprints, 2020.
\newblock \href {https://arxiv.org/abs/2002.12814} {\path{arXiv:2002.12814}}.

\bibitem{KDS+20}
Eugenia-Maria Kontopoulou, Gregory-Paul Dexter, Wojciech Szpankowski, Ananth Grama, and Petros Drineas.
\newblock Randomized linear algebra approaches to estimate the von {Neumann} entropy of density matrices.
\newblock {\em IEEE Transactions on Information Theory}, 66(8):5003--5021, 2020.
\newblock \href {https://doi.org/10.1109/TIT.2020.2971991} {\path{doi:10.1109/TIT.2020.2971991}}.

\bibitem{LW19}
Tongyang Li and Xiaodi Wu.
\newblock Quantum query complexity of entropy estimation.
\newblock {\em IEEE Transactions on Information Theory}, 65(5):2899--2921, 2019.
\newblock \href {https://doi.org/10.1109/TIT.2018.2883306} {\path{doi:10.1109/TIT.2018.2883306}}.

\bibitem{Hay24}
Masahito Hayashi.
\newblock Measuring quantum relative entropy with finite-size effect.
\newblock {\em Quantum}, 9:1725, 2025.
\newblock \href {https://doi.org/10.22331/q-2025-05-05-1725} {\path{doi:10.22331/q-2025-05-05-1725}}.

\bibitem{BCH06}
Dave Bacon, Isaac~L. Chuang, and Aram~W. Harrow.
\newblock Efficient quantum circuits for {Schur} and {Clebsch-Gordan} transforms.
\newblock {\em Physical Review Letters}, 97(17):170502, 2006.
\newblock \href {https://doi.org/10.1103/PhysRevLett.97.170502} {\path{doi:10.1103/PhysRevLett.97.170502}}.

\bibitem{BCH07}
Dave Bacon, Isaac~L. Chuang, and Aram~W. Harrow.
\newblock The quantum {Schur} and {Clebsch-Gordan} transforms: {I.} efficient qudit circuits.
\newblock In {\em Proceedings of the 18th Annual ACM-SIAM Symposium on Discrete Algorithms}, pages 1235--1244, 2007.

\bibitem{KS18}
William~M. Kirby and Frederick~W. Strauch.
\newblock A practical quantum algorithm for the {Schur} transform.
\newblock {\em Quantum Information and Computation}, 18(9--10):721--742, 2018.
\newblock \href {https://doi.org/10.26421/QIC18.9-10-1} {\path{doi:10.26421/QIC18.9-10-1}}.

\bibitem{Kro19}
Hari Krovi.
\newblock An efficient high dimensional quantum {Schur} transform.
\newblock {\em Quantum}, 3:122, 2019.
\newblock \href {https://doi.org/10.22331/q-2019-02-14-122} {\path{doi:10.22331/q-2019-02-14-122}}.

\bibitem{Ngu23}
Quynh~T. Nguyen.
\newblock The mixed {Schur} transform: efficient quantum circuit and applications.
\newblock ArXiv preprints, 2023.
\newblock \href {https://arxiv.org/abs/2310.01613} {\path{arXiv:2310.01613}}.

\bibitem{GBO23}
Dmitry Grinko, Adam Burchardt, and Maris Ozols.
\newblock {Gelfand-Tsetlin} basis for partially transposed permutations, with applications to quantum information.
\newblock ArXiv preprints, 2023.
\newblock \href {https://arxiv.org/abs/2310.02252} {\path{arXiv:2310.02252}}.

\bibitem{LWWZ24}
Nana Liu, Qisheng Wang, Mark~M. Wilde, and Zhicheng Zhang.
\newblock Quantum algorithms for matrix geometric means.
\newblock {\em npj Quantum Information}, 11:101, 2025.
\newblock \href {https://doi.org/10.1038/s41534-025-00973-7} {\path{doi:10.1038/s41534-025-00973-7}}.

\bibitem{LW25}
Yupan Liu and Qisheng Wang.
\newblock On estimating the trace of quantum state powers.
\newblock In {\em Proceedings of the 2025 Annual ACM-SIAM Symposium on Discrete Algorithms}, pages 947--993, 2025.
\newblock \href {https://doi.org/10.1137/1.9781611978322.28} {\path{doi:10.1137/1.9781611978322.28}}.

\bibitem{LW25b}
Yupan Liu and Qisheng Wang.
\newblock On estimating the quantum $\ell_\alpha$ distance.
\newblock In {\em Proceedings of the 33rd Annual European Symposium on Algorithms}, pages 106:1--106:19, 2025.
\newblock \href {https://doi.org/10.4230/LIPIcs.ESA.2025.106} {\path{doi:10.4230/LIPIcs.ESA.2025.106}}.

\bibitem{NRTM25}
Ryotaro Niwa, Zane~Marius Rossi, Philip Taranto, and Mio Murao.
\newblock Singular value transformation for unknown quantum channels.
\newblock ArXiv preprints, 2025.
\newblock \href {https://arxiv.org/abs/2506.24112} {\path{arXiv:2506.24112}}.

\bibitem{Wan24}
Qisheng Wang.
\newblock Optimal trace distance and fidelity estimations for pure quantum states.
\newblock {\em IEEE Transactions on Information Theory}, 70(12):8791--8805, 2024.
\newblock \href {https://doi.org/10.1109/TIT.2024.3447915} {\path{doi:10.1109/TIT.2024.3447915}}.

\bibitem{FW25}
Wang Fang and Qisheng Wang.
\newblock Optimal quantum algorithm for estimating fidelity to a pure state.
\newblock In {\em Proceedings of the 33rd Annual European Symposium on Algorithms}, pages 4:1--4:12, 2025.
\newblock \href {https://doi.org/10.4230/LIPIcs.ESA.2025.4} {\path{doi:10.4230/LIPIcs.ESA.2025.4}}.

\bibitem{CWZ24}
Kean Chen, Qisheng Wang, and Zhicheng Zhang.
\newblock Local test for unitarily invariant properties of bipartite quantum states.
\newblock ArXiv preprints, 2024.
\newblock \href {https://arxiv.org/abs/2404.04599} {\path{arXiv:2404.04599}}.

\bibitem{SY23}
Adrian She and Henry Yuen.
\newblock Unitary property testing lower bounds by polynomials.
\newblock In {\em Proceedings of the 14th Innovations in Theoretical Computer Science Conference}, pages 96:1--96:17, 2023.
\newblock \href {https://doi.org/10.4230/LIPIcs.ITCS.2023.96} {\path{doi:10.4230/LIPIcs.ITCS.2023.96}}.

\bibitem{Weg24}
Jordi Weggemans.
\newblock Lower bounds for unitary property testing with proofs and advice.
\newblock {\em Quantum}, 9:1717, 2025.
\newblock \href {https://doi.org/10.22331/q-2025-04-18-1717} {\path{doi:10.22331/q-2025-04-18-1717}}.

\bibitem{WZC+23}
Qisheng Wang, Zhicheng Zhang, Kean Chen, Ji~Guan, Wang Fang, Junyi Liu, and Mingsheng Ying.
\newblock Quantum algorithm for fidelity estimation.
\newblock {\em IEEE Transactions on Information Theory}, 69(1):273--282, 2023.
\newblock \href {https://doi.org/10.1109/TIT.2022.3203985} {\path{doi:10.1109/TIT.2022.3203985}}.

\bibitem{HHL09}
Aram~W. Harrow, Avinatan Hassidim, and Seth Lloyd.
\newblock Quantum algorithm for linear systems of equations.
\newblock {\em Physical Review Letters}, 103(15):150502, 2009.
\newblock \href {https://doi.org/10.1103/PhysRevLett.103.150502} {\path{doi:10.1103/PhysRevLett.103.150502}}.

\bibitem{Wat02}
John Watrous.
\newblock Limits on the power of quantum statistical zero-knowledge.
\newblock In {\em Proceedings of the 43rd Annual IEEE Symposium on Foundations of Computer Science}, pages 459--468, 2002.
\newblock \href {https://doi.org/10.1109/SFCS.2002.1181970} {\path{doi:10.1109/SFCS.2002.1181970}}.

\bibitem{BKL+19}
Fernando G. S.~L. Brand{\~a}o, Amir Kalev, Tongyang Li, Cedric Yen-Yu Lin, Krysta~M. Svore, and Xiaodi Wu.
\newblock Quantum {SDP} solvers: large speed-ups, optimality, and applications to quantum learning.
\newblock In {\em Proceedings of the 46th International Colloquium on Automata, Languages, and Programming}, pages 27:1--27:14, 2019.
\newblock \href {https://doi.org/10.4230/LIPIcs.ICALP.2019.27} {\path{doi:10.4230/LIPIcs.ICALP.2019.27}}.

\bibitem{HL11}
Aram~W. Harrow and Debbie~W. Leung.
\newblock A communication-efficient nonlocal measurement with application to communication complexity and bipartite gate capacities.
\newblock {\em IEEE Transactions on Information Theory}, 57(8):5504--5508, 2011.
\newblock \href {https://doi.org/10.1109/TIT.2011.2158468} {\path{doi:10.1109/TIT.2011.2158468}}.

\bibitem{Qia24}
Luowen Qian.
\newblock Unconditionally secure quantum commitments with preprocessing.
\newblock In {\em Advances in Cryptology – CRYPTO 2024}, pages 38--58, 2024.
\newblock \href {https://doi.org/10.1007/978-3-031-68394-7_2} {\path{doi:10.1007/978-3-031-68394-7_2}}.

\bibitem{SGSV24}
Eddie Schoute, Dmitry Grinko, Yi{\u{g}}it Suba{\c{s}}{\i}, and Tyler Volkoff.
\newblock Quantum programmable reflections.
\newblock ArXiv preprints, 2024.
\newblock \href {https://arxiv.org/abs/2411.03648} {\path{arXiv:2411.03648}}.

\bibitem{Hoe63}
Wassily Hoeffding.
\newblock Probability inequalities for sums of bounded random variables.
\newblock {\em Journal of the American Statistical Association}, 58(301):13--30, 1963.
\newblock \href {https://doi.org/10.1080/01621459.1963.10500830} {\path{doi:10.1080/01621459.1963.10500830}}.

\bibitem{Fel68}
William Feller.
\newblock {\em An Introduction to Probability Theory and Its Applications, Volume 1}.
\newblock John Wiley \& Sons, 1968.

\bibitem{BS93}
Christian Beck and Friedrich Sch{\"{o}}gl.
\newblock {\em Thermodynamics of Chaotic Systems: An Introduction}.
\newblock Cambridge University Press, 1993.
\newblock \href {https://doi.org/10.1017/CBO9780511524585} {\path{doi:10.1017/CBO9780511524585}}.

\bibitem{Hel67}
Carl~W. Helstrom.
\newblock Detection theory and quantum mechanics.
\newblock {\em Information and Control}, 10(3):254--291, 1967.
\newblock \href {https://doi.org/10.1016/S0019-9958(67)90302-6} {\path{doi:10.1016/S0019-9958(67)90302-6}}.

\bibitem{Hol73}
Alexander~S. Holevo.
\newblock Statistical decision theory for quantum systems.
\newblock {\em Journal of Multivariate Analysis}, 3(4):337--394, 1973.
\newblock \href {https://doi.org/10.1016/0047-259X(73)90028-6} {\path{doi:10.1016/0047-259X(73)90028-6}}.

\bibitem{BBD+97}
Adriano Barenco, Andr{\'{e}} Berthiaume, David Deutsch, Artur Ekert, Richard Jozsa, and Chiara Macchiavello.
\newblock Stabilization of quantum computations by symmetrization.
\newblock {\em SIAM Journal on Computing}, 26(5):1541--1557, 1997.
\newblock \href {https://doi.org/10.1137/S0097539796302452} {\path{doi:10.1137/S0097539796302452}}.

\bibitem{BCWdW01}
Harry Buhrman, Richard Cleve, John Watrous, and Ronald de~Wolf.
\newblock Quantum fingerprinting.
\newblock {\em Physical Review Letters}, 87(16):167902, 2001.
\newblock \href {https://doi.org/10.1103/PhysRevLett.87.167902} {\path{doi:10.1103/PhysRevLett.87.167902}}.

\bibitem{KMY09}
Hirotada Kobayashi, Keiji Matsumoto, and Tomoyuki Yamakami.
\newblock Quantum {Merlin-Arthur} proof systems: are multiple {Merlins} more helpful to {Arthur}?
\newblock {\em Chicago Journal of Theoretical Computer Science}, 2009:3, 2009.
\newblock \href {https://doi.org/10.4086/cjtcs.2009.003} {\path{doi:10.4086/cjtcs.2009.003}}.

\end{thebibliography}

\appendix

\section{Estimating \texorpdfstring{$2$}{2}-R\'enyi entropy} \label{sec:2-renyi}

For completeness, we give a simple quantum algorithm for estimating $2$-R\'enyi entropy $S_2\rbra{\rho} = -\ln\rbra{\tr\rbra{\rho^2}}$ based on the well-known SWAP test \cite{BBD+97,BCWdW01}.
We note that $P_2\rbra{\rho} = \tr\rbra{\rho^2}$ is known as the purity of $\rho$, which can be estimated by the quantum circuit shown in \cref{fig:purity} where the measurement outcome is $0$ with probability $\frac{1+P_2\rbra{\rho}}{2}$ (see \cite[Proposition 9]{KMY09}). 
Thus, we can compute an estimate $\widetilde P$ of $P_2\rbra{\rho}$ (with probability $\geq 3/4$) within additive error $\epsilon$ using $O\rbra{1/\epsilon^2}$ samples of $\rho$. 

\begin{figure} [!htp]
\centering
\begin{quantikz} [row sep = {20pt, between origins}]
    \lstick{$\ket{0}$} & \gate{H} & \ctrl{2} & \gate{H} & \meter{} \\
    \lstick{$\rho$} \setwiretype{b} & \qw & \swap{1} & \qw & \\
    \lstick{$\rho$} \setwiretype{b} & \qw & \targX{} & \qw & \\
\end{quantikz}
\caption{Quantum circuit for estimating purity.}
\label{fig:purity}
\end{figure}
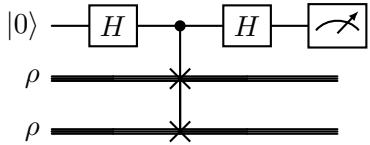

Suppose $\rho$ is of rank $r$. 
From \cref{fact:renyi}, we know that $r^{-1} \leq P_2\rbra{\rho} \leq 1$. 
By taking $\epsilon = \varepsilon/r$, we can obtain $\widetilde P$ such that $\rbra{1-\varepsilon} P_2\rbra{\rho} \leq \widetilde P \leq \rbra{1+\varepsilon} P_2\rbra{\rho}$.
Therefore, $\widetilde S = - \ln \rbra{\widetilde P}$ is an estimate of $S_2\rbra{\rho}$ within additive error $\Theta\rbra{\varepsilon}$. 
Finally, by choosing the median of $O\rbra{\log\rbra{\delta}}$ repetitions of the above procedure, we can amplify the success probability to $\geq 1-\delta$. 
We summarize the complexity of this simple algorithm as follows. 
\begin{lemma}
    There is a quantum algorithm with sample access to $N$-dimensional quantum state $\rho$ of rank $r$ that, with probability $\geq 1 - \delta$, estimates the $2$-R\'enyi entropy $S_2\rbra{\rho}$ within additive error $\varepsilon$ with sample complexity $M$ and time complexity $O\rbra{M\log{N}}$, where $M = O\rbra{r^2/\varepsilon^2 \cdot \log\rbra{1/\delta}}$.
\end{lemma}

\end{document}